\documentclass{llncs}
\usepackage{graphicx}
\usepackage{graphics}
\usepackage{epsfig}
\usepackage{epstopdf}
\usepackage{amsmath,amssymb}
\usepackage{latexsym}
\usepackage{algorithm}
\usepackage{algorithmicx}
\usepackage{ifpdf}

\newtheorem{observation}{Observation}

\newcommand{\deleted}[1]{}
\providecommand{\norm}[1]{\lVert#1\rVert}
\providecommand{\abs}[1]{\lvert#1\rvert}

\begin{document}

\pagestyle{plain}
\pagenumbering{arabic}
\setcounter{page}{1}

\title{On Clustering Induced Voronoi Diagrams\thanks{A preliminary version of this paper appeared in the Proceedings of the 54th Annual IEEE Symposium on Foundations of Computer Science (FOCS 2013). The research of the first author was supported in part by NSF under Grant CCF-1217906, and the research of the last three authors was supported in part by NSF under grant IIS-1115220.}} 
\author{Danny Z. Chen\inst{1} \hspace{0.25in} Ziyun Huang\inst{2}   \hspace{0.25in} Yangwei Liu\inst{2} \hspace{0.25in} Jinhui Xu\inst{2}}
\institute{
Department of Computer Science and Engineering\\
University of Notre Dame\\
\email{dchen@cse.nd.edu}\\
\and
 Department of Computer Science and Engineering\\
State University of New York at Buffalo\\
 \email{\tt \{ziyunhua, yangweil, jinhui\}@buffalo.edu}\\
}

\maketitle

\begin{abstract}
In this paper, we study a generalization of the classical Voronoi diagram, called {\em clustering induced Voronoi diagram (CIVD)}. Different from the traditional model,
CIVD takes as its sites the power set $U$ of an input set $P$ of objects. For each subset $C$ of $P$, CIVD uses an {\it influence function} $F(C,q)$ to measure the 
total (or joint) {\it influence} of all objects in $C$ on an arbitrary point $q$ in the space
$\mathbb{R}^d$, and determines the influence-based Voronoi cell in $\mathbb{R}^d$ for $C$.
This generalized model offers a number of new features ({\em e.g.,} simultaneous clustering and space partition) to Voronoi diagram 
which are useful in various new applications.
We investigate the general conditions for the influence function which ensure the existence of a small-size ({\em e.g.}, nearly linear)
approximate CIVD for a set $P$ of $n$ points in $\mathbb{R}^d$ for some fixed $d$. To construct
CIVD, we first present a standalone new technique, called {\em approximate influence (AI) decomposition}, for the
general CIVD problem.  With only $O(n\log n)$ time, the AI decomposition partitions the space $\mathbb{R}^{d}$ into a nearly linear number of cells so that all points in each cell receive their approximate maximum influence from the same (possibly unknown) site ({\em i.e.,} a subset of $P$).
Based on this technique, we develop assignment algorithms to determine a proper site for each cell in
the decomposition and form various $(1-\epsilon)$-approximate CIVDs for some small fixed
$\epsilon>0$. Particularly, we consider two representative CIVD problems, vector CIVD and density-based CIVD, and show that both of them admit fast assignment algorithms;
consequently, 
their $(1-\epsilon)$-approximate CIVDs can be built in $O(n \log^{\max\{3,d+1\}}n)$ and $O(n \log^{2} n)$ time, respectively. \\
 
\end{abstract}

\newpage

\section{Introduction}
Voronoi diagram is a fundamental geometric structure with numerous applications in many different areas \cite{Aur87,Aur91,Oka00}. 
The ordinary Voronoi diagram is a partition of the space $\mathbb{R}^d$ into a set of cells
induced by a set $P$ of points (or other objects) called sites, where each cell $c_{i}$ of the diagram is the union of all points in $\mathbb{R}^d$ which 
have a closer (or farther) distance to a site $p_{i}\in P$ than to any other sites. There are many variants of Voronoi diagram, depending on the types of objects in $P$,
the distance metrics,
the dimensionality of $\mathbb{R}^d$,
the order of Voronoi diagram, etc.
In some sense, the cells in a Voronoi diagram are formed by competitions among all sites in $\mathbb{R}^d$, such that the winner site for any point $q$ in $\mathbb{R}^d$ is the one having a larger ``influence" on $q$ defined by its distance to $q$.

A common feature shared by most known Voronoi diagrams is that the influence from every site object is independent of one another
and does not combine together. However, it is quite often in real world applications that the influences
 from multiple sources can be ``added'' together to form a {\it joint influence}. For example, in physics, a particle $q$ may receive forces from a number of other particles and the set of such
forces jointly determines the motion of $q$. This phenomenon also arises in other areas, such as social networks where the set of actors ({\em i.e.}, nodes) in a community may have joint influence
on a potential new actor ({\em e.g.,} a twitter account with a large number of followers may have a better chance to attract more followers). 
In such scenarios, it is desirable to identify the subset of objects which has the largest joint influence on one or more particular objects.

To develop a geometric model for joint influence, in this paper, we generalize the concept of
Voronoi diagram to {\em Clustering Induced Voronoi Diagram} (CIVD). In CIVD, we consider a set
$P$ of $n$ points (or other types of objects) and a non-negative influence function $F$ which measures the joint influence $F(C,q)$ from each subset $C$ of $P$ to any point $q$ in $\mathbb{R}^d$.
The Voronoi cell of $C$ is the union of all points in $\mathbb{R}^d$ which receive a larger influence from $C$ than from any other subset $C' \subseteq P$. This means that CIVD considers all subsets in the power set $U=2^{P}$ of $P$
as its sites (called {\em cluster sites}), and partitions $\mathbb{R}^d$ according to their influences.  For some interesting influence functions, it is possible that only a small number of subsets in $U$  have non-empty Voronoi cells. Thus the complexity of a CIVD is not necessarily exponential as in the worst case.

CIVD thus defined considerably generalizes the concept of Voronoi diagram. To our best knowledge,
there is no previous work on the general CIVD problem. It obviously extends the
ordinary Voronoi diagrams \cite{Aur91}, where each site is a one-point cluster. (Note that  the ordinary Voronoi diagrams can be viewed as special CIVDs equipped with proper influence functions.) Some Voronoi diagrams \cite{Oka00,Pap04}
allow a site to contain multiple points, but the distance 
functions used are often defined by the closest (or farthest) point in such a site, not by a collective effect
of all points of the site. The $k$-th order Voronoi diagram \cite{Oka00}, where each cell is the union of points in $\mathbb{R}^d$ sharing the 
same $k$ nearest neighbors in $P$, may be viewed as having clusters of points as sites,
and the ``distance" functions are defined on all points of 
each site;  but all cluster ``sites" of a $k$-th order Voronoi diagram have the same size $k$ and its ``distance" function 
is quite different from the influence function in our CIVD problem.
Some two-point site Voronoi diagrams were also studied \cite{Bar02,Bar11,Dic09,Dic08,Han10,Hod05,Vya10}, in which
each site has exactly two points and the distance functions are defined by certain ``combined"
features of point pairs. Obviously, such Voronoi diagrams are different from CIVD.

CIVD enables us to capture not only the spatial proximity of points, but more importantly their aggregation in the space. For example, a cluster site $C$ having a non-empty Voronoi cell may imply that 
the points in $C$ form a local cluster inside that cell. This provides an interesting connection between
clustering and space partition and a potential to solve clustering and space partition problems simultaneously. Such new insights could be quite useful for applications in data mining and social networks. For instance, 
in social networks, clustering can be used to determine communities in some feature space,
and space partition may allow to identify the nearest (or best-fit) community for any new actor.  Furthermore, since each point in $P$ may appear in multiple cluster sites with non-empty Voronoi cells, this could potentially
help find all communities in a social network without having to apply the relatively expensive overlapping clustering techniques \cite{And12,Ban05,Bon11,Cle04} or to explicitly generate multiple views of the network \cite{Che12,Chr08,Gre09,Liu12}. 
 
Of course, CIVD in general can have exponentially many cells, and an interesting question is
what meaningful CIVD problems have a small number (say, polynomial) of cells.
Thus, generalizing Voronoi diagrams in this way brings about a number of new challenges:
(I) How to efficiently deal with the exponential number of potential cluster sites; (II) how to
identify those non-empty Voronoi cells so that the construction time of CIVD is proportional only
to the actual size of CIVD; (III) how to partition the space and efficiently determine the cluster site for each non-empty Voronoi cell in CIVD. 

We consider in this paper the CIVD model for a set $P$ of $n$ points in $\mathbb{R}^{d}$ for some fixed $d$,
aiming to resolve the above difficult issues.
We first investigate the general and sufficient conditions which allow the influence function to
yield only a small number of non-empty approximate Voronoi cells. Our focus is thus mainly on the family of influence functions satisfying these conditions. We then present a standalone new technique,
called {\em approximate influence decomposition} (or AI decomposition), for general CIVD problems.
In $O(n \log n)$ time, this technique partitions the space $\mathbb{R}^{d}$ into a nearly linear number ({\em i.e.,} $O(n\log n)$) of cells so that for each such cell $c$,
there exists a (possibly unknown) subset $C \subseteq P$ whose influence to any point $q \in c$
is within a $(1-\epsilon)$-approximation of the maximum influence that $q$ can receive from any subset of $P$, where $\epsilon >0$ is a fixed small constant. 
For this purpose, we develop a new data structure called {\it box-clustering tree}, based on
an extended quad-tree decomposition
and guided by a {\em distance-tree} built from the well-separated pair decomposition \cite{CK95}.
In some sense, our AI decomposition may be viewed as a generalization of the well-separated pair decomposition.

The AI decomposition partially overcomes challenges (II) and (III) above. However, we still need
to assign a proper cluster site (selected from the power set $U$ of $P$) to each resulted non-empty Voronoi cell.
To illustrate how to resolve this issue, we consider some important CIVDs and make use of both
the AI decomposition and the specific properties of the influence functions of these problems
to build approximate CIVDs. Particularly, we study two representative CIVD problems. The first
problem is {\it vector CIVD} in which the influence between any two points $p$ and $q$ is
defined by a force-like vector ({\em e.g.}, gravity force)  and the joint influence is  the vector sum.
Clearly, this problem can be used to construct Voronoi diagrams in some force-induced fields.
The second problem is {\it density-based CIVD} in which the influence from a cluster $C$ to a point $q$ is the density of the smallest enclosing ball of $C$ centered at $q$. This problem enables us to generate all density-based clusters as well as their approximate Voronoi cells. Since density-based clustering is widely used in many areas such as data mining, computer vision,
pattern recognition, and social networks \cite{Cao06,Che05,Che07,Kri05,San98}, we expect that the density-based CIVD is also applicable in these areas. 
For both these problems, we present efficient assignment algorithms that determine a proper cluster site for each cell generated by the AI decomposition in polylogarithmic time.
Thus, $(1-\epsilon)$-approximate CIVDs for both problems can be constructed in $O(n \log^{\max\{3,d+1\}}n)$ and $O(n \log^{2}n)$ time, respectively.

Since the conditions and the AI decomposition are all quite general and do not require to know the exact form of the influence function, we expect that our techniques will be applicable to many other CIVD problems.

It is worth pointing out that although significant differences exist, several problems/techniques can be viewed as related to CIVD. The first one is the {\em approximate Voronoi diagram or nearest neighbor search} problem \cite{Ary02,Ary02a,Ary98,Har01,Har12,Ind98}, which shares with our approximate CIVD the same strategy of using regular shapes to approximate the Voronoi cells. However,
since their sites are all single-point, such problems are quite different from our approximate CIVD problem.  The second one is the {\em Fast Multipole Method (FMM)} for the N-body problem \cite{Gre88,Gre90,Gre94}, which shares with the Vector CIVD a similar idea of modeling joint force by influence functions.  The difference is that  FMM mainly relies on simple functions ({\em i.e.,} kernels) to reduce the computational complexity, while Vector CIVD uses perturbation and properties of the influence function to achieve faster computation.

The rest of this paper is organized as follows. Section \ref{sec-ov} overviews the main ideas and difficulties in designing a small-size approximate CIVD. Section \ref{sec-inf} discusses the needed general properties of the influence function. 
In Section \ref{sec-AI}, we present our approximate influence decomposition technique.  In Sections \ref{sec-ex1} and \ref{sec-ex2}, we show how the AI decomposition technique can be applied to construct approximate CIVDs for the two representative problems.

\section{Overview of Approximate CIVD}
\label{sec-ov}

In this section, we give an overview of the main ideas for and difficulties in computing approximate CIVDs. In the subsequent sections, we will show how to overcome  each of the major obstacles.

Let $P=\{p_{1}, p_{2}, \ldots,p_{n}\}$ be a set of $n$ points in $\mathbb{R}^{d}$ for some fixed $d$, $C$ be a subset of $P$, and $q$ be an 
arbitrary point in  $\mathbb{R}^{d}$ (called a {\em query point}).  The influence from $C$ to $q$ is 
a function $F(C,q)$ of the vectors from every point $p\in C$ to $q$ (or from $q$ to $p$).  
Among all possible cluster sites of $P$, let $C_m(P,q) \subseteq P$ denote the cluster site which has the maximum
influence, $F_{max}(q)$, on $q$, called the {\em maximum influence site} of $q$.
Below we define the $(1-\epsilon)$-approximate CIVD induced by the influence function $F$.

\begin{definition}\label{def-app}
Let $\mathcal{R}=\{c_{1},c_{2},\ldots, c_{k}\}$ be a partition of the space $\mathbb{R}^{d}$, and $\mathcal{C}=\{C_{1},C_{2},\ldots,C_{k}\}$ be a set
(possibly a multiset) of  cluster sites  of ${P}$. The set of pairs $\{(c_{1},C_{1}),(c_{2},C_{2}),\ldots,(c_{k},C_{k})\}$ is a $(1-\epsilon)$-approximate CIVD
with respect to the influence function $F$ if for each $c_{i} \in \mathcal{R}$,  $F(C_{i},q) \geq (1 - \epsilon)F_{max}(q)$ for any point $q\in c_{i}$,
where $\epsilon>0$ is a small constant. Each $c_{i}$ is an {\em approximate Voronoi cell}, and $C_{i}$ is the {\em approximate maximum influence site} of $c_i$.
\end{definition}

Based on the above definition, for computing an approximate CIVD,  there are two major tasks: (1) partition $\mathbb{R}^{d}$ into a set $\mathcal{R}=\{c_{1},c_{2},
\ldots, c_{k}\}$ of cells, and (2) determine $C_{i}\subseteq P$ for each $c_{i}$.  We call task (2) the {\em assignment problem}, which finds an approximate
maximum influence site $C_{i}$ in the power set $U$ of $P$ for each cell $c_{i}$ of $\mathcal{R}$. Since the choice of $C_{i}$ often depends on the 
properties of the  influence function, we need to develop a specific assignment algorithm for each CIVD problem.
In Sections \ref{sec-ex1} and \ref{sec-ex2}, we present efficient assignment algorithms for two CIVD problems.

We call task (1) the {\em space partition problem}.  For this problem, we develop a standalone technique, called {\em Approximate Influence (AI) Decomposition},
for general CIVDs. The size of a CIVD (or an approximate CIVD) in general can be exponential.
Thus, we study some key conditions of 
the influence function that yield a small-size space partition.
In Section \ref{sec-inf}, we investigate the general and sufficient conditions that ensure the existence of a small-size approximate CIVD. The AI
decomposition makes use of only these general conditions and need not know the exact form of the influence function.

Roughly speaking, the general conditions ensure to achieve simultaneously two objectives on the resulting cells of the space partition: (a) Cells that are far
away from the input points of $P$ should be of as ``large'' diameters as possible (where ``far away'' means that the diameter of a cell is small comparing to the distance from the cell to the nearest input point),
and (b) cells that are close to the input points should not be too small (in terms of their
diameters). Each objective helps reduce the number of cells in the space partition from a different perspective. 
To understand this better, consider a query point $q$ and its approximate maximum influence site $C$.
For objective (a), we expect that all points in a sufficiently large neighborhood of $q$ share $C$ (together
with $q$) as their common approximate maximum influence site. Particularly, we assume that there is some constant   
$\lambda_{1}(\epsilon)$ (depending on $\epsilon$) such that a region containing $q$ and with a diameter of
roughly $r\lambda_{1}(\epsilon)$ can be a cell of the partition, where $r$ is the distance from $q$ to the closest point in $P$.
For objective (b), we assume that when $q$ is close enough to a subset $C$ of $P$, there is a polynomial function $\mathcal{P}(\cdot)$ such that a region containing $q$ and with a diameter of
$\lambda_{2}(\epsilon)r'/\mathcal{P}(n)$ can be a cell of the partition, where $r'$ is the distance from $q$ to the closest point in $P\setminus C$ and $\lambda_{2}(\epsilon)$ is a constant depending on $\epsilon$.

Corresponding to the two objectives above, the AI decomposition presented in Section \ref{sec-AI} partitions $\mathbb{R}^{d}$ 
into two types of cells: type-1 cells and type-2 cells. 
Type-1 cells are those close to some input points ({\em i.e.,} corresponding to objective (b)), and type-2 cells are those far away from the input points ({\em i.e.,} corresponding to objective (a)).
The AI decomposition has the following properties.

\begin{enumerate}
\item The space $\mathbb{R}^{d}$ is partitioned (in $O(n \log n)$ time) into a total of $O(n \log n)$ type-1 and type-2 cells.

\item A type-1 cell $c$ is either a box region ({\em i.e.,} an axis-aligned hypercube) or the difference of two box regions, and is associated with a known approximate maximum influence site $C$.

\item A type-2 cell $c$ is a box region with a diameter of $D(c) \leq 2r\lambda_{1}(\epsilon)/3$,
where $r$ is the minimum distance between $c$ and any point in $P$. All points in a  type-2 cell $c$ share a (not yet identified) cluster site $C\subseteq P$ as their common approximate maximum influence site.
\end{enumerate}

To ensure the above properties, we need to overcome several difficulties. First, we need to efficiently maintain the (approximate) distances between the input points of $P$ and all potential cells, in order
to distinguish the cell types.  To resolve this difficulty, we make use of the well-separated pair decomposition \cite{CK95} to build a new data structure called {\em distance-tree} and use it to
approximate the distances between the cells and the input points. Second, we need to generate the two types of cells and make sure that each cell has a common approximate maximum influence site.
For this, we recursively construct a new data structure called {\em box-clustering tree} to partition $\mathbb{R}^{d}$ into type-1 and type-2 cells.
Third, we need to analyze the bounds for the total number of cells and the running time of the
AI decomposition, for which we prove a key packing lemma in the space $\mathbb{R}^{d}$.
We will unfold our ideas in detail for resolving each difficulty in Section \ref{sec-AI}. 

As stated above, every type-1 cell in the AI decomposition is associated with a known approximate
maximum influence site. Thus, our assignment algorithms only need to focus on determining the approximate maximum
influence sites for the type-2 cells.

\section{Influence Function}
\label{sec-inf}

In this section, we discuss the general conditions for the influence function to yield a small-size approximate CIVD. 

By the definition of CIVD, a straightforward construction algorithm would consider the exponential-size power set $U$ of $P$
and the influence to every point in the space $\mathbb{R}^{d}$.
The actual size of CIVD depends on the nature of its influence function. For a given influence function, it is possible that most of the cluster sites in $U$ have a non-empty Voronoi cell, and
hence the resulting CIVD is of exponential size. Of course, for this to happen, the influence function needs to have certain properties ({\em e.g.,} the range system defined by its iso-value
surfaces and $P$ have exponential VC dimensions). Fortunately, many influence functions in applications have good properties that induce CIVDs of much smaller sizes. Thus, it is desirable to understand how an influence function affects the size of the corresponding CIVD. For this purpose,
we investigate the general and sufficient conditions of the influence functions which allow to yield a small-size (approximate) CIVD.

Note that since an influence function can be arbitrary, we shall focus on its general properties rather than its exact form. We will make some reasonable and self-evident assumptions about  the influence function.
Also, because even a small-size CIVD may still take exponential time to construct, our objective is to
obtain a set of general conditions which ensure a fast construction of an approximate CIVD.
Ideally, we desire that the construction time be nearly linear.

Let $q$ be an arbitrary point in $\mathbb{R}^{d}$ and $C$ be a subset of $P$. The influence from $C$ to $q$ is defined as follows.

\begin{definition}
\label{def-if}
The influence from $C$ to $q$,
$q \not\in C$,
is a function $F(C,q)$ satisfying the following condition:
$F(C,q)=f(G(C,q))$, where
$G(C,q) = \{p-q \mid p \in C \}$ is the multiset of vectors defined by $C$ and $q$ and $f(\cdot)$ is a non-negative function defined over all possible multisets of vectors in $\mathbb{R}^d$. 
For convenience, $f$ is also called the influence function.
\end{definition}

In the above definition, the influence depends solely on the set of vectors pointing from $q$
to each point $p \in C$ or from each $p$ to $q$. It is possible that some CIVD problems use only the lengths of these vectors.
This implies that the influence of $C$ on $q$ remains the same under translation.

Note that $F(C,q)$ is undefined when $q \in C$.
In this paper, every point $q \in P$ is considered as a singularity.
In the rest of the paper, the case of $q$ being a singularity is  ignored.

The influence function is also desired to have good properties on scaling and rotation, as follows.

\begin{property}[Similarity Invariant]
\label{pro-3}
Let $\phi$ be a transformation of scaling or rotation about $q$, and $C$ be any set (possibly multiset) of points in $\mathbb{R}^{d}$. The ratio $F(\phi(C),q)/F(C,q)$ is uniquely determined by $\phi$.  
\end{property}

The above property implies that the maximum influence site $C_{m}(P,q)$ of $q$ remains the same under any
scaling or rotation about $q$.  This is because all subsets of $P$ change their influences on $q$ by the same factor after such a transformation. Combining this with Definition \ref{def-if}, we know that the 
 maximum influence site $C_{m}(P,q)$ of $q$ is invariant under the similarity transformation. Thus Property \ref{pro-3} is also called the {\em similarity invariant} property, and is necessary for the following locality property.

As discussed in the previous section, to ensure a small-size approximate CIVD, we expect that the cells (of the CIVD) that are far away from the input points should be ``large'' ({\em i.e.,} objective (a)) and the
cells that are close to the input points should not be too small ({\em i.e.,} objective (b)). 
This means that many spatially close points in $\mathbb{R}^{d}$ would have to share the same
cluster site $C$ as their approximate maximum influence site, which implies that the influence function must have a certain degree of locality (to achieve objective (a)).  Below we define the precise meaning of the locality property.

\begin{definition}
\label{def-per}
Let $q$ be a point and $C$ be a set (possibly multiset) of points in $\mathbb{R}^{d}$. For $C$ and $q$, a
one-to-one mapping $\psi$ from $C$ to $\psi(C)$ in $\mathbb{R}^{d}$ is called an {\em $\epsilon$-perturbation} with respect to $q$ if 
$\norm{p-\psi(p)} \le \epsilon\norm{p-q}$ for every point $p\in C$, where $0<\epsilon <1$ is the {\em error ratio} and $q$ is called the {\em witness point} of $\psi$.
\end{definition}

Intuitively speaking, from the witness point's view, an $\epsilon$-perturbation only changes
slightly the position of a point that it maps. 

\begin{definition}
\label{stable}
Let $q$ be a point and $C$ be a set (possibly multiset) of points in $\mathbb{R}^{d}$. 
For any $\gamma \in (0,1)$,
let $\delta$ be a continuous monotone function with $\delta(\gamma)<1$ and $\lim_{x \to 0} \delta (x)=0$.
An influence function $F$ is said to be $(\delta, \gamma)$-{\em stable} at $(C,q)$ if for any $\epsilon$-perturbation $C'$ of $C$ with $\epsilon \leq \gamma<1$, $(1 - \delta(\epsilon))F(C,q) \leq F(C',q) \leq (1 + \delta(\epsilon))F(C,q)$.
 \end{definition}

In the above definition, $(C,q)$ is called a $(\delta, \gamma)$-stable pair or simply a stable pair.

To define the locality property, it might be tempting to simply require that $F$ be stable at any subset $C$ and any query point $q$ in $\mathbb{R}^{d}$. However, this would be a too strong condition, as we will show 
later that some problems ({\em e.g.}, the vector CIVD problem) not satisfying this condition still have a small-size approximate CIVD. Thus, we need to use a weaker condition for the locality property.

\begin{definition}
\label{pair}
Let $C$ be a set (possibly multiset) of points in $\mathbb{R}^{d}$,  $q$ be a query point, and $F$ be the influence function. 
$(C,q)$ is  a \emph{maximal pair} of $F$ if for any  subset $C'$ of $C$, $F(C',q) \leq F(C,q)$.
\end{definition}

 From the above definition, we know that any maximum influence site and any of its corresponding query points always form a maximum pair.
Since each maximal pair could potentially correspond to a non-empty Voronoi cell and any locality requirement on the influence function has to ensure stability on all Voronoi cells, it is sufficient to define the locality based on the stability of all maximal pairs.

\begin{property}[Locality]
\label{pro-1}
The influence function $F$ is $(\delta, \gamma)$-stable at any maximal pair $(C,q)$ for some continuous monotone function $\delta$ and small constant $0<\gamma<1$.
 \end{property}

The locality property above means that a small perturbation on $P$ changes only slightly the maximum influence on a query point $q$. This implies
that we can use the perturbed points of $P$ to determine an approximate  maximum influence site for each point $q$. 
The following lemma further shows that a good approximation of the maximum influence site for $q$ is still a good approximation after an $\epsilon$-perturbation.

\begin{figure}[ht]
\vspace{-0.15in}
\centering
\includegraphics[height=2.2in]{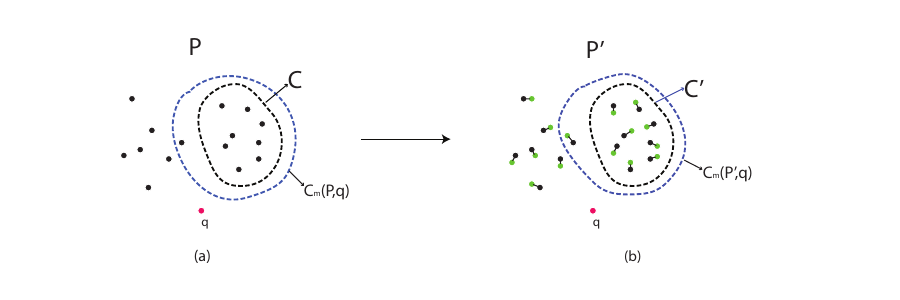}
\caption{Illustrating Lemma \ref{lem-app}:
(a) $C\subseteq P$ is nearly the same as $C_{m}(P,q)$; (b) after the perturbation, $C'$ is still almost the same as $C_{m}(P',q)$.}
\label{fig-lem1}
\end{figure}

\begin{lemma}
\label{lem-app}
Let $F$ be any influence function satisfying Property \ref{pro-1} ({\em i.e.}, $(\delta, \gamma)$-stable  at any maximal pair), and $\psi$ be an $\epsilon$-perturbation on $P$ (with a witness point $q$ and $\epsilon \leq \gamma$).
Let $C$ be any subset of $P$ with influence $F(C,q) \geq (1 - \epsilon)F_{max}(q)$. If  $F$ is $(\delta, \gamma)$-stable at $(C,q)$, 
then there exists a constant $\epsilon<\gamma'<1$ and a continuous monotone function $\Delta$ with $\Delta(\gamma')<1$ and $\lim_{x \to 0}\Delta(x) = 0$ such that 
$ F(C',q)\ge (1 - \Delta(\epsilon))F_{max}'(q)$, where $P' = \psi(P)$, $C' = \psi(C)$,  and $F_{max}'(q) = F(C_{m}(P',q),q)$
(see Fig. \ref{fig-lem1}).
\end{lemma}
\begin{proof}
By Definition \ref{stable},
we have
$$ F(C',q) \geq (1 - \delta(\epsilon))F(C,q) \ge (1 - \delta(\epsilon)) (1 - \epsilon) F_{max} (q).$$

Let $J = \psi^{-1}(C_{m}(P',q))$. Since $\psi$ is an $\epsilon$-perturbation, 
by Definition \ref{def-per}, it is easy to see that its inverse $\psi^{-1}$ is an $\epsilon'$-perturbation on $P'$ with  
$\epsilon' = \frac{\epsilon}{1 - \epsilon}$.
If $0< \epsilon < \frac{\gamma}{1 + \gamma}$, then we have $0 < \epsilon' < \gamma$.
Since $C_{m}(P',q)$ is the maximum influence site of $q$ in the power set of $P'$, 
$(C_{m}(P',q),q)$ is a maximal pair. 
 By Property \ref{pro-1}, we know that if $0 < \epsilon' < \gamma$, then
$$F(J,q) \ge (1 - \delta(\epsilon'))F(C_{m}(P',q),q).$$
Also, since $F_{max}(q) \ge F(J,q)$, we have
$$ F(C',q) \geq (1 - \delta(\epsilon)) (1 - \epsilon) F_{max} (q) \geq (1 - \delta(\epsilon)) (1 - \epsilon) F(J,q)$$
$$\geq (1 - \delta(\epsilon)) (1 - \epsilon)(1 - \delta(\epsilon'))F(C_{m}(P',q),q).$$

Thus, we can set $\Delta(\epsilon) = 1 - (1 - \delta(\epsilon)) (1 - \epsilon)(1 - \delta(\frac{\epsilon}{1 - \epsilon}))$, 
and choose a value $\gamma'$ to satisfy the following conditions: (i) $0 < \gamma' < \frac{\gamma}{1 + \gamma}$, and (ii) $\gamma'$ is small enough so that for any $0  < \epsilon \le \gamma'$, $\delta(\epsilon)<1$ and 
$\delta(\frac{\epsilon}{1 - \epsilon})<1$. Then the lemma follows.
\qed
\end{proof}

In Lemma \ref{lem-app} above, the error caused by the perturbation can be estimated by the function $\Delta$. Thus $\Delta$ is also called the {\em error estimation function}.
Since $\Delta$ is a monotone function around $0$, for a sufficiently small value $\epsilon>0$, $\Delta^{-1}(\epsilon)$ exists (this fact will be used later). 
For ease of analysis, we assume that $\epsilon$ is sufficiently small so that $\Delta^{-1}(\epsilon) < 1/2$.

By Property \ref{pro-1}, we know that the locality of an influence function is defined based on perturbation. Since perturbation uses relative error, the locality property is not uniform throughout the entire space. 
Such non-uniformity enables us to achieve objective (a) (discussed in Section \ref{sec-ov}), but does not help attain objective (b). 
For example, when a query point $q$ is far away from some input point, say $p \in P$, an $\epsilon$-perturbation allows $p$ (or equivalently $q$) to change its location by a large distance. 
However, when $q$ is close to $p$,  an $\epsilon$-perturbation can change only slightly ({\em i.e.,} by a distance of $\epsilon \norm{p-q}$) the location of $p$. 
This means that if the influence function has only the locality property, then two query points, say $q_{1}$ and $q_{2}$, which have a distance larger than $2\epsilon \max \{\norm{p-q_{1}}, \norm{p-q_{2}}\}$ to each other cannot be 
grouped into the same Voronoi cell. Since there are infinitely many query points arbitrarily close to $p$, we would need an infinite number of Voronoi cells to approximate their influences. Thus, some additional property is needed to ensure a small-size CIVD ({\em i.e.,} mainly to achieve objective (b)).

To get around this problem, one may imagine a situation that when
a query point $q$ is very close to a subset $C \subseteq P$, it is reasonable to assume that the influence from $C$ completely ``dominates'' the influence from all other points in $P \setminus C$. This means that  when determining the influence for $q$, we can simply ignore 
all points in $P \setminus C$, without losing much accuracy. This suggests that the influence function should also have the following {\em Local Domination} property.

\begin{property}[Local Domination]
\label{pro-2}
There exists a polynomial function $\mathcal{P}(\cdot)$ such that for any point $q$ in $\mathbb{R}^d$ and any subset
$P' \subseteq P$, if there is a point $p \in P'$ with $\mathcal{P}(n)\norm{q-p} < \epsilon\cdot\norm{q - p'}$ for all $p' \in P\setminus P'$ for a sufficiently small constant $\epsilon>0$,
then $F_{max}' (q)> (1-\epsilon)F_{max} (q)$, where $F_{max}'(q) = F(C_{m}(P',q),q)$ (see Fig. \ref{fig-pro2}).
\end{property}

Property \ref{pro-2} above suggests that there is a dominating region for each input point of $P$, which is not too small ({\em i.e.}, not exponentially small comparing to its closest distance to other input points). For each point $p \in P$, consider a ball centered at $p$ and with a
radius $\frac{\epsilon\norm{p-p'}}{2(\mathcal{P}(n)+\epsilon)}$, where $p'$ is the nearest neighbor of $p$ in $P$.
By Property \ref{pro-2}, we know that for any query point $q$ inside this ball, the influence received by $q$ mainly comes from $p$.

Note that the above local domination property naturally holds for some decreasing influence functions ({\em e.g.,} those functions where the influence from each input point $p$ to a query point $q$
decreases polynomially when the distance $\norm{p-q}$ increases). Such influence functions appear in many applications ({\em e.g.}, force-like influence). Still, it remains an open problem to determine whether this property is necessary for all problems to yield small-size approximate CIVDs.

The above three properties and Lemma \ref{lem-app} suggest a way to construct an approximate CIVD. By Property \ref{pro-1}, we know that
it suffices to use a perturbation of $P$ to construct an approximate CIVD. Since our influence function considers the vectors between a
query point $q$ and the input points of $P$, we can equivalently perturb all query points (i.e., the entire space $\mathbb{R}^d$),
instead of the input points, and still obtain an approximate CIVD. This means that we can first
approximate the space $\mathbb{R}^d$ by partitioning it into small enough regions, and then associate each such region with a cluster site having an (approximate) maximum influence on it. 
The set of regions together forms an approximate CIVD. During the partition process, we also use Property \ref{pro-2} to avoid generating regions of too small sizes,
hence preventing a large number of regions.
This leads to our approximate influence decomposition, which is discussed in detail in the next section.

\begin{figure}[ht]
\centering
\includegraphics[height=2.2in]{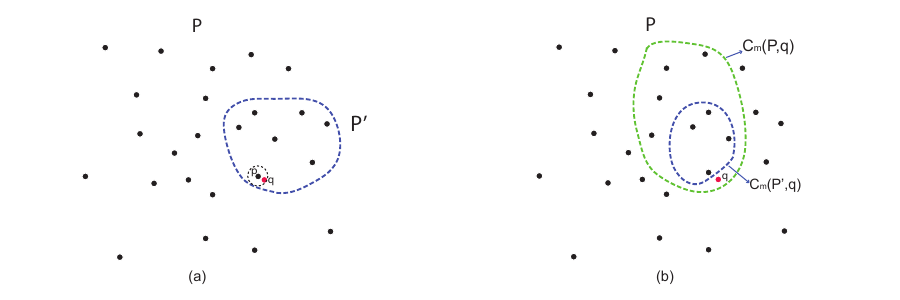}
\caption{Illustrating Property \ref{pro-2}:
(a) A dominating region of $p$ (the dashed circle centered at $p$) and $P'\subseteq P$ satisfying Property \ref{pro-2};
(b) the influence on $q$ from $P$ differs from that on $q$ from $P'$ by only a small $\epsilon$ factor.}
\label{fig-pro2}
\end{figure}

\section{Approximate Influence Decomposition}
\label{sec-AI}

In this section, we present a general space-partition technique called {\em approximate influence (AI) decomposition} for constructing an approximate CIVD.
We assume that the influence function satisfies  the similarity invariant, locality, and local domination properties. 

To build an approximate CIVD, we utilize the locality and local domination properties to partition the space $\mathbb{R}^d$ into two types of cells  ({\em i.e.,} type-1 and type-2 cells).
Our idea for partitioning $\mathbb{R}^{d}$ is based on a new data structure called {\it box-clustering
tree} or simply {\it box-tree}, which is constructed by an extended quad-tree decomposition
and is guided by another new data structure called {\em distance-tree} built by the well-separated pair decomposition \cite{CK95}.
Roughly speaking, the
box-tree construction begins with a big enough bounding box of the input point set $P$
({\em i.e.,} an axis-aligned hypercube), recursively partitions each box 
into smaller boxes, and stops the recursion on a box when a certain condition is met. There are two types of boxes in the partition: One type is a box generated by the normal
quad-tree decomposition ({\em e.g.}, see Fig.~\ref{fig-2box}(a)), and the other type involves
the intersection or difference of two boxes ({\em e.g.}, see Fig.~\ref{fig-2box}(b)). 
The stopping condition of recursion on a box $B$ is that either $B$ is small enough (comparing with
its distance to the closest point in $P$, or equivalently, $B$ is sufficiently far away from $P$ and hence can be viewed as a type-2 cell),
or $B$ is inside the dominating region of some cluster site $C\subseteq P$ (and thus $B$ can be viewed as a type-1 cell).
For the first case, by Property \ref{pro-1}, we know that all points in $B$ can be viewed as perturbations of a single query point and hence share the same approximate maximum influence site. For the second case, by Property \ref{pro-2}, we know that the approximate maximum influence site for all points in $B$ is $C$.  
During the above space-partition process, a box becomes a {\em cell} if no further decomposition of it is needed.

\begin{figure}[h]
\vspace{-0.15in}
\centering
\includegraphics[height=1.8in]{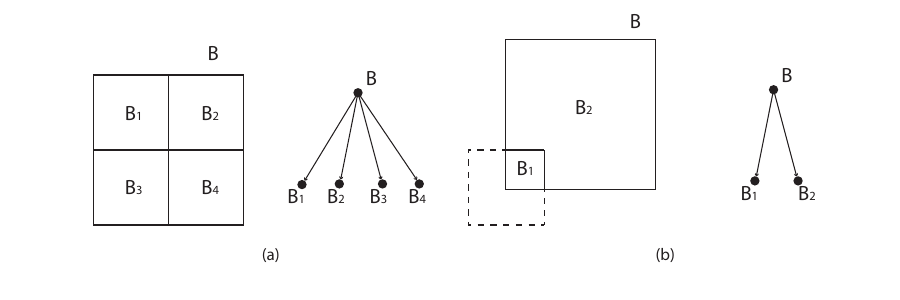}
\caption{Illustrating the two types of boxes in a box-tree $T_{q}$ in $\mathbb{R}^{2}$:
(a) The normal quad-tree boxes; (b) $B_1$ is the intersection of the box $B$ and the partially
dashed box, and $B_2=B-B_1$.}
\label{fig-2box}
\end{figure}

As mentioned in Section \ref{sec-ov}, in order for the resulted Voronoi cells to have the desired properties, we need to overcome a number of difficulties: (1) How to efficiently maintain the
(approximate) distances between all potential cells ({\em i.e.,} the boxes) and the input points
of $P$ so that their types can be determined? (2) How to efficiently generate the two types of cells?
(3) How to bound the total number of  cells and the running time of the space-partition process? Below we discuss our ideas for resolving these difficulties.

\subsection{Distance-tree $T_{p}$ (for Difficulty (1))}

As discussed in Section \ref{sec-ov}, the type of a cell is determined mainly by its distance to the input points of $P$.  Corresponding to the two types of cells, we need to maintain two types of distances for each box $B$ generated by the space-partition process:
(i) the distance, denoted by $r_{min}$, between $B$ and the closest input point (in case $B$ becomes a type-2 cell), and (ii) the distance, denoted by $r_{c}$, between $B$ and the second closest
input point or cluster site (in case $B$ becomes a type-1 cell).  A straightforward way to maintain such information is to explicitly determine the values of $r_{min}$ and $r_{c}$ for each generated box $B$.
But, this would be rather inefficient. The reason is that the number of possible values for $r_{c}$  could be very large (since $B$ could be potentially  in the dominating regions of many different
cluster sites). A seemingly possible method for this problem is to keep track of only the distances between $B$ and the closest and second closest input points.
This means that we consider only the dominating region of a single input point ({\em i.e.,} only checking whether $B$ is in the dominating region of its closest input point). Unfortunately, this could cause the space-partition process to generate unnecessarily many boxes.

To see why this is the case, consider the dominating region of a point $p\in P$. The size of $p$'s dominating region depends on the distance to its nearest neighbor $p'$ in $P$.
If $\norm{p-p'}$ is too small, then the decomposition near $p$ should be stopped at
some range to avoid generating
too many quad-tree boxes ({\em e.g.,} when the box size is smaller than $c\norm{p-p'}$ for
some constant $c>0$). To have a better understanding of this, consider an example in the 2D space $\mathbb{R}^2$ which contains
only three input points, $(0,0)$, $(1,0)$, and $(M,0)$, for some large value $M$. The size of the dominating region of $(1,0)$ is small since its nearest neighbor is $(0,0)$. The space between $(1,0)$ and $(M,0)$ is then decomposed into many (small) boxes in order to generate small enough boxes that are fully contained in the dominating region of $(1, 0)$ (see Fig. \ref{fig-2large}).  
One way to avoid this pitfall is that when a subset $C$ of $P$ is far away from the other points of $P$, we  treat $C$ as a single point. In the above example, we may view $(0,0)$ and $(1,0)$ as forming a ``heavy'' point (with a certain ``weight").
The dominating region size is then based on the distance between the ``heavy'' point and $(M,0)$, which is significantly larger than $1$. In this manner, we can reduce the total number of boxes.

\begin{figure}[ht]
\centering
\includegraphics[height=2.0in]{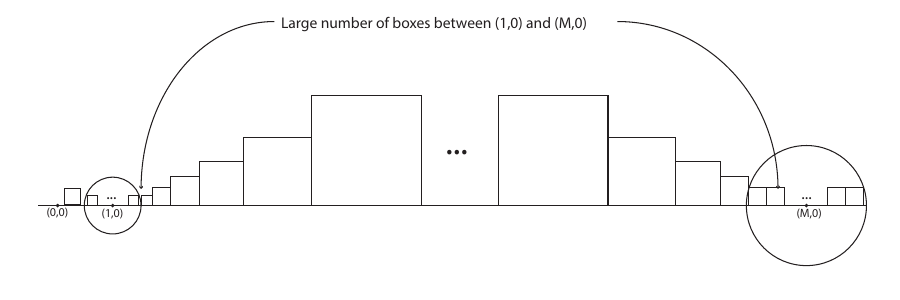}
\caption{An example of 3 input points in $\mathbb{R}^2$ that cause many quad-tree boxes.}
\label{fig-2large}
\end{figure}

This means that we consider the dominating region of a subset of input points only if they can be viewed as a single ``heavy'' input point.  In this way, we can dramatically reduce the number of choices for $r_{c}$, and consequently the cost of maintaining the distances between the boxes and input points.

To implement the above ideas, we use the well-separated pair decomposition (WSPD) \cite{CK95} to first preprocess the input points of $P$. This will result in a tree structure $T_{p}$, called {\em distance-tree}, in which
every node stores the location of one input point together with a value (whose exact meaning will be explained later). 
For ease of analysis, we assume that the error tolerance is $\beta < \frac{1}{2}$. The main steps of our algorithm ({\bf Algorithm \ref{alg-1}}) for constructing the
distance-tree $T_{p}$ are as follows.

\begin{algorithm}[h]
\caption{Preprocessing($P, \beta$)}
\textbf{Input:} A set $P$ of $n$ points in $\mathbb{R}^d$, and an error tolerance $0<\beta <1/2$.\\
\textbf{Output:} A tree  $T_{p}$, in which every node $v$ stores a value $s(v)$, an input point $l(v)$, and is associated with a bounding box $E(v)$ in $\mathbb{R}^d$.
\label{alg-1}
\begin{algorithmic}[1]
\State{Compute a $12$-well-separated pair decomposition $W=\{(A_{1},B_{1}), (A_{2},B_{2}), \ldots, (A_{k},B_{k})\}$ of $P$.} 
\State{Construct a graph $G(W)$
with points in $P$ being its vertices
by connecting the representatives of $A_{i}$ and $B_{i}$, for every $(A_{i},B_{i})\in W$.}
\State{Build a min-priority queue $Q$ for all edges in $G(W)$, based on their edge lengths.}
\State{Build a tree $T_{p}$ in the following bottom-up manner.

For each $p \in P$, there is a leaf node $v_{p}$ in $T_{p}$ ({\it i.e.}, $T_{p}$ is initially a forest of $|P|$ single-node trees), with $s(v_{p})=0$, $l(v_{p})=p$,
and $E(v_{p})$ and $E'(v_{p})$ both being $0$-sized bounding boxes containing $p$.}

{\bf While} $T_{p}$ is not a single tree {\bf Do} 
\begin{itemize}
\item   Extract from $Q$ the shortest edge $e=(p_{1},p_{2})$ with edge length $w(e)$.
If $v_{p_1}$ and $v_{p_2}$ are leaves of two different trees in $T_{p}$ rooted at $v_{1}$ and $v_{2}$,
then create a new node $v$ in $T_{p}$ as the parent of $v_{1}$ and $v_{2}$, and let
$s(v) = s(v_{1}) + s(v_{2})+ w(e)$,  $l(v)$ be either $l(v_{1})$ or $l(v_{2})$,
$E'(v)$ be the box centered at $l(v)$ and with edge length $\frac{4\cdot s(v)}{\beta}$,
and $E(v)$ be the box centered at $l(v)$ and with edge length $\frac{8\cdot s(v)}{\beta}$ (see Fig. \ref{fig-tree}).
\end{itemize}

\end{algorithmic}

\end{algorithm}

Note that in {\bf Algorithm \ref{alg-1}},
since we choose 12 as the approximation factor of the well-separated pair decomposition,
$G(W)$ forms a spanner of $P$ with a stretch factor of $2$ \cite{CK95} (note that the stretch factor $t$ of the spanner can be other constants; we choose $t=2$ for simplicity reason). 
In the resulted distance-tree $T_{p}$, each node $v$ (called a {\em distance-node}) is associated with a point set, $P_{v}$, 
with a diameter upper-bounded by $s(v)$ and with $l(v)$ as its {\it representative point}. 
$P_v$ is the subset of input points in $P$ associated with all leaves
of the subtree of $T_{p}$ rooted at $v$ (see Fig.~\ref{fig-tree}(b)). When a query
point $q$ is far away from $P_{v}$, each point in $P_{v}$ can be viewed as a perturbation of any
other points. Thus, it will not incur too much error if we simply treat them as one ``heavy'' point, represented by $l(v)$.
In this way, we can avoid generating many small boxes in the quad-tree decomposition process and reduce the cost of maintaining the (approximate) distance information between the boxes and the input points.
$E(v)$ gives the boundary for the query point $q$, {\it i.e.}, when $q$ is outside $E(v)$, it is safe to view $P_{v}$ as a single point (in other words, when $q$ is outside the bounding box $E(v)$, $q$ is viewed as {\bf far away} from $P_{v}$). As to be shown later, the edge length of $E(v)$ is crucial for analyzing the worst case performance of our space-partition algorithm. $E'(v)$ is defined only for analysis purpose.

\begin{figure}[h]
\centering
\includegraphics[height=2.2in]{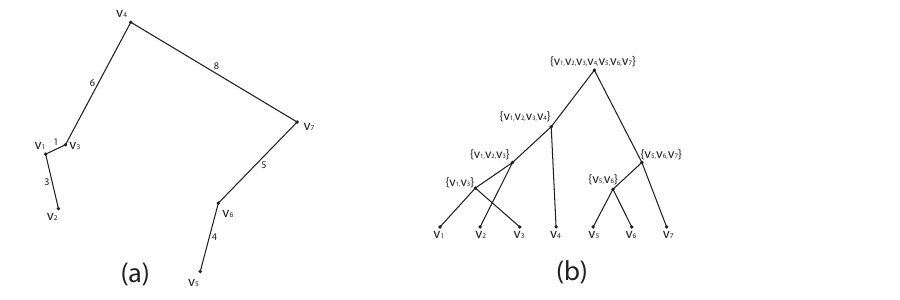}
\caption{An example illustrating {\bf Algorithm \ref{alg-1}}:
(a) Input points $v_1,v_2, \ldots, v_7$ and $G(W)$, with each edge labeled with its length;
(b) the distance tree $T_p$ produced from $G(W)$, in which
each node is for a subset of the input points.}
\label{fig-tree}
\end{figure}

Based on the distance-tree $T_{p}$, we can further reduce the cost of maintaining the distance information between the boxes and the input points. The idea is to use approximation. To see this,
consider a key issue in the space-partition process. Note that the 
space-partition proceeds recursively in a top-down fashion to produce a tree structure,
called {\it box-tree} and denoted by $T_{q}$, with the root of $T_{q}$ corresponding to a large enough bounding box containing all points
of $P$.  Let $u$ be a node (called a {\em box-node}) in the box-tree $T_{q}$. The key issue on $u$ is to determine whether we should further decompose the box $B(u)$ associated with $u$. 
To resolve this issue, we need to know the values of $r_{min}$ and $r_{c}$ ({\em i.e.,} the closest and second closest distances between the input points or ``heavy'' points and $B(u)$). 
Clearly, if such distance information were obtained from scratch for each box $B(u)$, then it would be too costly. 
To overcome this difficulty, observe that if an input point $p \in P$ (or a ``heavy'' point) is sufficiently far away from $B(u)$ (comparing to the edge length of $B(u)$), then
the distance from $p$ to $B(u)$ is a good approximation of the distance from $p$ to any smaller boxes resulted from further decomposition of $B(u)$.
Thus, if we save this distance for future computation on $u$ and its descendants in $T_{q}$, then we
no longer need to consider $p$. Since the decomposition on $B(u)$ depends only on $r_{min}$ and $r_{c}$, 
this means that we can save these two distances and ignore all input points outside $B(u)$.

\subsection{Box-tree $T_{q}$ (for Difficulty (2))}

Suppose the distance-tree $T_{p}$ has already been constructed. We now discuss our idea for efficiently building the {\em box-tree} $T_{q}$ ({\em i.e.,} for resolving difficulty (2)). 

To show how to build $T_{q}$, consider an arbitrary box-node $u$ of $T_{q}$. As we indicated earlier, the key issue on $u$ is to determine whether its associated box $B(u)$  should be further 
decomposed. To resolve this issue, we maintain a list $L=\{v_1,v_2,\ldots,v_k\}$ of distance-nodes in $T_{p}$. Each distance-node $v_{i} \in L$ is associated with a subset $P_{v_{i}}$ of input points which 
may possibly give arise to the distance $r_{min}$ for $B(u)$ (and also possibly the distance $r_{c}$). 
The value of $r_{c}$ is recursively maintained to approximate 
the closest distance from $B(u)$ to all points in $P\setminus \cup_{i=1}^{k} P_{v_{i}}$ ({\em i.e.,} all input points not in $L$). 

To determine whether $B(u)$ should be decomposed, 
we examine all distance-nodes in $L$. 
For each $v_{i} \in L$, there are three possible cases to consider. The first case is that the bounding box $E(v_{i})$ of $v_{i}$ significantly overlaps with $B(u)$ (see
{\bf Algorithm \ref{alg-2}} for the exact meaning of ``significant overlapping").
In this case, the region $B(u) \cap E(v_{i})$ is not far away from $P_{v_{i}}$, and thus we cannot view $P_{v_{i}}$ as a single ``heavy''
point.  This means that we cannot use $l(v_{i})$ ({\em i.e.,} the representative point of $P_{v_{i}}$) to compute the value of $r_{min}$. To handle this case, we replace $v_{i}$ in $L$ by its two children, say $v_{i,1}$ and $v_{i,2}$, in the distance-tree $T_{p}$. This can potentially
increase the distance between $B(u)$ and each of $P_{v_{i,1}}$ and $P_{v_{i,2}}$, and hence enhance the chance for $B(u)$ to be far away from these two child nodes.  

The second case is that $B(u)$ is far away from $P_{v_{i}}$. In this case, we remove $v_{i}$ from $L$ and save its distance (to $B(u)$) in $r_{c}$
if it is smaller than the current value of $r_{c}$. If all distance-nodes are removed from $L$ in this way, then it means that $B(u)$ is far away from all input points and therefore becomes a type-2 cell. When this occurs, the value of $r_{min}$ is the value of $r_{c}$ at the time when $L$ becomes empty. 

The third case is that $v_{i}$ does not fall in any of the above two cases. 
In this scenario, if $v_{i}$ is the only distance-node left in $L$ and $B(u)$ (or part of $B(u)$) is inside the dominating region of $P_{v_{i}}$, then the part of $B(u)$ outsides $E(v_{i})$ becomes a type-1 cell,
and the part of $B(u)$
inside $E(v_{i})$ will be recursively determined for its decomposition. Otherwise, either multiple distance-nodes are still in $L$ or $B(u)$ is not in the dominating region of
$P_{v_{i}}$. For both these sub-cases, we decompose $B(u)$
into $2^{d}$ sub-boxes and recursively process each sub-box.

To generate the box-tree $T_{q}$, we use a recursive algorithm called {\em AI-Decomposition}, in which $\mathcal{P}(\cdot)$ is a polynomial function for Property \ref{pro-2}.
The core of this algorithm is a procedure called {\em Decomposition}, which produces the box-subtree of $T_{q}$ rooted at
a box-node $u$ that is part of the input to the procedure. In the procedure Decomposition, Step 1 corresponds to the first case; Steps 2 and 3 are for the second case; Steps 4 and 5 handle the third case.

It should be pointed out that in the procedure Decomposition, each recursive call inherits a copy of $L$;
thus, different recursive calls use their own copies of $L$, and such copies are independent of one another.
This means that the same node $v$ of $T_p$ can appear in (and also get removed from) different copies of $L$
throughout the algorithm.

\begin{algorithm}[t]
\caption{Decomposition$(u,\beta,L,T_{p},r_{c})$}

\textbf{Input:} A box-node $u$ with box $B(u)$, error tolerance $\beta > 0$, distance-tree $T_{p}$, linked list $L$, and a value $r_{c}$.
\textbf{Output:} A subtree of $T_{q}$ rooted at $u$ (see Fig.~\ref{fig-algo2}).
\label{alg-2}

\begin{algorithmic}[1]
\State{{\bf While} $\exists$ $v$ in $L$ such that the length of at least one edge of $B(u)\cap E(v)$ is no smaller than $\frac{edgeLength(B(u))}{2}$ {\bf do}
\begin{itemize}
\item Replace  $v$ in $L$ by its two children in $T_{p}$, if any.
\end{itemize}}

\State{Let $D(u)$ be the diameter of $B(u)$.  For each node $v$ in $L$ {\bf do}
\begin{enumerate}
\item[2.1] Let $r_{min}$ be the distance between $B(u)$ and $l(v)$.
\item[2.2] If $D(u) < r_{min}\beta/2$, remove $v$ from $L$, and if $r_{c} > r_{min}$, let $r_{c} = r_{min}$.
\end{enumerate}}
\State{If $L$ is empty, return, and $B(u)$ becomes a {\bf type-2} cell.}
\State{If there is only one element $v$ in $L$, let $r_{min}$ be the smallest distance between $l(v)$ and $B(u)$.
\begin{itemize}
\item[4.1]  If $\frac{r_{min} + D(u)}{r_{c}} <  \frac{\beta}{2\mathcal{P}(n)}$, then

\begin{itemize}
\item[4.1.1]  
If $E(v) \cap B(u)=\phi$ or $v$ is a leaf node in $T_{p}$, $B(u)$ is a {\bf type-1} cell dominated by $v$. Return. 
\item[4.1.2] Let $B'$ be the smallest hypercube box in $B(u)$ fully containing $B(u)\cap E(v)$. Create two box-nodes $u_{0}$ and $u_{1}$, with $u_{0}$ corresponding to $B'$ and $u_{1}$ corresponding to the difference of $B(u)$ and $B'$. Let $u_{0}$ and $u_{1}$ be the children of $u$ in $T_{q}$. In this case, $u_{1}$ is a {\bf type-1} cell dominated by $v$.
\item[4.1.3] Replace $v$ in $L$ by its two children $v_{1}$ and $v_{2}$ in $T_{p}$. Call Decomposition$(u_{0},\beta,L,T_{p},r_{c})$, and return.
\end{itemize}
\end{itemize}
}

\State{Decompose $B(u)$ into $2^{d}$ smaller boxes, and make the corresponding nodes $u_{1}, u_{2}, \ldots, u_{2^{d}}$ as the children of $u$ in $T_{q}$. Call 
Decomposition$(u_{i},\beta,L,T_{p},r_{c})$ for each $u_{i}$. Return.}

\end{algorithmic}
\end{algorithm}

\begin{figure}[ht]
\centering
\includegraphics[height=2.0in]{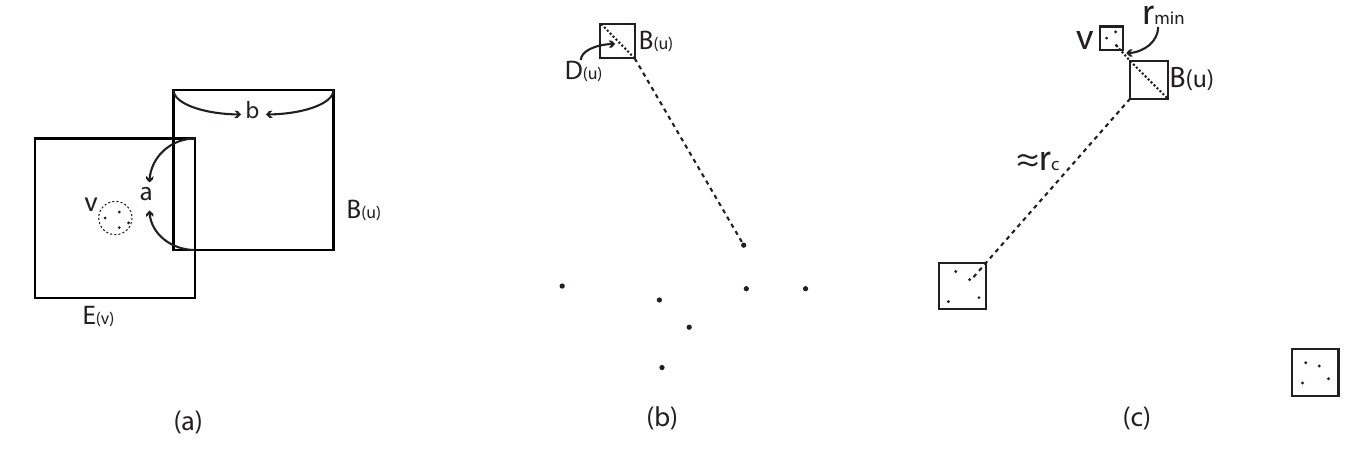}
\caption{Examples illustrating {\bf Algorithm \ref{alg-2}}:
(a) The case for Step 1, {\it i.e.}, $B(u) \cap E(v)$ has an edge with a length
$a \ge edgeLength(B(u))/2=b/2$; This means that a considerable large part of $B(u)$
intersects $E(v)$, therefore input points in distance node $v$(viewed as a subset of $P$ here) might not be viewed as one.
$v$ should be replaced in the list $L$.
(b) Case for step 3. Step 2 removes far away distance nodes. If $L$ becomes empty in step 3,
it means all input points are far away from $B(u)$
(c) The case for Step 4.1.
Here $r_c$ is used as an approximation of closest distance between $B(u)$ and points not in $v$.
When this happens, $B(u)$ is very close to points in $v$ compared to points that are not. }
\label{fig-algo2}
\end{figure}

\begin{algorithm}[H]
\caption{AI-Decomposition($P$, $\beta$)}
\textbf{Input:} A set $P$ of $n$ points in $\mathbb{R}^d$, and a small error tolerance $\beta>0$.\\
\textbf{Output:} A box-tree $T_{q}$.
\label{alg-3}
\begin{algorithmic}[1]
\State{Run the preprocessing algorithm on $P$ and obtain a distance-tree  $T_{p}$. Let $u$ be the root of $T_{p}$. View $E(u)$  as a box-tree node. Run Decomposition$(E(u), \beta, \{u\}, T_{p}, \infty)$.}
\State{Output the box-tree rooted at $E(u)$ as $T_{q}$.}
\end{algorithmic}
\label{algo-frm}
\end{algorithm}

Below we analyze the above algorithms.

\subsection{Algorithm Analysis (for Difficulty (3))}

Proving the correctness and running time of {\bf Algorithm \ref{alg-3}} is nontrivial. We first show some properties of the AI decomposition which will be used for proving the correctness and running time or  for designing the assignment algorithms in Sections \ref{sec-ex1} and \ref{sec-ex2}. We start the analysis with the following definition.

\begin{definition}
\label{def-rec}
A distance-node $v \in T_{p}$ is said to be {\em recorded} for a box-node $u$ if $v$ is removed from the list $L$ in Step 2.2 of {\bf Algorithm \ref{alg-2}} when processing $u$ or one of $u$'s ancestors in $T_{q}$.
The value of $r_{min}$ in the iteration when $v$ is removed from $L$ is
the \emph{recorded distance} of $v$ for $u$. If $v$ is recorded for $u$, then any point $p \in P_v$ is also {\em recorded} for $u$ with the same recorded distance as $v$. 
\end{definition}

The following lemma shows a useful property of the AI decomposition.

\begin{lemma}
\label{lm:distance}
If a point $p \in P$ is recorded for a box-node $u$ with a recorded distance $x$, then for any point $q \in B(u)$,
\[
(1 - \beta)x \leq \norm{p - q} \leq (1 + \beta)x.
\] 
\end{lemma}

\begin{proof}
Let $v$ be the distance-node such that $P_v$ contains $p$ and $v$ is recorded for $u$. Let $u'$ be the box-node being considered at the time when $v$ is removed from $L$. By Definition \ref{def-rec}, we know that $u'$ is 
either $u$ or an ancestor of $u$ in $T_{q}$.
Let $q$ be any point in $B(u)$. Obviously, $q$ is also in $B(u')$. Let $q'$ be the closest point in $B(u')$ to $l(v)$.
(See \textbf{Figure} \ref{fig-lmdis} to help understanding the configuration.)
Then by the definition of $r_{min}$, we know $x = \norm{q' - l(v)}$.
By {\bf Algorithm \ref{alg-2}}, we have $D(u') < x\beta/2$,
where $D(u')$ is the diameter of $B(u')$. 
Thus,
\begin{equation}\label{eq:dis1}
\norm{q - q'} \leq \norm{q' - l(v)}\beta/2.
\end{equation}

Now we claim that
\begin{equation}\label{eq:dis2}
\norm{p - l(v)} \leq s(v) \leq \norm{q' - l(v)}\beta/2.
\end{equation}

To prove this, we first show that $B(u')$ does not intersect $E'(v)$.
Assume by contradiction that it is not the case. Then $q'$ is included in $E'(v)$.
Recall that the edge length of $E'(v)$ is $\frac{4s(v)}{\beta}$.
Thus $\norm{q' - l(v)} \leq \frac{2\sqrt{d}s(v)}{\beta}$.
Since $\beta < 1/2$,
we have $$D(u') < \frac{x\beta}{2} \leq \frac{\norm{q' - l(v)}}{4} \leq \frac{\sqrt{d}s(v)}{2\beta}.$$
This means that the edge length of $B(u')$ is no bigger than $\frac{s(v)}{2\beta}$, which is smaller than half the edge length of $E'(v)$.
Combining  this with the assumption that $B(u')$ intersects $E'(v)$, we know that $B(u')$ is entirely inside $E(v)$ (whose edge length is two times of that of $E'(v)$).
This means that 
$v$ should have already been removed from $L$ in Step 1, instead of in Step 2 (by {\bf Algorithm \ref{alg-2}}). But this is a contradiction.

Since $B(u')$ does not intersect $E'(v)$, $\norm{q' - l(v)}$ must be larger than half the edge length of $E'(v)$, which is $\frac{2s(v)}{\beta}$. By the definition of $s(v)$, we also know $\norm{p - l(v)} \leq s(v)$. Thus, Claim (\ref{eq:dis2}) easily follows.

Combining (\ref{eq:dis1}) and (\ref{eq:dis2}) and based on the triangle inequality, we obtain
\[
(1 - \beta)\norm{q' - l(v)} \leq \norm{q' - l(v)} - \norm{q - q'} - \norm{p - l(v)} \leq \norm{p - q}
\]
and
\[
\norm{p - q} \leq \norm{q' - l(v)} + \norm{q - q'} + \norm{p - l(v)} \leq (1 + \beta)\norm{q' - l(v)}
\]

The lemma follows from the fact that $x = \norm{q' - l(v)}$.
\qed
\end{proof}

\begin{figure}[ht]
\centering
\includegraphics[height=2.5in]{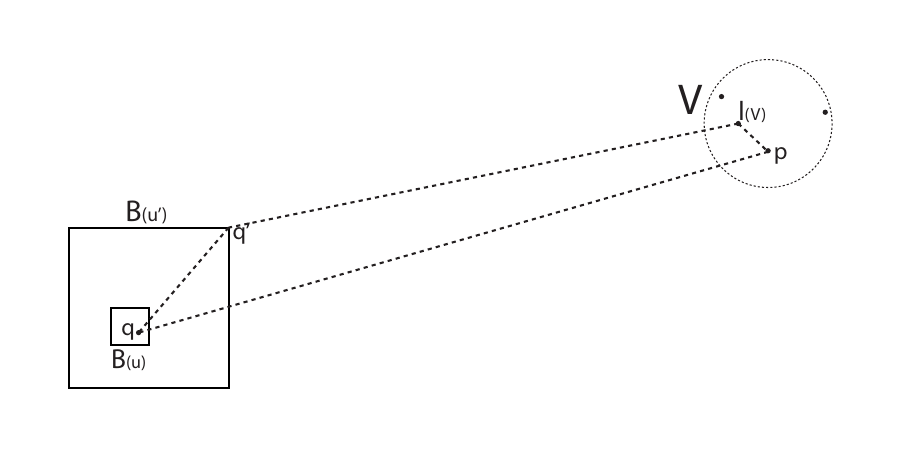}
\caption{A figure for \text{Lemma} \ref{lm:distance}.}
\label{fig-lmdis}
\end{figure}

The next two lemmas show some important properties of the type-2 cells.  

\begin{lemma}\label{lm:partition}
For any type-2 cell $c$ produced by {\bf Algorithm \ref{alg-3}}, the set of distance-nodes (also viewed as subsets of the input points) recorded for $c$ forms a partition of $P$.
\end{lemma}

\begin{proof}
For any point $p \in P$, if $p$ is not recorded at the end of processing a box-node $u$ but $p$
is in some distance-node $v$ such that $p \in P_v$ and $v \in L$ at the time of processing $u$,
then it must be the case that some distance-node $v'$ remains in $L$ at the end of processing $u$.
By {\bf Algorithms \ref{alg-2}} and {\bf \ref{alg-3}}, we know that initially, every point of $P$ is included in $L$, and the only possibility for $p$ not appearing in any of the distance-nodes in $L$ is that at some iteration, $p$ becomes recorded.
Since we have a type-2 cell $c$ only when $L$ is empty, this means that $p$ must become recorded for $c$ at some point of the recursion. 
\qed
\end{proof}

\begin{lemma}\label{lm:cell2}
For any type-2 cell $c$ and any $p \in P$, let $D(c)$ be the diameter of $c$ and $r$ be the shortest distance between $c$ and $p$. Then
\[
D(c) \leq \frac{2r\beta}{3}.
\]
\end{lemma}

\begin{proof}
Assume that $p$ becomes recorded when processing a box-node $u$, where $u$ is either $c$ or an ancestor of $c$ in $T_q$.
We first show $D(u) \leq \frac{2r_u\beta}{3}$,
where $D(u)$ is the diameter of $B(u)$ and $r_u$ is the shortest distance between $B(u)$ and $p$.
Let $v$ be the distance-node that contains $p$ and is removed in Step 2.2 of {\bf Algorithm \ref{alg-2}},
and $q$ be the closest point in $B(u)$ to $l(v)$. Then
\begin{equation}\label{eq:dis3}
D(u) \leq \frac{\norm{q - l(v)}\beta}{2}.
\end{equation}
By using Claim (\ref{eq:dis2}) in the proof of  Lemma \ref{lm:distance}, we can see that
\[
\norm{p - l(v)} \leq \frac{\norm{q - l(v)}\beta}{2}.
\]
Let $q'$ be the closest point in $B(u)$ to $p$. By the assumption of $\beta < \frac{1}{2}$ and the triangle inequality, we have
\[
\norm{p - l(v)} + \norm{q' - p} \geq \norm{q' - l(v)} \geq \norm{q - l(v)}, \mbox{and}
\]
\[
r_u = \norm{q' - p} \geq \norm{q - l(v)} - \norm{p - l(v)} \geq (1 - \frac{\beta}{2})\norm{q - l(v)} \geq \frac{3\norm{q - l(v)}}{4}.
\]
Plugging these into (\ref{eq:dis3}), we obtain $D(u) \leq \frac{2r_u\beta}{3}$.

Now compare $D(c)$ and $r$ with $D(u)$ and $r_u$. Since $c$ is contained inside $B(u)$, we have $D(c) \leq D(u)$ and $r \geq r_u$.
Thus, the lemma follows.
\qed
\end{proof}

The lemma below characterizes the type-1 cells.

\begin{lemma}\label{lm:cell1}

If $c$ is a type-1 cell dominated by a distance-node $v$, then
for any point $q \in c$ and any point $p' \in P\setminus P_{v}$,
\[
\frac{\norm{q - l(v)}}  {\norm{q - p'}} \leq \frac{\beta}{\mathcal{P}(n)}.
\]

\end{lemma}

\begin{proof}
Since $c$ is a type-1 cell dominated by $v$,  $v$ must be the only element in $L$ after Step 2 of {\bf Algorithm \ref{alg-2}}.
This means that any $p' \in P\setminus P_{v} $ must have already been recorded with a distance, say $x$.
By Lemma \ref{lm:distance}, we know $\norm{q - p'} \geq (1 - \beta)x$.
Since $r_c$ maintains the minimum of all recorded distances, we have $x \geq r_c$.

By Step 4.1 of {\bf Algorithm \ref{alg-2}}, we know $\frac{r_{min} + D(u)}{r_{c}} < \frac{\beta}{2\mathcal{P}(n)}$,
where $r_{min}$ is the distance between $B(u)$
(which becomes $c$)
and $l(v)$, and $D(u)$ is the diameter of $B(u)$.
Since $q$ is in $B(u)$,  $\norm{q - l(v)} \leq r_{min} + D(u)$. Thus
\[
\frac{\norm{q - l(v)}}{\norm{q - p'}} \leq \frac{r_{min} + D(u)}{(1 - \beta)x} \leq
\frac{r_{min} + D(u)}{(1 - \beta)r_{c}} \leq \frac{\beta}{2(1 - \beta)\mathcal{P}(n)} \leq \frac{\beta}{\mathcal{P}(n)},
\]
where the last inequality is by the assumption of $\beta < \frac{1}{2}$. Hence the lemma holds.
\qed
\end{proof}

The following definition is mainly for the proof of Theorem \ref{the-aid} below.

\begin{definition}
In $\mathbb{R}^d$, let $C$ be a set of $k$ coincident points and $q$ be any query point. The
{\em maximum duplication function} $\rho$ for an influence function $F$ satisfying Property \ref{pro-3} is defined as $\rho(k)=|C_{m}(C,q)|$ ({\em i.e.}, the cardinality of $C_{m}(C,q)$).
For any set $C'$ of $k$ points in $\mathbb{R}^d$ (not necessarily coincident points), the {\em selection mapping} $\eta$ maps $C'$ to an arbitrary subset  $\eta(C')$ of $C'$ with cardinality  $\rho(k)$.
\end{definition}

Note that in the above definition, it is possible that, for some influence function $F$, the maximum influence of a set $C$ of $k$ coincident points on a query point $q$ is attained by a subset of  $C$. By Property \ref{pro-3}, we know that $\rho(k)$ depends only on the influence function $F$ and is independent of $C$ and $q$.

The following theorem ensures that all points in each cell generated by the AI decomposition have a common approximate maximum influence site ({\em i.e.,} the correctness of the AI decomposition).

\begin{theorem}
\label{the-aid}
Let $c$ be any cell generated by the AI-Decomposition algorithm with an error tolerance $\beta = \Delta^{-1}(\epsilon)$, where $\Delta$ is the error estimation function.  Then the following holds.
\begin{enumerate}

\item
If $c$ is a type-1 cell dominated by a distance-node $v$, then for any query point $q \in c$,
$F(\eta(P_v), q) \geq (1 - \epsilon)F(C_{m}(P,q),q)$.

\item
If $c$ is a type-2 cell and $q'$ is an arbitrary point in $c$, then $F(C_{m}(P,q'), q) \geq (1 - \epsilon)F(C_{m}(P,q),q)$ for any point $q \in c$. Furthermore, if there exists a subset $C \subseteq P$ such that $F(C, q') \geq (1 - \beta)F(C_{m}(P,q'),q')$ and $(C,q')$ is a stable pair, then $F(C, q) \geq (1 - \epsilon)F(C_{m}(P,q),q)$ for any point $q$ in $c$. 

\item
For any query point $q$ outsides the bounding box $B(u_{root})$,  $F(\eta(P_{v_{root}}), q) \geq (1 - \epsilon)F(C_{m}(P,q),q)$, 
where $u_{root}$ is the root of $T_{q}$ and $v_{root}$ is the root of $T_{p}$.

\end{enumerate}
\end{theorem}

\begin{proof}
For case 1 above, 
we define a mapping $\psi_1$ on $P$  as follows.
\[
\psi_1(p) =
\begin{cases}
p   &\text{if $p\not\in P_v$,}\\
l(v) &\text{if $p\in P_v$}.
\end{cases}
\]

Note that $\psi_{1}(P) = \psi_{1}(P_v) \cup \psi_{1}(P \setminus P_v)$ ($\psi_{1}(\cdot)$
is a multiset). By Lemma \ref{lm:cell1}, we know that for any $p \in \psi_{1}(P_v)$ and $p' \in \psi_{1}(P \setminus P_v)$,
\[
\frac{\norm{p - q}}{\norm{q - p'}} \leq \frac{\beta}{\mathcal{P}(n)}.
\]
Since $\psi_{1}(\eta(P_v))=C_{m}(\psi_{1}(P_v), q)$, by Property \ref{pro-2}, we have
\begin{equation}\label{eq:t1}
F(\psi_{1}(\eta(P_v)),q) \geq (1 - \beta)F_{max}(\psi_{1}(P),q).
\end{equation}
In this case, $c$ does not intersect $E(v)$ (by {\bf Algorithm \ref{alg-2}}). 
This means that the minimum distance between $q$ and $l(v)$ is greater than $\frac{4s(v)}{\beta}$.
By {\bf Algorithm \ref{alg-1}}, we know that the distance between $l(v)$ and any point in $P_v$ is upper-bounded by $s(v)$.
It is easy to see that the inverse $\psi_1^{-1}$ of $\psi_{1}$ is a $\beta$-perturbation with respect to $q$.
By (\ref{eq:t1}), Lemma \ref{lem-app}, and the fact that $(\psi_{1}(\eta(P_v)),q)$ is a maximal and stable pair, we know
$F(\psi_{1}^{-1}(\psi_{1}(\eta(P_v))),q) \geq (1 - \beta)F_{max}(\psi_{1}^{-1}(\psi_{1}(P)),q)$.
Thus $F(\eta(P_v), q) \geq (1 - \epsilon)F(C_{m}(P,q),q)$.

For case 3, note that $B(u_{root})$ is  simply $E(v_{root})$ and $P_{v_{root}} = P$. 
By the same argument as for case 1,  we can show that $F(\eta(P_{v_{root}}), q) \geq (1 - \epsilon)F(C_{m}(P,q),q)$ for any $q$ outsides $B(u_{root})$.

For case 2, we only prove the second part of this case since it implies the first part.
Let $q$ be any fixed point in $c$ and $\psi_2$ be a mapping which maps every point $p \in P$ to a point at the location of $\psi_2(p) = p + q' - q$ ({\em i.e.,} $\psi_{2}$ is a translation). 
Clearly, for any $P' \subseteq P$, $F(\psi_2(P'),q') = F(P',q)$ (by Property \ref{pro-3}).
By Lemma \ref{lm:cell2}, we know $\norm{q - q'} \leq \frac{2\norm{q - p}\beta}{3}$, for any $p \in P$.
This means that $\psi_2$ is a $\beta$-perturbation with respect to $q'$.
Since  $F(C, q') \geq (1 - \beta)F(C_{m}(P,q'), q')$,
by Lemma \ref{lem-app}, we have $F(\psi_2(C), q') \geq (1 - \epsilon)F(C_{m}(\psi_2(P),q'),q')$. If we translate all points back to their original positions, then $\psi_{2}(P)$ becomes $P$ and $q'$ becomes $q$. By Definition \ref{def-if} and Property \ref{pro-3}, we know that the influence remains the same under translation. Thus, we have  
 $F(C, q) \geq (1 - \epsilon)F(C_{m}(P,q),q)$.
Since $(C_{m}(P,q'), q')$ is a maximal pair and hence a stable pair by Property \ref{pro-1},
it follows that for any $q$ in $c$, $F(C_{m}(P,q'), q) \geq (1 - \epsilon)F(C_{m}(P,q),q)$.
\qed

\end{proof}

The following packing lemma is a key to upper-bounding the total number of type-1 and type-2 cells and the running time of the AI decomposition ({\em i.e.,} Theorem \ref{the-time} below).  It is also a key lemma for designing our efficient assignment algorithm for the vector CIVD problem.

\begin{lemma}[Packing Lemma]
\label{lem-ann}
Let $o_{c}$ be any point in 
$\mathbb{R}^d$, and $S_{in}$ and $S_{out}$ be two $d$-dimensional boxes ({\em i.e.,} axis-aligned hypercubes) co-centered at $o_{c}$
and with edge lengths $2r_{in}$ and $2r_{out}$, respectively, with $0 < r_{in} < r_{out}$.
Let $\mathcal{B}$ be a set of mutually disjoint $d$-dimensional boxes such that for any $B \in \mathcal{B}$, $B$ intersects the region $S' = S_{out} - S_{in}$ (i.e., the region sandwiched by $S_{in}$ and $S_{out}$) and  its edge length $L(B) \geq C \cdot r$, where $r$ is the minimum distance between $B$ and $o_{c}$ and $C$ is a positive constant.
Then $|\mathcal{B}| \leq C'(C,d)\log(r_{out}/r_{in})$, where $C'(C,d)$ is a constant depending only on $C$ and $d$.
\end{lemma}

\begin{proof}
We prove a slightly different version of this lemma, where $S_{in}$ and $S_{out}$ are two $d$-dimensional balls co-centered at $o_{c}$ and with radii $r_{in}$ and $\sqrt{d} \times r_{out}$, respectively (see Fig. \ref{fig-ann}(a)).
The outer ball can be viewed as the minimum enclosing ball of the original outer box $S_{out}$ and the inner ball can be viewed as the maximum inscribed ball of the original
inner box $S_{in}$. Since the new region $S'= S_{out} - S_{in}$ contains the original region $S'$, any box intersecting the original $S'$ also intersects the new $S'$. Thus, the size of $\mathcal{B}$ can
only increase in the new version. The difference is that the radii $\sqrt{d} \times r_{out}$ and $r_{in}$ have changed by a constant factor depending on $d$. Thus, the new version of the lemma implies the original version.

Without loss of generality, we assume that $o_{c}$ is at the origin of $\mathbb{R}^{d}$.
We first consider the special case that every box of $\mathcal{B}$ is entirely contained inside $S'$. 

For each $B \in \mathcal{B}$, let $r_{max}(B)$ be the maximum distance from $B$ to $o_{c}$ and $r_{min}(B)$ be the minimum distance from $B$ to $o_{c}$.
By the statements of this lemma, we know that the edge length $L(B) \geq C \cdot r_{min}(B)$. Since $r_{max}(B) \leq r_{min}(B) + dL(B)$, 
we have $r_{max}(B) \leq (d + 1/C)L(B)$. Thus, $L(B) \geq C'r_{max}(B)$ for all $B \in \mathcal{B}$, where $C'=1/(d + 1/C)$. 

Define a function $f:\mathbb{R}^d \to \mathbb{R}$ as $f(p) = r(p)^{-d}$ for any $p \in \mathbb{R}^d$, where $r(p)$ is the distance between $p$ and $o_{c}$. Then, we have $\int\limits_{S'} f= C_d\log(r_{out}/r_{in}) $, where $\int\limits_{S'} f$ is the integration of $f$ over $S'$ and $C_d$ is a constant depending only on $d$.

Now consider $\sum_{B \in \mathcal{B}} \int\limits_{B} f$. Since all boxes in $\mathcal{B}$ are disjoint and completely contained in $S'$, we have $$\sum_{B \in \mathcal{B}} \int\limits_{B} f \leq \int\limits_{S'} f = C_d\log(r_{out}/r_{in}).$$

For each $B \in \mathcal{B}$, since $r(p) \leq r_{max}(B)$ for any $p\in B$, we have a lower bound, $(r_{max}(B))^{-d}$, on the value of $f$. 
This implies that $\int\limits_{B} f \geq (L(B))^d \cdot (r_{max}(B))^{-d}$. Since $L(B) \geq C'r_{max}(B)$, we have  $\int\limits_{B} f \geq {C'}^d$ for any $B \in \mathcal{B}$. Thus, $\sum_{B \in \mathcal{B}} {C'}^d \leq \sum_{B \in \mathcal{B}} \int\limits_{B} f \leq \int\limits_{S'} f = C_d\log(r_{out}/r_{in}) $. This means $|\mathcal{B}| \leq {C'}^{-d} C_d \log(r_{out}/r_{in})$.

Now, we consider the general case that $B$ may not be fully contained inside $S'$.  
$\mathcal{B}$ can be partitioned as $\mathcal{B} = \mathcal{B}_1 \cup \mathcal{B}_2 \cup \mathcal{B}_3$, where $\mathcal{B}_1$ is the set of boxes fully inside $S'$ (see Fig. \ref{fig-ann}(b)), $\mathcal{B}_2$ is the set of boxes intersecting both the inside and outside regions of $S_{out}$, and $\mathcal{B}_3$ is the set of boxes intersecting $S_{in}$ (see Fig. \ref{fig-ann}(a)). By the above discussion, we know 
$|\mathcal{B}_1| \leq C_1 \log(r_{out}/r_{in})$, where the constant $C_1$ depends only on $C$ and $d$. Below, we will show that $|\mathcal{B}_2|$ and $|\mathcal{B}_3|$ are both bounded by some constants depending only on $C$ and $d$.  

Let $B'$ be any box in $\mathcal{B}_2$. Then $B'$ intersects both the inside and outside regions of $S_{out}$. Consider the following process to determine a box $R(B')$ from $B'$. Let $B_0$ be $B'$. For $i=0, 1, \ldots$, iteratively 
divide $B_i$ into $2^d$ smaller boxes in a quad-tree decomposition fashion. 
At the $i$-th iteration, try to find a small box such that it intersects both the inside and outside regions of $S_{out}$, and its edge length is no smaller than $C$ times its closest distance to the origin. If such a small box exists, then let it be $B_{i+1}$ and continue to the next iteration. Repeat this process until no such small box exists.  We let the last $B_i$ be $R(B')$.
Note that  $R(B')$ must exist, since in each iteration, the edge length of the box is halved.
Eventually, the edge length of the box will be smaller than $C$ times its closest distance to the origin.

Let $\mathcal{B}'_2 = \{R(B') \mid B' \in \mathcal{B}_2 \}$. Let $B''$ denote $R(B')  \in \mathcal{B}'_2$. We claim that $B''$ is fully contained in a big box $B_{bound}$, where $B_{bound}$ is a box centered at the origin $o_{c}$ and with an edge length $L_{bound} = 2r_{out} + \max\{8r_{out}, 4Cr_{out}\}$. For contradiction, suppose this is not the case.
Then since $B''$ intersects $S_{out}$ and the outside region of $B_{bound}$, it implies that the edge length of $B''$ is no smaller than both $4r_{out}$ and $2Cr_{out}$.
Divide $B''$ into $2^d$ smaller boxes in a quad-tree decomposition fashion. 
Let $p$ be the point in $B''$ that is the closest to $o_{c}$. Then one of these smaller boxes, say $B_p$, contains $p$. Clearly, $B_p$ intersects $S_{out}$ and is also not completely inside 
$S_{out}$. This is due to the fact that the edge length of $B_p$ is no smaller than $2r_{out}$. 
Furthermore, we know that the edge length of $B_p$ is no smaller than $Cr_{out}$. But this is a contradiction, since $B_p$ satisfies the condition in the above iterative selection process and $R(B')$ should not be in $\mathcal{B}'_{2}$.

Note that for any $B'' \in \mathcal{B}'_2$, the edge length of $B''$ is larger than $(1/(d + 1/C))r_{max}(B'')$, where $r_{max}(B'')$ is the largest distance between a point in $B''$ and the origin $o_{c}$. Since $r_{max}(B'') \geq r_{out}$, the edge length of $B''$ is larger than $(1/(d + 1/C))r_{out}$. Also since all boxes in $\mathcal{B}'_{2}$ are disjoint and contained in $B_{bound}$, which has an edge length of $2r_{out} + \max\{8r_{out}, 4Cr_{out}\}$,
by comparing the volumes of $B''$ and $B_{bound}$, it is easy to see that the total number of such boxes is bounded by a constant depending only on $C$ and $d$ (as $r_{out}$ is canceled out). 

By a similar argument, we can show that  $|\mathcal{B}_3|$ is also bounded by a constant. Hence the lemma follows.
\qed
\end{proof}

\begin{figure}[h]
\centering
\includegraphics[height=2.2in]{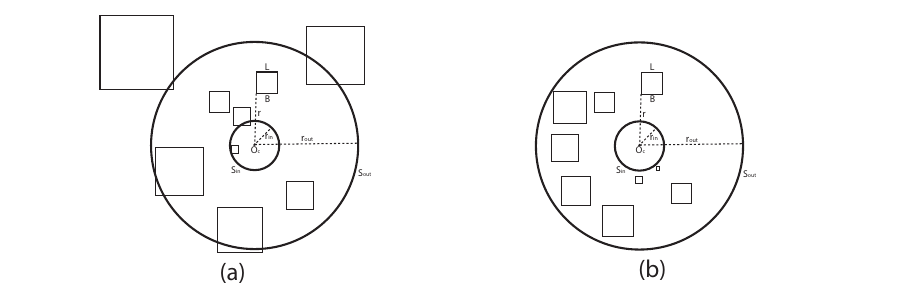}
\caption{An example illustrating Lemma \ref{lem-ann}.}
\label{fig-ann}
\end{figure}

The next lemma is needed by the proof of Theorem \ref{the-time}.

\begin{lemma}
The AI-Decomposition algorithm eventually stops.
\end{lemma}

\begin{proof}
For contradiction, suppose this is not the case. Then there must exist a chain of infinitely many box-nodes $u_1, u_2, \ldots$.
Clearly, after $i > M$ levels of recursion for some large enough integer $M>0$, the set of distance-nodes in $L$ will no longer change,  since otherwise $L$ will either 
become empty and therefore the algorithm stops, or contain every node in $T_{p}$ and become stable.
Let $v_1, v_2, \ldots, v_m$ be the set of unchanging distance-nodes in $L$. 
Since at each level of recursion, $B(u_{i+1})$ always halves the edge length of $B(u_{i})$ and is contained inside $B(u_{i})$ for each $i$,
$B(u_1), B(u_2), \ldots$, will eventually converge to a single point, say $p$, in $\mathbb{R}^d$.
If $p$ is not coincident with any of $l(v_1), l(v_2), \ldots, l(v_m)$, say $p \neq l(v_1)$, then since the sizes of $B(u_1), B(u_2), \ldots$, approach to zero,  the distance between each box $B(u_{i})$ and $l(v_1)$ converges to $\norm{p - l(v_1)}>0$.  After a sufficient number of recursion
levels, $B(u_i)$ will become small, comparing to the distance between $B(u_i)$ and $l(v_1)$,
and thus $v_1$ will be removed from $L$ in Step 2 of {\bf Algorithm \ref{alg-2}}. This is a contradiction.

Thus, the only remaining possibility is that $m = 1$ and $p$ is coincident with $l(v_1)$
({\it i.e.}, $l(v_1) = p$).
Since $r_c$ will no longer change after $i$ levels of recursion,
and the sizes of $B(u_1), B(u_2), \ldots$, and their distances to $l(v_1)$ all approach to zero,
there must be some $u_l$ such that the condition in Step 4.1 of {\bf Algorithm \ref{alg-2}} is satisfied,
which will result in the removal of $v_{1}$ from $L$ or the algorithm stops. 
This is a contradiction.
\qed
\end{proof}

The next lemma shows a property of the distance-tree $T_{p}$ that will be used in the proof of Theorem \ref{the-time}.

\begin{lemma}
\label{lem-wspd}
Let $v$ be any node in $T_{p}$ other than the root, and $r$ be the minimum distance between any input point in $P_v$ and any input point in $P\setminus P_{v}$.
Let $v'$ be the parent of $v$ in $T_{p}$.  Then $s(v') \leq 2nr$.
\end{lemma}

\begin{proof}
Let $r_G$ be the minimum length of any edge in the graph $G(W)$ connecting an input point in $P_v$ to an input point in $P \setminus P_v$.
Since $G(W)$ is a 2-spanner for $P$, $r_G \leq 2r$.
By {\bf Algorithm \ref{alg-1}},
we know that the parent node $v'$ of $v$ (and $P_{v'}$) is formed by a sequence of no more than $n$ merge operations on the nodes of $T_{p}$.
The last one of these operations extracts an edge connecting some input point in $P_v$ to some input point in $P \setminus P_v$,
whose length is no larger than $r_G$.
Each merge operation contributes to $s(v')$ a value no bigger than $r_G$, since the edge $e$ extracted from the min-priority queue $Q$ by {\bf Algorithm \ref{alg-1}} has a length $w(e)$ no larger than $r_G$.
Hence, $s(v') \leq nr_G \leq 2nr$.
\qed
\end{proof}

\begin{theorem}
\label{the-time}
For any set of $n$ input points in $\mathbb{R}^{d}$ and an influence function $F$ satisfying the three properties in Section \ref{sec-inf}, the AI-Decomposition algorithm yields $O(n \log n)$ type-1 and type-2 cells in $O(n\log n)$ time, where the constants hidden in the big-$O$ notation depend on the error tolerance $\beta$ and $d$.
\end{theorem}

\begin{proof}
To prove this theorem, we need to bound only the running time since the total number of cells
cannot be larger than the running time.

We first introduce the following two definitions. A box-node $u$ is called the {\em current box-node} if {\bf Algorithm \ref{alg-2}} is executing on $u$. 
A box-node $u$ is said to {\em refer to} a distance-node $v$ if $v$ ever appears in $L$ while executing Step 1 to Step 3 of {\bf Algorithm \ref{alg-2}} on $u$, and equivalently,
$v$ is called a {\em reference} of $u$. Note that since  $v$ may not be removed from $L$ when  $u$ is the current box-node, it is possible that $u$ and its children or descendants all refer to $v$. 

By {\bf Algorithm \ref{alg-2}}, we know that the execution time on a box-node $u$ (excluding the time taken by the recursive calls) is linear in terms of the number  $ref_{u}$ of its references ({\em i.e.,} the number of distance-nodes in $L$ when $u$ is the current box-node), and therefore the running time of {\bf Algorithm \ref{alg-3}} is linear in the summation of the numbers of references over all box-nodes generated by the algorithm.
This means that to prove the theorem, it is sufficient to count the total number of references, $\sum_{u} ref_{u}$. By the linearity of the summation, we know $\sum_{u} ref_{u} = \sum_{v} refd_{v}$, where $refd_{v}$ is the number of box-nodes which refer to a distance-node $v$ during the entire execution time of {\bf Algorithm \ref{alg-3}}. Thus, if we can prove $refd_{v}=O(\log n)$,
then we immediately have the desired $O(n\log n)$ time bound for the theorem because there are only $O(n)$ distance-nodes in $T_{p}$.
Below, we show that $refd_{v}=O(\log n)$ is indeed true for any distance-node $v$.

To show $refd_{v}=O(\log n)$, we first consider the case that $v$ is the root of $T_{p}$. In this case,  $v$ is referred to only once, by the root of the box-tree $T_{q}$, and the statement is trivially true.
Thus, we assume that $v$ is an arbitrary distance-node other than the root of $T_p$.

To bound $refd_{v}$, we first observe that if a box-node $u$ refers to $v$,
then either all or none of $u$'s children refers to $v$ (the latter case happens if $v$ is removed from $L$ when $u$ is the current box-node). This means that we only need to count those box-nodes $u$ which refer to $v$ and have $v$ 
remove from $L$ when $u$ is the current box-node. The reason is that although we do not count those box-nodes, say $u'$, which do not remove 
$v$ from their $L$ lists when they become the current box-nodes, the number of box-nodes ({\it i.e.}, the $2^{d}$ children of $u'$) which refer to $v$ at the next level of recursion
increases exponentially. This implies that 
the total number of box-nodes which refer to $v$ but are not counted is no bigger than the total number of box-nodes which are counted. Thus, we can safely ignore those $u'$.
Let $U_{v}$ denote the set of box-nodes $u$ which are counted.

We define a mapping $\Phi$ on $U_{v}$. Let $u'$ be the parent of $u$ in $T_q$ (if existing). $\Phi(u)$ is defined as
\[
\Phi(u) =
\begin{cases}
B(u') & \text{if $u$ is generated in Step 4 of {\bf Algorithm \ref{alg-2}} during the processing of $u'$,} \\
B(u) & \text{otherwise.}
\end{cases}
\]
It is not hard to see that for $u_1, u_2 \in U_{v}$ and $u_1 \not= u_2$, 
$\Phi(u_1)$ and $\Phi(u_2)$ are disjoint.
Let $\mathcal{B} = \{\Phi(u) \mid u \in U_{v}\}$.
It is sufficient to show $\abs{\mathcal{B}} = O(\log n)$. Our strategy is to use Lemma \ref{lem-ann} for counting.
To do this, we prove that there exist boxes $B_{out}$ and $B_{in}$ with edge lengths $s_{out}$ and $s_{in}$ respectively and a constant $c_0$ depending only on $d$ and $\beta$ such that
all of the following hold:
\begin{enumerate}
\item
$B_{out}$ and $B_{in}$ are  co-centered at $l(v)$.
\item
Every box in $\mathcal{B}$ intersects $B_{out}$.
\item
No box in $\mathcal{B}$ is contained entirely in $B_{in}$.
\item
$\frac{s_{out}}{s_{in}}$ is bounded by some polynomial of $n$.
\item
For any $B \in \mathcal{B}$,  $s \geq c_{0} r$, where $r$ is the shortest distance between $B$ and $l(v)$, $s$ is the edge length of $B$, and $c_{0}$ is some positive constant depending on $d$ and $\beta$. 
\end{enumerate}
Clearly, if all of the above hold, then by Lemma \ref{lem-ann}, we have $\abs{\mathcal{B}} = O(\log n)$.

Let $r'$ be the minimum distance between a point in $P_v$ and a point in $P\setminus P_{v}$. 
Observe that by the way $T_{p}$ is built and the property of the well-separated pair decomposition, we have $s(v') \leq 2nr'$ (by Lemma \ref{lem-wspd}), where $v'$ is the parent of $v$ in $T_p$.

We first determine $B_{out}$. Let $v'$ be the parent of $v$ in $T_{p}$. Let $s'$ be the edge length of $E(v')$. We choose $s_{out} = 7s'$, and claim that for every box-node $u$ that 
refers to $v$, $B(u)$ is fully contained inside $B_{out}$.
Let $u'$ be the parent of $u$ such that either $v'$ is removed from $L$ in Step 1 of {\bf Algorithm \ref{alg-2}} when processing $u'$ or 
$u$ is created in Step 4 of {\bf Algorithm \ref{alg-2}} (where $v'$ is also removed from $L$).
Note that $u'$ must exist since these are the only two ways for $v$ to appear in $L$.
If $v'$ is removed from $L$ in Step 1, then we know that $B(u')$ intersects $E(v')$ and has at most twice the edge length of $E(v')$.
Therefore, $B(u')$ is contained entirely inside $B_{out}'$, where $B_{out}'$ is the box centered at $l(v')$ and with an edge length $5s'$.
If $v'$ is removed from $L$ in Step 4, then we know that $B(u)$ intersects $E(v')$ and has an edge length no bigger than that of $E(v')$.  
This means that $B(u)$ is contained inside $B_{out}'$ as defined above.
Thus, in either case, $B(u)$ is fully contained in $B_{out}'$.
Since  $\norm{l(v) - l(v')} \leq s(v') \leq \beta s'/8 \leq s'$, $B_{out}'$ is completely inside $B_{out}$.
Thus, the above claim is true. 

Based on this claim, it is clear that every box in $\mathcal{B}$ intersects $B_{out}$, whose edge length is 
$s_{out}= 7s' = \frac{56s(v')}{\beta} \leq \frac{112n}{\beta}r'$.

Let $\beta_0=\frac{2(1 + \beta)\mathcal{P}(n)}{\beta}$. We choose $s_{in}=\frac{r'}{6\sqrt{d}(1 + \beta_0)}$,
and claim that for every $u$ that refers to $v$, $\Phi(u)$ cannot be completely inside $B_{in}$.
Suppose this is not the case, and there exists such a box-node $u$ whose $\Phi(u)$ is fully contained inside $B_{in}$. 

First of all, it is easy to see that such a box-node $u$ cannot be the root of the box-tree $T_{q}$, since otherwise, $B(u)$ should be contained inside $B_{in}$. 
(Note that in this case, $\Phi(u)=B(u)$.) 
But this cannot be true, as $B(u)$ contains 
all input points and its size is obviously larger than that of $B_{in}$.

Next, we show that such a box-node $u$ (i.e., whose $\Phi(u)$ is inside $B_{in}$) is not generated in Step 4 of {\bf Algorithm \ref{alg-2}} when processing $u$'s parent $u'$ in $T_{q}$.
Suppose, for contradiction, $u$ is generated in Step 4. Let $v'$ be the parent of $v$ in $T_{p}$. Then $E(v')$ does not fully contain $B(u')$, since
otherwise $v'$ would have been deleted from $L$ in Step 1 of {\bf Algorithm \ref{alg-2}}, instead of Step 4, when processing $u'$. Note that since $v'$ contains
at least one input point that is not in $P_v$, the diameter of $E(v')$ must be greater than $r'$. This means that $E(v')$ is at least 6 times larger
than $B_{in}$ in edge length. The distance between $l(v)$ and $l(v')$ ({\em i.e.,} the centers of $B_{in}$ and $E(v')$, respectively) 
satisfies the inequalities $\norm{l(v) - l(v')} \leq s(v') \leq R\frac{\beta}{8} \leq \frac{R}{16}$, where
$R$ is the edge length of $E(v')$. This means that $B_{in}$ is fully contained in $E(v')$,
and therefore cannot contain $B(u')$, which is $\Phi(u)$. This is a contradiction, and thus $u$ cannot be generated in Step 4.

Finally, we show that $u$ cannot be generated in Step 5 of {\bf Algorithm \ref{alg-2}}. Suppose $u$ is generated in Step 5 by a quad-tree decomposition on $B(u')$, where $u'$ is the parent 
of $u$ in $T_{q}$. Since  
$B(u)=\Phi(u)$ is contained in $B_{in}$ (by assumption), we know that $B(u')$, which contains $B(u)$ and has an edge length twice of that of $B(u)$, must be contained in a box $B_{in}'$ centered at $l(v)$ and with an edge length
$\frac{r'}{2\sqrt{d}(1 + \beta_0)}$.
This means $D(u') \leq \frac{r'}{2(1 + \beta_0)}$, where $D(u')$ is the diameter of $B(u')$. Let $r''$ be the distance between $l(v)$ and $B(u')$.  Then, by the fact that $B_{in}'$ contains $B(u')$, we have $r'' \leq \frac{r'}{2\sqrt{d}(1 + \beta_0)}$.
Combining the above two inequalities, we get $r'' + D(u') \leq \frac{r'}{(1 + \beta_0)}$. 
For any point $p \in P\setminus P_{v}$, let $r_{p}$ be the distance between $p$ and $B(u')$, and $q'$ be the closest point on $B(u')$ to $p$. Then by the triangle inequality, we know that the distance $\norm{l(v)-q'}$ between $l(v)$ and $q'$ is no larger than $r''+D(u')$. Thus, we have
$\norm{l(v)-q'} \leq \frac{r'}{(1 + \beta_0)}$. Also, by the definition of $r'$, we know that the distance $\norm{p-l(v)}$ between $p$ and $l(v)$ is no smaller than $r'$.
By the triangle inequality (in the triangle $\Delta l(v)pq'$), we know $r_{p} =\norm{p-q'}\geq \norm{p-l(v)}-\norm{l(v)-q'}\geq r'-\frac{r'}{(1 + \beta_0)}=
\frac{\beta_0 r'}{(1 + \beta_0)}$.
Therefore, we have $\frac{r'' + D(u')}{r_p} \leq \frac{1}{\beta_0} = \frac{\beta}{2(1 + \beta)\mathcal{P}(n)}$.
This implies $\frac{D(u')}{r_p} \leq \frac{1}{\beta_0} \leq \frac{\beta}{2}$. 
Since the above inequality holds for every point in $P \setminus P_{v}$, this indicates that every such point must be recorded for $u'$ (see  Step 2 of {\bf Algorithm \ref{alg-2}}).  
By {\bf Algorithm \ref{alg-2}}, we know that $r_c$ stores the minimum recorded distance.
Also, note that a point in $P$ is recorded for $u'$ if and only if it is in $P \setminus P_{v}$.
Therefore, some $p \in P \setminus P_{v}$ gives rise to the recorded distance $r_c$.
By Lemma \ref{lm:distance}, we know $r_p \leq (1 + \beta)r_c$.
Thus, we have $\frac{r'' + D(u')}{r_c} \leq \frac{\beta}{2\mathcal{P}(n)}$.
Since each point $p \in P \setminus P_{v}$ is recorded for $u'$ and  $v$ is referred to by $u$ ($u$ is a child of $u'$), it must be the case that after finishing Step 2 of {\bf Algorithm \ref{alg-2}}
in the recursion for $u'$, $v$ is the only distance-node in $L$. Then, by the fact of $\frac{r'' + D(u')}{r_c} \leq \frac{\beta}{2\mathcal{P}(n)}$, we know that $u'$ will be processed in Step 4, which includes the generation of the node $u$, instead of Step 5.
This is a contradiction.

Summarizing the above three cases, we know that every box in $\mathcal{B}$ is not fully contained in $B_{in}$.

 From the above discussion, we know that the edge lengths of $B_{out}$ and $B_{in}$ satisfy the following inequality
\[
	\frac{s_{out}}{s_{in}} \leq \frac{672 \sqrt{d} n (1 + \frac{2(1 + \beta)\mathcal{P}(n)}{\beta})}{\beta}.
\]
This means that the ratio of $\frac{s_{out}}{s_{in}}$ is bounded by a polynomial of $n$.

The only remaining issue now is to show that for any $u \in U_{v}$, the edge length $s$ of $\Phi(u)$ and the distance $r$ 
between $\Phi(u)$ and $l(v)$ satisfy the relation of $s \ge c_0 r$ for some constant $c_0>0$.
Note that such a relation is trivially true for any $c_0$ if $u$ is the root of $T_{q}$, 
since in this case $B(u)=\Phi(u)$ contains all input points and the distance $r$ is $0$
({\em i.e.}, the distance of $B(u)$ to $l(v)$ is $0$).  Hence, we assume below that $u$ is not the root of $T_q$ and has a parent  $u'$ in $T_{q}$.

For any box-node $u_{0} \in T_{q}$ and any distance-node $v_{0} \in T_{p}$, let 
$r(u_0,v_0)$ be the shortest distance between $B(u_0)$ and $l(v_0)$.  We consider two possible cases. 
\begin{enumerate}
\item  $u$ is generated in Step 4 when processing $u'$. In this case, $\Phi(u) = B(u')$. Let $v'$ be the parent of $v$ in $T_{p}$.
We consider two possible sub-cases, depending on whether $E'(v')$ intersects $B(u')$ (see {\bf Algorithm \ref{alg-1}} for the definition of $E'(v')$).
\begin{enumerate}
\item $E'(v')$ intersects $B(u')$. In this sub-case, since $v'$ is not removed from $L$ in Step 1 when processing
$u'$,  some part of $B(u')$ must be outsides $E(v')$. ($E'(v')$ is co-centered at $l(v')$ with $E(v')$ and is of half the edge length of $E(v')$. If $B(u')$ is fully inside $E(v')$, then an edge length of $B(u')\cap E(v')$ will be larger than half
the edge length of $B(u')$, and hence $v'$ will be removed from $L$ in Step 1.)  This means that the edge length of $B(u')$ is at least half the edge length
of $E'(v')$, which is $\frac{2s(v')}{\beta}$. Thus, the diameter $D(u')$ of $B(u')$ exceeds $\frac{2\sqrt{d}s(v')}{\beta}$. 
Furthermore, since $E'(v')$ intersects $B(u')$, we have $r(u',v') \leq \frac{2\sqrt{d}s(v')}{\beta}$ (by the definition of $r(u',v')$ and the size of $E'(v')$).  Also, since $P_{v'}$ contains both $l(v)$ and $l(v')$, the distance between $l(v)$ and $l(v')$ is upper-bounded by the diameter $s(v')$ of $P_{v'}$, i.e.,  $\norm{l(v) - l(v')} \leq s(v')$. Thus, we have 
$r(u',v) \leq \norm{l(v) - l(v')} + r(u',v') \leq s(v') + \frac{2\sqrt{d}s(v')}{\beta}
\leq \frac{4\sqrt{d}s(v')}{\beta}$. Therefore, we have $edgeLength(B(u')) \ge c_0 r(u',v)$ if we choose 
$c_0 \leq \frac{1}{2\sqrt{d}}$.

\item $E'(v')$ does not intersect $B(u')$. In this sub-case, we have $r(u',v') \geq \frac{2s(v')}{\beta}$ (by the fact that $E'(v')$ is centered at $l(v')$ and with an edge length of $\frac{4s(v')}{\beta}$).  Since $v'$ is not removed from $L$ in Step 2 when processing $u'$,  the diameter $D(u')$ of $B(u')$ must exceed
$r(u',v')\frac{\beta}{2}$.
Note that $s(v') \leq \frac{2s(v')}{\beta}$,
and thus $s(v') \leq r(u',v')$.
Then $r(u',v) \leq \norm{l(v) - l(v')} + r(u',v') \leq 2r(u',v')$.
This means that the diameter $D(u')$ of $B(u')$ exceeds
$r(u',v')\frac{\beta}{2} \geq r(u',v)\frac{\beta}{4}$.
From this, we immediately know $edgeLength(B(u')) \ge c_0 r(u',v)$ if $c_0 \leq \frac{\beta}{4\sqrt{d}}$.

\end{enumerate}

\item $u$ is generated in Step 5 when processing $u'$. In this case, $\Phi(u) = B(u)$. Let $v'$ be the distance-node in $L$ when processing $u'$ which is either an ancestor of $v$ in $T_p$ or $v$ itself.
For this case, we also consider two possible sub-cases, depending on whether $E'(v')$ intersects $B(u')$. 
\begin{enumerate}
\item $E'(v')$ intersects $B(u')$. In this sub-case, by exactly the same argument given above for Case 1(a), we know that 
the diameter $D(u')$ of $B(u')$ is at least $\frac{r(u',v)}{2}$.
Then, $ r(u,v) \leq D(u') + r(u',v) \leq 3D(u') $. Also, note that
$D(u') = 2D(u)$. Thus, $D(u) \geq \frac{r(u,v)}{6}$. In this sub-case, we can choose $c_0 \leq \frac{1}{6\sqrt{d}}$.

\item $E'(v')$ does not intersect $B(u')$.
By the same argument given above for Case 1(b), we know
$D(u') \geq r(u',v)\frac{\beta}{4}$.
Thus, $r(u,v) \leq D(u') + r(u',v) \leq \frac{4 + \beta}{\beta}D(u')$.
Since $D(u') = 2D(u)$, we have $D(u) \geq \frac{\beta r(u,v)}{8 + 2\beta}$.
This means that we can choose $c_0 \leq \frac{\beta}{(8 + 2\beta)\sqrt{d}}$.
\end{enumerate}
\end{enumerate}

Based on the above discussion, we know that if we choose $c_0$ as the minimum of the four possible choices, we have the desired bound $s \ge c_0 r$ for the edge length $s$ of each box in $\mathcal{B}$.
This means that the theorem then follows from Lemma \ref{lem-ann}.
\qed
\end{proof}

\section{Vector CIVD}
\label{sec-ex1}

In this section, we show that the AI decomposition can be combined with an assignment algorithm to compute a $(1-\epsilon)$-approximate CIVD for  the vector CIVD problem.
We first give the problem description and show that its influence function satisfies the three
properties given in Section \ref{sec-inf}. We then present our assignment algorithm. An overview of the 
assignment algorithm is given in Section \ref{sec-vecov}.

\subsection{Problem Description and Properties of the Influence Function}

Let $P$ be a set of $n$ points in $\mathbb{R}^d$ and $F$ be the influence function. For each point $p \in P$ and a query point $q$ in $\mathbb{R}^d$, the influence
 $F(\{p\},q)$ is a vector in the direction of $p-q$ (or $q-p$) and with a magnitude of $\norm{p-q}^{-t}$ for some constant $t \ge1$.
Such a vector may represent force-like influence between objects, such as the
gravity force between planets and stars (with $t=d-1$) or electric force between physical bodies like electrons and protons (with $t=2$).
For a cluster site $C$ of $P$, the influence  from $C$ to a query point $q$ is the vector sum of the individual influence  from each point of $C$ to $q$, {\em i.e.,} $F(C,q)=\sum_{p \in C} F(\{p\},q)$.
Note that for ease of discussion, in the remaining of this section, we also use $F(C,q)$ to denote the magnitude of the influence  ({\em i.e.,} $F(C,q) = \norm{F(C,q)} = \norm{\sum_{p \in C} F(\{p\},q)}$) when there is no ambiguity about its direction. 
The vector CIVD problem is to partition the
space $\mathbb{R}^d$ into Voronoi cells such that each cell is the union of all points whose maximum influence comes from the same cluster site of $P$ (see Fig.~\ref{fig-exact} for examples of the exact vector CIVD in $\mathbb{R}^{2}$).

\begin{figure}
\center
\includegraphics[height=3.3in]{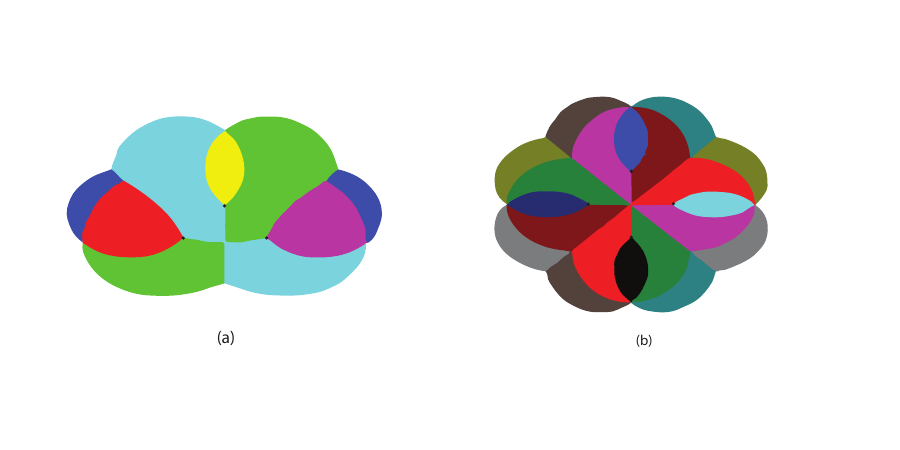}
\caption{Examples of the exact vector CIVD with $t=2$ in $\mathbb{R}^{2}$, where the regions with the same color form a Voronoi cell:
(a) The vector CIVD of 3 input points; (b) the vector CIVD of 4 input points.}
\label{fig-exact}
\end{figure}

Our objective for vector CIVD is to obtain a $(1-\epsilon)$-approximate CIVD, in which each cell $c$ is associated with a cluster site whose influence to every point $q\in c$ is no smaller than $(1-\epsilon)F_{max}(q)$. 
To  make use of the AI decomposition, we first show that the vector CIVD problem satisfies the three properties in Section \ref{sec-inf}.

\begin{theorem}
The vector CIVD problem satisfies the three properties in Section \ref{sec-inf} for any constant $t \geq 1$.
\end{theorem}

\begin{proof}
We first prove Property \ref{pro-1}. Consider a set of vectors in $\mathbb{R}^d$, $A = \{a_{1}, a_{2}, \ldots, a_{m}\}$,
such that $(\{a_{1}+q, a_{2}+q, \ldots, a_{m}+q\}, q)$ is a maximal pair for some $q$.  Note that $p_{i}=a_{i}+q$ is a point in $P$ and $C=\{a_{1}+q, a_{2}+q, \ldots, a_{m}+q\}$ is a cluster site of $P$.
Let $b_{i}$ be the vector that has the same direction as $a_{i}$ and a length $\norm{a_{i}}^{-t}$ ({\em i.e.}, $b_{i}$ is the influence from $p_{i}$ to $q$).
Let $B = \{b_{1},b_{2},\ldots, b_{m}\}$.
We assume that $\norm{\sum_{i=1}^m b_{i}} = K$. Let $\epsilon_{i} =(\underbrace{0, \ldots,0, 1}_{i}, 0, \ldots, 0), i = 1,2, \ldots,d$, be a standard basis of the $\mathbb{R}^d$ space.  Then
each $b_{i}$ can be written as a linear combination of the basis, $c_{i1}\epsilon_{1}+c_{i2}\epsilon_{2}+\cdots+c_{id}\epsilon_{d}$.

We claim that for every $j$, $\sum_{i=1}^{m} c_{ij} \leq 2K$. To prove this claim by contradiction, we assume that there exists some $j$ such that $\sum_{i=1}^{m} c_{ij} > 2K$. Consider two subsets $B_{+}$ and $B_{-}$ of $B$, where $B_{+} = \{b_{i} \mid c_{ij}>0\}$ and $B_{-} = B\setminus B_{+}$. Since $\sum_{i=1}^{m} c_{ij} > 2K$, we have
$\sum_{i:b_{i}\in B_{+}}c_{ij}+\sum_{i:b_{i}\in B_{-}}c_{ij}>2K$. This means that $\abs{\sum_{i:b_{i}\in B_{+}}c_{ij}}>K$ or $\abs{\sum_{i:b_{i}\in B_{-}}c_{ij}}>K$. Without loss of generality,  we assume that the latter case occurs.
Then, we have $\norm{\sum_{i:b_{i}\in B_{-}}b_{i}}>K$. This implies that the influence from the subset $\{p_{i} \ | \ b_{i} \in B_{-}\}$ to $q$ is larger than $K$. But this contradicts with
the fact that $(C,q)$ is a maximal pair.  

Therefore, we have  $\sum \norm{b_{i}} \leq \sum_{i}\sum_{j} \norm{c_{ij}} \leq 2dK$. If we change every $b_{i}$ by adding a vector with a length not larger than $\epsilon' \norm{b_{i}}$ for some small constant $\epsilon'>0$, then the total change will be no larger than $\epsilon' \sum\norm{b_{i}} \leq 2\epsilon' dK$.

Now consider what happens if we change each $a_{i}$ by adding a vector with a length smaller than $\epsilon\norm{a_{i}}$. It can be verified that, with a
sufficiently small $\epsilon$, the corresponding $b_{i}$ will be changed by no more than $O(1)(1-(1-\epsilon)^t)\norm{b_{i}}$, and therefore the sum of $b_{i}$ will change by no more than $O(1)2d(1-(1-\epsilon)^t)K$. This proves Property \ref{pro-1}.

To prove Property \ref{pro-2}, consider a point $q \in \mathbb{R}^d$ and
a subset of input points
$P'$ such that there exists $a \in P'$ satisfying the inequality $n^{\frac{1}{t}}\cdot\norm{q-a} < \epsilon\norm{q-a'}$ for every $a' \in P \setminus P'$,
where $0 < \epsilon < 1$ is a small constant and $n$ is the number of input points.
Then, we have $\norm{q - a'}^{-t} \leq \epsilon^t\cdot\norm{q-a}^{-t}/n \leq \epsilon F(\{a\},q)/n$ for every $a' \in P \setminus P'$. 
Since $\abs{P \setminus P'} \leq n$ and the maximum influence $F_{max}(q)$ of $q$  is clearly no smaller
than $F(\{a\},q)$, we immediately know that
the maximum  influence from any subset of $P'$
is smaller than $F_{max}(q)$ by at most $\sum_{a' \in P \setminus P'} \norm{q-a'}^{-t} \leq \epsilon F(\{a\},q)$, and is therefore no smaller than $(1-\epsilon)F_{max}(q)$.

Property \ref{pro-3} is obvious since after a rotation about $q$ or a scaling, for any point $p \in P$, $F(\{p\},q)$ is changed by a factor that depends only on the rotation or scaling itself.\qed
\end{proof}

The above theorem implies that the AI decomposition can be applied to the vector CIVD problem.
We assume that  $\beta$ in the AI decomposition is set to $\Delta^{-1}(\epsilon)$,
where $\epsilon$ is the error tolerance in the vector CIVD  and $\Delta$ is the error estimation function for the problem.

\subsection{Overview of the Assignment Algorithm}
\label{sec-vecov}

As discussed in Section \ref{sec-AI}, the AI decomposition only gives a space partition; an assignment algorithm is still needed to determine an appropriate cluster site for each Voronoi cell.
By Theorem \ref{the-aid} and {\bf Algorithm \ref{alg-2}}, we know that  each type-1 cell is dominated by a distance-node $v$, and $P_{v}$ (or a subset of $P_{v}$) is its approximate maximum influence site. Thus, we only need to consider those type-2 cells.
By Theorem \ref{the-aid}, we know that to determine an approximate maximum influence site for a type-2 cell $c$, it is 
sufficient to pick an arbitrary point $q \in c$ and find a cluster site which gives $q$ the maximum influence.

To assign a cluster site to a query point $q$ in a type-2 cell, our main idea is to transform the assignment problem to an {\em optimal hyperplane partition} (OHP)
problem, which uses a hyperplane passing through $q$ to partition the input points so as to identify the maximum influence site of $q$. Optimally solving the OHP problem in a straightforward
manner takes $O(n^{d})$ time. To improve the running time, our idea is to significantly reduce the number of input points involved
in the OHP problem. Our main strategy for reducing the number of input points involved is to perturb the aggregated input points so that
each aggregated point cluster is mapped to a single point. Also, those input points that are far away from $q$ and have little influence on $q$ are ignored. In this way, we can reduce the number of input points  from $n$ to $O(\log n)$. 
A quad-tree decomposition based {\em aggregation-tree} $T$  is built to help identify those point clusters that can be perturbed.
The to-be-perturbed point clusters  form an {\em effective cover} in the aggregation-tree $T$. Straightforwardly computing the effective cover takes $O(n)$ time. To improve the time bound, we first present a slow method called {\em SlowFind} to shed some light on how to speed up the computation. 
The main obstacle is how to avoid recursively searching on a {\em long path} (with a possible length of $O(n)$) in the aggregation tree. To overcome this long-path difficulty, we use a number of techniques, such as the
majority path decomposition, to build some auxiliary data structures for $T$ so that we can perform binary search on such long path and therefore speed up the computation from $O(n)$ time
to $O(\log^{2} n)$. Combining this with a key fact that the effective cover has a size of $O(\log n)$, we obtain an assignment algorithm which assigns a $(1-\epsilon)$-approximate maximum influence site to any type-2 cell in $O(\log^{\max\{2,d\}}n)$ time.

\subsection{Assignment Algorithm}

To develop the assignment algorithm, we first give the following key observation.

\begin{observation}
\label{obs-h}
In the vector CIVD problem, if a subset $C$ of $P$ is the maximum influence site of a query point $q$, then there exists a hyperplane $H$ passing through $q$ such that all points of $C$ lie on one side of $H$ and all points of $P \setminus C$ lie on the other side of $H$.
\end{observation}
\begin{proof}
Consider the hyperplane $H$ that passes through $q$ and is perpendicular to the influence (vector) $F(C,q)$ from $C$ to $q$. If there is an input point $p \not\in C$ that lies on the same side of $H$ as $C$ (which is the side of $H$ pointed by $F(C,q)$), then adding $p$ to $C$ will only increase the magnitude of the influence. If there is an input point of $C$ lying on the side of $H$ opposite to the influence's direction, then deleting this point from $C$ will only increase the magnitude of the influence. Thus the observation is true.
 \qed
\end{proof}

The above observation suggests that to find $C_{m}(P,q)$ for a query point $q$, we can try all possible partitions of $P$ by using hyperplanes passing through $q$ and pick the best partition. We call this problem the {\em optimal hyperplane partition} (OHP) problem.  Since there are $n$ input points, we may need to consider a total of $O(n^d)$ such hyperplanes in order to optimally solve the problem.  Thus straightforwardly solving this problem could be too costly. 
To obtain a faster solution, our idea is to treat those aggregating points as a single point so as to reduce the total number of points that need to be considered for the sought hyperplane.

To implement this idea, we first build a tree structure $T$ called {\it aggregation-tree}, in which each node is associated with a set of input points. {\bf Algorithm \ref{alg-4}} below generates the aggregation-tree $T$.(See Also \textbf{Figure} \ref{fig-treebuild}.)

\begin{algorithm}
\caption{Tree-Build$(v, R(v))$}
\textbf{Input:} A node $v$ of the aggregation-tree $T$, together with the bounding box $R(v)$ of
its associated input points.\\
\textbf{Output:} A subtree of $T$ rooted at $v$.
\label{alg-4}
\begin{algorithmic}[1]

\State{If $v$ contains only one input point, return.}

\State{Quad-tree decompose $R(v)$ into $2^{d}$ smaller boxes $R'(\cdot)$.}

\State{Create nodes $v_1,v_2,\ldots,v_l$ as the children of $v$ in $T$, each child corresponding to a smaller box $R'(v_{i})$ containing at least one input point of $v$.}

\State{For each $1 \leq i \leq l$, let $R(v_i)$ be the smallest hypercube box containing all points in $v_i$, $S(v_{i})$ be the edge length of $R(v_{i})$, and $L(v_{i})$ be a representative point of $v_{i}$.}

\State{For each $i$, call Tree-Build$(v_i, R(v_{i}))$.}
\end{algorithmic}
\label{algo-build}
\end{algorithm}

\begin{figure}
\center
\includegraphics[height=2.0in]{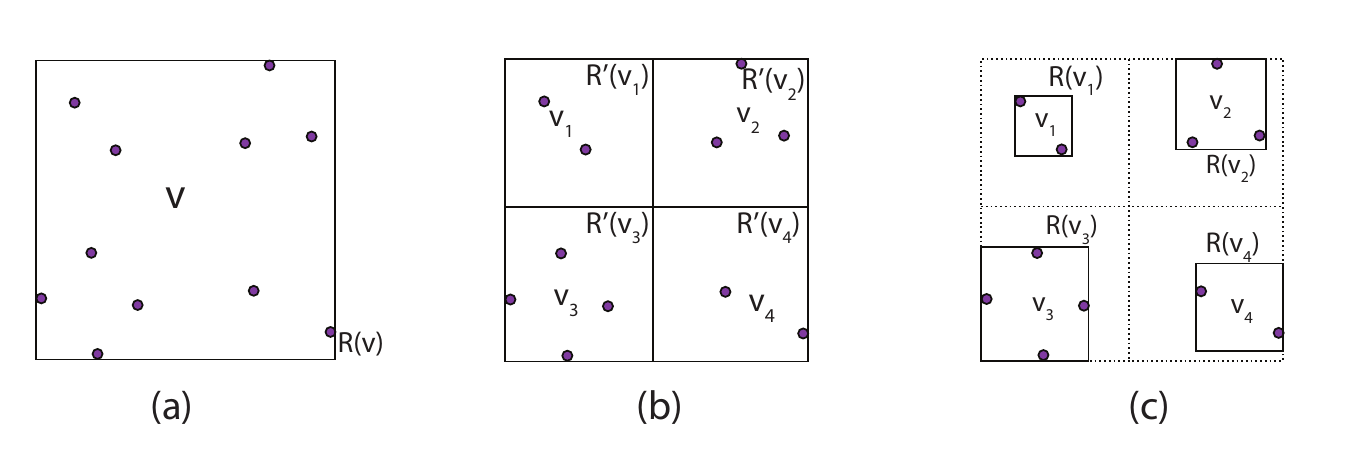}
\caption{Examples of first few steps of building an aggregation tree.
It shows how children of $v$ are determined by quad-tree decomposition and shrinking.}
\label{fig-treebuild}
\end{figure}

To build the whole aggregation-tree $T$, we simply run {\bf Algorithm} Tree-Build$(v_{r}, R(v_{r}))$, where $v_{r}$ is a (root) node constructed for representing $P$ and $R(v_{r})$ ($R'(v_{r})$ as well)  is the smallest bounding box 
of $P$. Let $S(v_{r})$ denote the edge length of $R(v_{r})$. For each node $v$ of $T$, let $v$ also denote the set of input points associated with the node $v$ and $\abs{v}$ denote its cardinality.

In the aggregation-tree $T$, we may view all input points in some node $v$ as $\abs{v}$ coincident points at its representative point $L(v)$. In this way, we reduce  the total number of points that need to be considered for the optimal hyperplane partition problem.  (Later, we will discuss how to identify such nodes $v$ in $T$.) 

Let $c$ be a type-2 cell produced by the AI decomposition, and $q$ be an arbitrary point in $c$. The following lemma enables us to bound the error incurred by viewing all input points in a node of the aggregation-tree $T$ as a single point.

\begin{lemma}
\label{lem-mgc}
Let $\psi$ be a perturbation on a set (possibly multiset) $P'$ of input points with a witness point $q$ in a type-2 cell $c$ and an error ratio $\frac{\Delta^{-1}(\epsilon)}{3}$.
Let $C \subseteq P'$ be a cluster site such that $(C,q)$ is a stable pair and has influence $F(C,q) \geq (1 - \Delta^{-1}(\epsilon))F(C_{m}(P',q),q)$.
Then for any point $q'$ in $c$, $F(\psi(C),q') \geq (1 - \epsilon) F(C_{m}(\psi(P'), q'),q')$.
\end{lemma}

\begin{proof}
For every point $p \in C$, consider the difference between the two vectors, $\psi(p) - q$ and $p - q'$. By the perturbation $\psi$, we have $\norm{\psi(p) - p}/\norm{\psi(p) - q} \leq  \Delta^{-1}(\epsilon)/3$. Since $c$ is a type-2 cell, by Lemma \ref{lm:cell2}, we also have $\norm{q-q'} \leq 2r_{min}\Delta^{-1}(\epsilon)/3 \leq 2\norm{\psi(p)-q}\Delta^{-1}(\epsilon)/3$, where $r_{min}$ is the distance from $q$ to any input point in $C$. Combining the above two inequalities,  we get $\norm{(\psi(p) - q) - (p - q')} \leq \Delta^{-1}(\epsilon)\norm{\psi(p)-q}$.
Since $F(C,q) \geq (1 - \Delta^{-1}(\epsilon))F(C_{m}(P',q),q)$, by Lemma \ref{lem-app} and the fact that $F$ is invariant under translation, 
we have $F(\psi(C),q') \geq (1 - \epsilon) F(C_{m}(\psi(P'), q'),q')$.
\qed
\end{proof}

Based on Lemma \ref{lem-mgc}, we can assign an approximate maximum influence site to a type-2 cell $c$ using the following approach. 
\begin{enumerate}
\item  Take an arbitrary point $q_{c}$ in $c$. 
\item Identify a set of pairwise disjoint subsets/nodes $\{v_{1},v_{2},\ldots,v_{m}\}$ in the aggregation-tree $T$ satisfying the condition of $S(v_i) \leq \Delta^{-1}(\epsilon)\norm{q_{c}- L(v_i)}/(3d)$. 
\item Define a perturbation $\psi: \mathbb{R}^{d}\rightarrow \mathbb{R}^{d}$ 
which maps each point $p$ in $v_{i}$ to $\psi(p) = L(v_{i})$ for every $i = 1,2,\ldots, m$. Let $P' = \psi(P)$. 
\item Find a subset $C' \subseteq P'$ so that $F(C', q_{c}) \geq (1 - \Delta^{-1}(\epsilon))F(C_{m}(P',q_{c}),q_{c})$. 
\item 
Map $C'$ back to $C$.
\end{enumerate}

In the above approach,  $C'$ is determined by solving the optimal hyperplane partition problem on $P'$ and $q_{c}$. 
Since $\psi$ maps all points in each $v_{i}$ to a single point $L(v_{i})$, the total number of distinct points in $P'$ is significantly reduced from that of $P$. 

The number of distinct points in $P'$ could still be too large even after  the perturbation. 
To further reduce the size of $P'$, we consider those points far away from $q_{c}$. 
Particularly, we consider a point $p' \in P'$ whose distance to $q$ is 
at least $r_s = (\Delta^{-1}(\epsilon))^{-1/t}n^{1/t}r_{min}$, where $r_{min}$ is the shortest distance from $q_c$ to $P'$. Let $p_{min}$ be the point in $P'$ which has the closest  distance to $q_{c}$. 
Since $F(\{p',\},q_c) \leq \Delta^{-1}(\epsilon)F(\{p_{min}\},q_c)/n$ and the number of such points $p'$ is smaller than $n$,  the influence of any set of such points $p'$ is no bigger than $\Delta^{-1}(\epsilon)F(\{p_{min}\},q)$, and hence is also smaller than $\Delta^{-1}(\epsilon)F(C_{m}(P',q_c),q_{c})$.  This means that we can remove all such far away points from $P'$ before searching  for $C'$ in $P'$.

Below we discuss how to efficiently implement the above approach. We start with the following definition.

\begin{definition}
\label{def-ec}
Let $c$ be a type-2 cell of the AI decomposition and $q_{c}$ be any point in $c$. 
A set $V$ = $\{v_1,v_2,\ldots,v_m\}$ of nodes in the aggregation-tree $T$ is called an \emph{effective cover} for $q_{c}$ if it satisfies the following conditions.
\begin{enumerate}
\item
$v_1,v_2,\ldots,v_m$ are pairwise disjoint when viewed as sets of input points.

\item
Let $B$ be the box centered at $q_{c}$ and with an edge length that is at least $4(\Delta^{-1}(\epsilon))^{-1/t}n^{1/t}r_{min}$ and is $O((\Delta^{-1}(\epsilon))^{-1/t}n^{1/t}r_{min})$,
where $r_{min}$ is the shortest distance between $q_c$ and $P$.
The union of $v_1, v_2, \ldots, v_m$ contains all points in $P \cap B$.

\item
$S(v) \leq \Delta^{-1}(\epsilon)\norm{q_{c} - L(v)}/(3d)$ for every $v \in V$.

\end{enumerate}
\end{definition}

See also \textbf{Figure} \ref{fig-ecover}.

\begin{figure}[ht]
\centering
\includegraphics[height=2.5in]{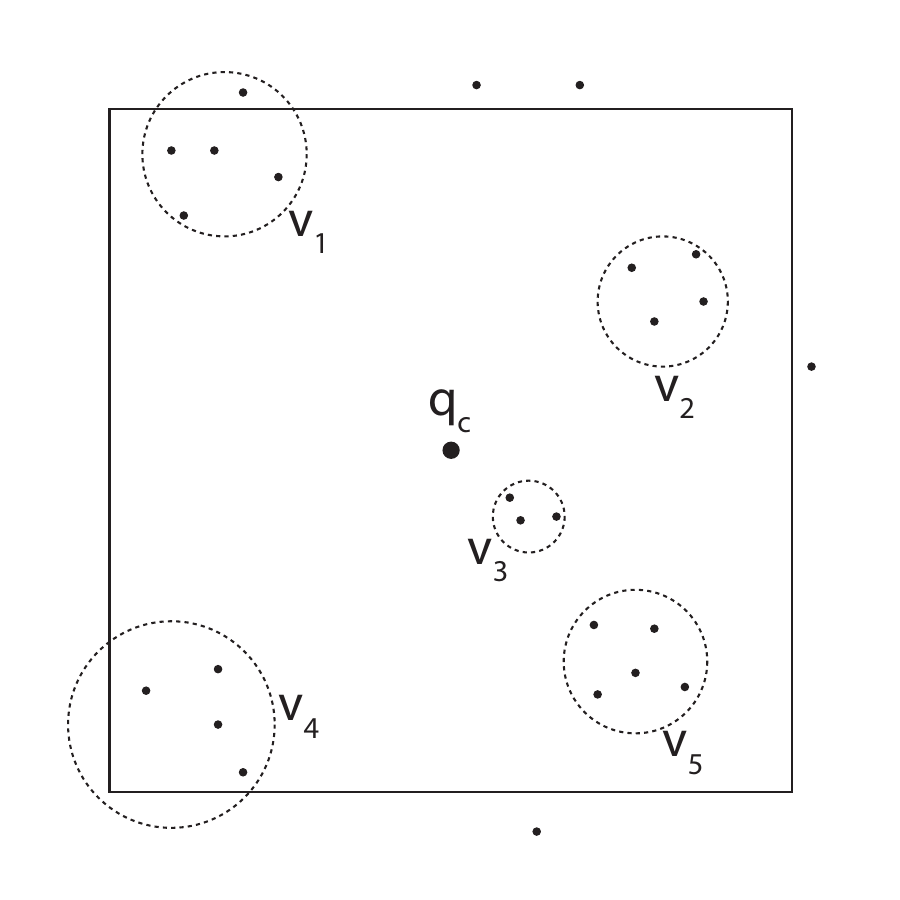}
\caption{An example of effective cover. Instead of considering all input points in order to find the optimal hyperplane,
we can consider only $v_1,\ldots,v_5$, each viewed as 1 ``heavy" point.
This significantly reduce the time of searching.}
\label{fig-ecover}
\end{figure}

An effective cover $V$ in the aggregation-tree $T$ can be used to find the approximate maximum influence site $C$ for $c$. Below are the main steps of the assignment algorithm; the implementation of Find$(v_{r},q_{c})$ will be discussed later.

\begin{algorithm}
\caption{Assign$(c)$}
\textbf{Input:} A type-2 cell $c$ of the AI decomposition.\\
\textbf{Output:} A set of nodes in the aggregation-tree $T$ whose union forms the approximate maximum influence site for $c$.
\label{alg-5}
\begin{algorithmic}[1]

\State{Pick an arbitrary point $q_{c}$ in $c$.}

\State{Call Find$(v_{r}, q_{c})$ to find an effective cover for $q_{c}$. Let $V = \{v_1,v_2,\ldots,v_m\}$ be the resulted effective cover.}

\State{For each partition of $V$ induced by a hyperplane $H$ passing through $q_{c}$, let  $V' = \{v_{i_1},v_{i_2},\ldots,v_{i_k}\}$ be the subset of $V$ on one side of $H$ and having a larger influence on $q_{c}$.
Let $V_{max}$ be the $V'$ having the largest influence $F(P(V'), q_{c})$ on $q_{c}$ among all possible hyperplane partitions, 
where $P(V')$ is a multiset of points with the following form
\[
\{\underbrace{L(v_{i_1}), L(v_{i_1}), \ldots, L(v_{i_1})}_{|v_{i_{1}}|},
   \underbrace{L(v_{i_2}), L(v_{i_2}), \ldots, L(v_{i_2})}_{|v_{i_{2}}|},
   \ldots,
   \underbrace{L(v_{i_k}), L(v_{i_k}), \ldots, L(v_{i_k})}_{|v_{i_{k}}|}
\}.
\]
}

\State{Output $V_{max}$.}

\end{algorithmic}
\end{algorithm}

The following lemma ensures the correctness of the above assignment algorithm.

\begin{lemma}
\label{lem-assign}
Let $c$ be a type-2 cell of the AI decomposition and $V_{max} = \{v_{i_1},v_{i_2},\ldots,v_{i_k}\}$ be the output of Assign$(c)$.
Let $C = \cup_{v \in V_{max}} v$. Then, $F(C,q_{c}) \geq (1-\epsilon)F(C_{m}(P,q_{c}),q_{c})$ for any point $q_{c} \in c$.
\end{lemma}
\begin{proof}
Let $V = \{v_1,v_2,\ldots,v_m \}$ be the effective cover obtained in Step 2 of  {\bf Algorithm} Assign.
Let $\psi$ be a mapping on $P$ defined as follows.
\[
\psi(p) =
\begin{cases}
L(v) & \text{if $p$ is covered by $V$, \emph{i.e.}, $p \in v$ for some $v \in V$.}\\ 
p & \text{Otherwise.}
\end{cases}
\]
By Definition \ref{def-ec}, we know that $\psi^{-1}$ is a perturbation with an error ratio $\Delta^{-1}(\epsilon)/3$ and a witness point $q_c$, where $\psi^{-1}$ is a loosely defined inverse of $\psi$ which maps $\psi(p)$ back to $p$ for each $p \in P$.

We now show that the output $V_{max}$ of {\bf Algorithm} Assign$(c)$ satisfies the inequality $$F(\psi(U(V_{max})),q_c) \geq (1 - \Delta^{-1}(\epsilon))F(C_m(\psi(P),q_c),q_c), $$
where $U(V_{max}) = \cup_{v \in V_{max}} v$.
Let $r_{min}$ and $r'_{min}$ denote the shortest distances
from $q_c$ to $P$ and $\psi(P)$, respectively, and $p_{min}$ and $p'_{min}$ be $q_{c}$'s closest points in $P$ and $\psi(P)$, respectively.
Then, by the triangle inequality, we have
\begin{eqnarray}
r'_{min} &\leq& \norm{\psi(p_{min}) - q_c} \leq \norm{\psi(p_{min}) - p_{min}} + \norm{p_{min} - q_c}. \label{for-1}
\end{eqnarray}
By Definition \ref{def-ec} and the assumption of $\Delta^{-1}(\epsilon) \leq 1/2$, we know $$\norm{\psi(p_{min}) - p_{min}} \leq \Delta^{-1}(\epsilon)\norm{q_{c} - \psi(p_{min})}/3 \leq
\norm{q_{c} - \psi(p_{min})}/6.$$ 
Then by the triangle inequality, we have $$\norm{p_{min} - q_c} \geq \norm{q_{c} - \psi(p_{min})}  - \norm{\psi(p_{min}) - p_{min}} \geq 5\norm{q_{c} - \psi(p_{min})}/6.$$
 Thus, $$\norm{\psi(p_{min}) - p_{min}} \leq \norm{p_{min} - q_c}/5.$$
Plugging the above inequality into (\ref{for-1}), we have $$r'_{min} \leq \norm{\psi(p_{min}) - p_{min}} + \norm{p_{min} - q_c} \leq (1 + 1/5)\norm{p_{min} - q_c} \leq 2\norm{p_{min} - q_c}
\leq 2r_{min}.$$
By Definition \ref{def-ec}, we know that any point $p'$ of $P$ not covered by $V$ is outside $B$. Hence, 
$$\norm{p' - q_c} \geq 2(\Delta^{-1}(\epsilon))^{-1/t}n^{1/t}r_{min} \geq (\Delta^{-1}(\epsilon))^{-1/t}n^{1/t}r'_{min}, $$
which implies that $\norm{p' - q_c}^{-t} \leq \Delta^{-1}(\epsilon){r'}^{-t}_{min}/n$.

Let $C'_m$ denote $C_m(\psi(P_{cov}),q_c)$, where $P_{cov} \subseteq P$ is the set of input points that is covered by $V$.
Then $\psi(U(V_{max})) = C'_m$. Let $C''_m$ denote $C_m(\psi(P),q_c) \cap \psi(P_{cov})$.
Then we know that $F(C''_m,q_c) \leq F(C'_m,q_c)$.
By the definition of the influence function of the vector CIVD and the above discussion, we know that
\begin{eqnarray*}
F(C_m(\psi(P),q_c),q_c) & \leq & F(C''_m,q_c) + \sum_{p \in P \setminus P_{cov}} \norm{p - q_c}^{-t} \leq
F(C''_m,q_c) + \Delta^{-1}(\epsilon){r'}_{min}^{-t} \\
&=&  F(C''_m,q_c) + \Delta^{-1}(\epsilon)F(\{p'_{min} \},q_c) \leq
F(C''_m,q_c) + \Delta^{-1}(\epsilon)F(C_m(\psi(P),q_c),q_c).
\end{eqnarray*}
This means that $$F(C''_m,q_c) \geq (1 - \Delta^{-1}(\epsilon))F(C_m(\psi(P,q_c),q_c).$$
Thus, we have $F(\psi(U(V_{max})),q_c) \geq (1 - \Delta^{-1}(\epsilon))F(C_m(\psi(P),q_c),q_c)$.

The lemma then follows from Lemma \ref{lem-mgc} with the perturbation $\psi^{-1}$.
\qed
\end{proof}

\subsection{Finding an Effective Cover}

We now discuss how to implement the procedure of $Find(v_{r}, q_{c})$ in {\bf Algorithm \ref{alg-5}} for generating an effective cover. 

By Definition \ref{def-ec}, we know that an effective cover can be found straightforwardly by searching the aggregation-tree $T$ in a top-down fashion. We start at the root $v_r$. If $R(v_r)$ is small enough or it is disjoint with $B$ ({\em i.e.}, the box in Definition \ref{def-ec}), then we are done. Otherwise, we recursively search all its children.
A major drawback of this simple approach is that it could take too much time ({\em i.e.,} $O(n)$ in the worst case). Thus, a faster method is needed. 

To design a fast method, we first introduce two definitions.

\begin{definition}
An internal node $v$ of the aggregation-tree $T$ is {\em splittable} if the box $B$ (in Definition \ref{def-ec}) intersects at least two of the $2^{d}$ sub-boxes resulted from a quad-tree decomposition on $R(v)$.
\end{definition}

\begin{definition}
A node $v$ of the aggregation-tree $T$ {\em touches} $B$ if $R'(v)$ intersects $B$.
\end{definition}

To obtain a fast method for computing an effective cover, we first present a slow algorithm called {\em SlowFind} which may shed some light on how to speed up the computation.

\begin{algorithm}[]
\caption{SlowFind($v$, $q_{c}$)}
\textbf{Input:} A node $v$ of the aggregation-tree $T$ and a query point $q_{c}$.\\
\textbf{Output:} Part of an effective cover for $q_{c}$ in the subtree of $T$ rooted at $v$.
\begin{algorithmic}[1]

\State{If $R(v)$ does not intersect $B$, return.}

\State{If $R(v)$ is small enough, {\em i.e.}, $S(v) \leq \norm{q_c - L(v)}\Delta^{-1}(\epsilon)/(3d)$, report $v$ as one of the output nodes, return.}

\State{If $v$ is splittable, call SlowFind($v_{i}$, $q_{c}$) on each of $v$'s children, $v_{i}$, in the aggregation-tree $T$ that touches $B$, return.}

\State{Let $R$ be one of the $2^{d}$ sub-boxes resulted from a quad-tree decomposition on $R(v)$ that intersects $B$.  If $R$ contains no input point, return.}

\State{Let $v_{1}$ be the child of $v$ whose $R(v_{1})$ is contained inside $R$. 
For $l=1,2, \ldots$, do  

 Perform Steps 1 to 4 on $v_l$. If it does not return, this means that $v_l$ is  non-splittable and exactly one of its children intersects $B$. Let $v_{l+1}$ be that child, and $l=l+1$. Continue the loop (see Fig.~\ref{fig-longpath}).} 
\end{algorithmic}

\end{algorithm}

\begin{figure}
\center
\includegraphics[height=3.5in]{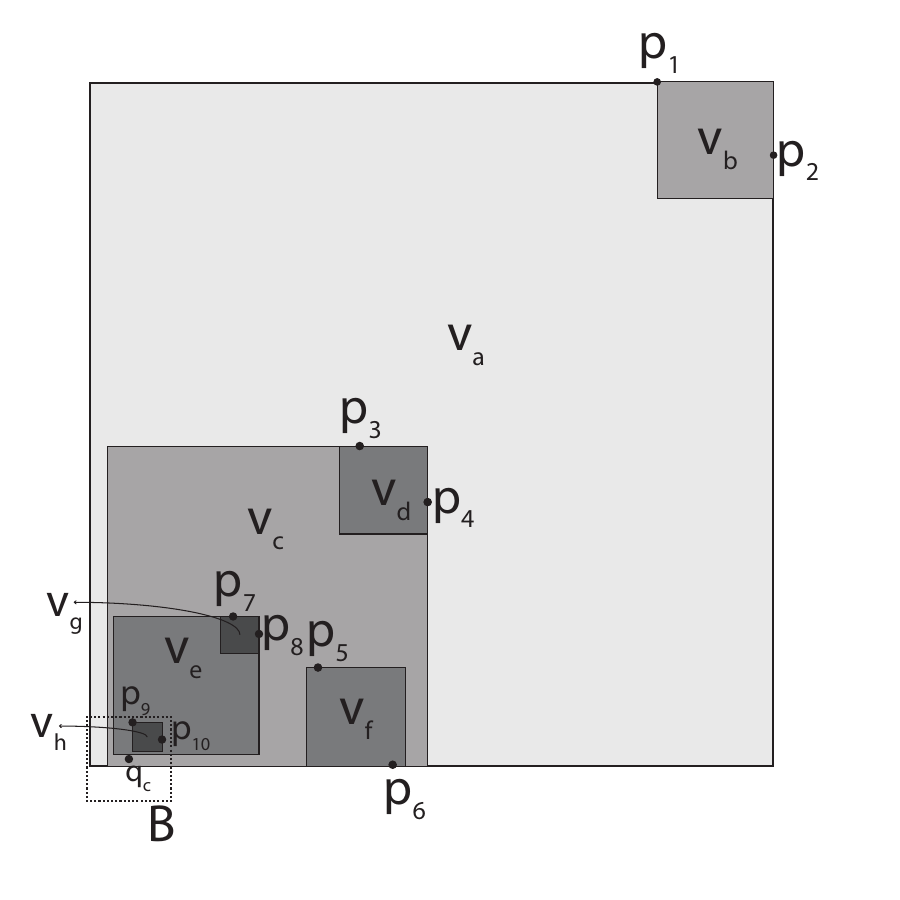}
\caption{An example illustrating Step 5 of {\bf Algorithm} SlowFind. Box $B$ (bounded by dotted line segments) intersects a sequence of nodes in the aggregation-tree $T$
(see \textbf{Figure} \ref{fig-longpath2})
which form a long path
(enclosed by dashed curves) in the aggregation-tree $T$.}
\label{fig-longpath}
\end{figure}

\begin{figure}
\center
\includegraphics[height=3.2in]{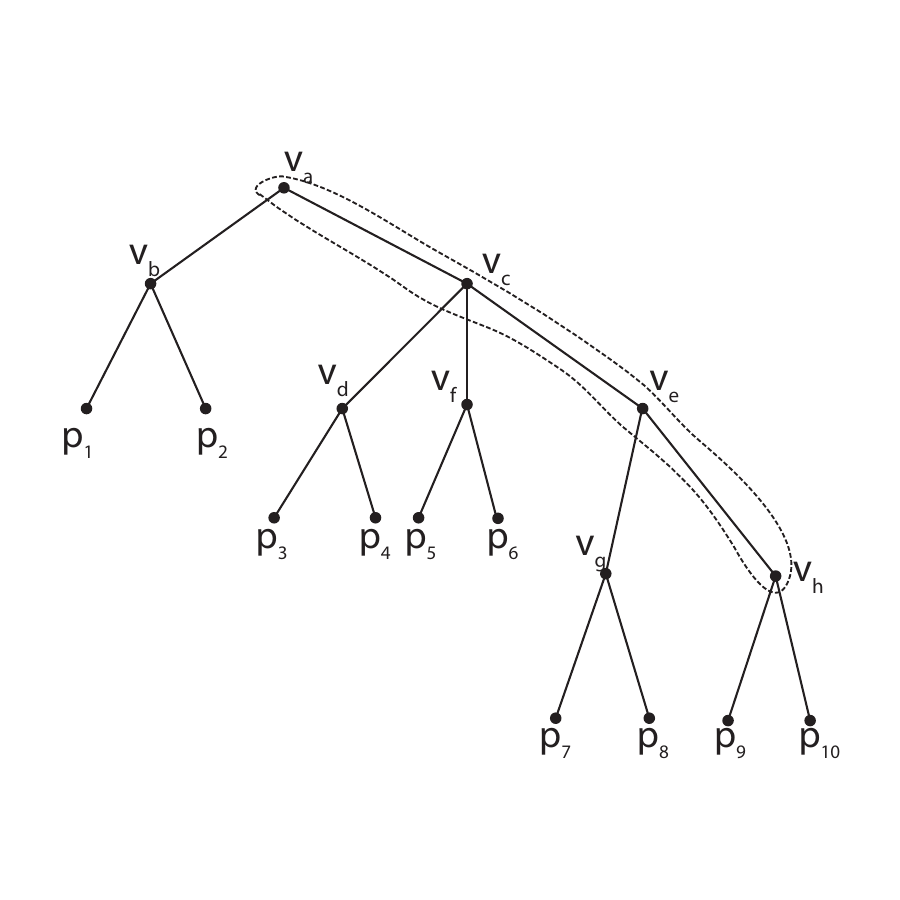}
\caption{The aggregation-tree $T$ for \textbf{Figure} \ref{fig-longpath}.}
\label{fig-longpath2}
\end{figure}

It should be pointed out that in the above SlowFind procedure, we use a loop, instead of recursive calls, in Step 5 to avoid the case that the recursion of SlowFind forms a possible {\em long path} in the aggregation-tree $T$ (see Fig. \ref{fig-longpath} and
\ref{fig-longpath2}). Searching through a long path would be the most time consuming computation in finding an effective cover. We call it the {\em long path} problem. Later, we will show how to overcome this main obstacle.

To obtain an effective cover, we can run SlowFind$(v_{r}, q_{c})$ on the root $v_{r}$ of the aggregation-tree $T$.
Below we show that for a properly chosen box $B$, the size of the recursion tree of SlowFind is only $O(\log n)$. 

\begin{lemma}\label{lem-slow}
The size of the recursion tree of SlowFind is $O(\log n)$,
if the size of $B$ is bounded by $c_{\epsilon}(\Delta^{-1}(\epsilon))^{-1/t}n^{1/t}r_{min}$, where $c_{\epsilon} > 0$ is a constant depending only on $\epsilon$.
\end{lemma}

\begin{proof}
First, we slightly change the SlowFind procedure. Note that in Step 3, it is possible that $v$ is splittable, but has less than two children touching $B$.
This is because some sub-box of $R(v)$ that intersects $B$  may not contain any input point and thus it does not correspond to a child of $v$ in the aggregation-tree $T$.
If this happens, we make a dummy child $v'$ of $v$ for this sub-box. The execution of SlowFind on a dummy child does not do anything and returns immediately.

Clearly, such a change can 
only increase the size of the recursion tree.
Below we show that the modified SlowFind has a recursion tree of size $O(\log n)$.
Note that we only need to prove that there are $O(\log n)$ leaves in the recursion tree,
since  every node in the recursion tree has either $0$ ({\em i.e.,} a leaf node) or  at least two children.

We associate each leaf node, SlowFind($v$, $q_{c}$), in the recursion tree with the box $R'(v)$.
Let $\mathcal{B}$ denote the set of such associated boxes.
It is easy to see that the following holds (except for the trivial case in which $v$ is the root  $v_{r}$ of $T$ and $B$ is disjoint with $R'(v_{r})$; in this case, the lemma is trivially true).
\begin{enumerate}
\item
The boxes in $\mathcal{B}$ are disjoint with each other.

\item
All boxes in $\mathcal{B}$ intersect $B$.

\item
Every box in $\mathcal{B}$ is not completely contained in $B_{min}$, where $B_{min}$ is a box centered at $q_c$ with an edge length of $\frac{r_{min}}{\sqrt{d}}$ 
(since otherwise it contradicts with the assumption that $r_{min}$ is the minimum distance between $q_c$ and all input points of $P$).
\end{enumerate}

Note that the ratio of the sizes of $B$ and $B_{min}$ is a polynomial of $n$. Hence, if we can prove that every box in $\mathcal{B}$ is big enough (comparing to its distance to $q_c$), 
then the lemma follows from Lemma \ref{lem-ann}.

Let SlowFind($v$, $q_{c}$) be a leaf node of the recursion tree and $v_p$ be the parent of $v$ in the aggregation-tree $T$. We assume that $v$ is not the root of the aggregation-tree $T$, since in this case the lemma is trivially true.
Clearly, $S(v_p) \geq \norm{q_c - l(v_p)}\Delta^{-1}(\epsilon)/(3d)$ (since otherwise, it is a leaf node in the recursion tree). Let $r_v$ denote the distance between $R'(v)$ and $q_c$, and $S'(v)$ denote the edge length of $R'(v)$.
Then, $r_v \leq \norm{q_c - l(v_p)} + \frac{\sqrt{d}S(v_p)}{2}$. 
Since $S'(v) = \frac{S(v_p)}{2}$, it is easy to see that
$S'(v) \geq r_v\frac{\Delta^{-1}(\epsilon)}{6d+\Delta^{-1}(\epsilon)}$.
Thus the lemma follows from the above discussion.
\qed

\end{proof}

The above lemma indicates that 
to find an appropriate $B$, it is sufficient to use an approximate value of $r_{min}$. Note that for a type-2 cell, all input points are recorded,
and $r_c$ is the smallest recorded distance. Let $r'_{min}$ be the value of $r_c$ at the time when $c$ becomes a cell ({\em i.e.,} no longer be partitioned), $p'_{min}$ be the input point with the recorded distance $r'_{min}$,
$p_{min} \in P$ be the input point such that $\norm{q_c - p_{min}} = r_{min}$, and $r_p$ be the recorded distance of $p_{min}$ for $c$.
By Lemma \ref{lm:distance}, we know that $r_{min} \geq  (1 - \beta) r_{p} \geq (1 - \beta) r'_{min}$,
and $r'_{min} \geq \frac{1}{1+\beta} \norm{q_c - p'_{min}} \geq \frac{1}{1+\beta} \norm{q_c - p_{min}} = \frac{1}{1+\beta} r_{min}$.
This means that $r'_{min}$ can be used as a good approximation of $r_{min}$.
We can set the edge length of $B$ as $ 4(1 + \Delta^{-1}(\epsilon))(\Delta^{-1}(\epsilon))^{-1/t}n^{1/t}r'_{min}$.
The value of $r'_{min}$ can be easily obtained from the AI Decomposition algorithm ({\em i.e.,} in $O(1)$ time).

SlowFind is  slow since Step 5 may take $O(n)$ time (due to the long path problem). Note that for some node $v$, after $l$ iterations in the loop of Step 5, SlowFind either returns or continues its recursion on the children of $v_l$.
If we can somehow find $v_l$ without actually iterating through the loop, then SlowFind will be much more efficient.
To solve this long path problem, we present below an improved method to search for the last $v_l$ (also denoted as $v_{l}$) in Step 5 (see Fig.~\ref{fig-longpath}, in which $v_{l}$ is $v_{h}$). 
Each search in the new method takes $O(\log n)$ time. Thus, the running time of SlowFind is improved to $O(\log^2 n)$ time.

\subsubsection{Long Path Problem:}

To solve the long path problem, we first label each edge in the aggregation-tree $T$ with a number in $\{1,2,\ldots,2^{d}\}$. The number is determined by a child's relative position in the box of its parent. This means that we label each $v$ of the $2^{d}$ possible children of the parent node $v_{p}$ based on the relative position of  the box $R'(v)$ of $v$ in the box $R(v_{p})$. We say that $v$ is the $i$-child  of $v_{p}$ if the edge connecting $v$ to its parent $v_{p}$ is labeled with the number $i$.

Consider a list of nodes $v_1',v_2',\ldots,v_m'$ in the aggregation-tree $T$, where $v_j'$ is the parent of $v_{j+1}'$ for each $j=1, 2, \ldots, m-1$.
If $v_1'$ is not an $i$-child of its parent for some $i$, $v_m'$ does not have an $i$-child, and $v_{j+1}'$ is the $i$-child of $v_j'$  for every $1 \leq j \leq m-1$, then such a path in $T$ is called an \emph{$i$-path} (see Fig.~\ref{fig-2path} and Fig.~\ref{fig-direc}).
\begin{figure}[h]
\centering
\includegraphics[height=2.8in]{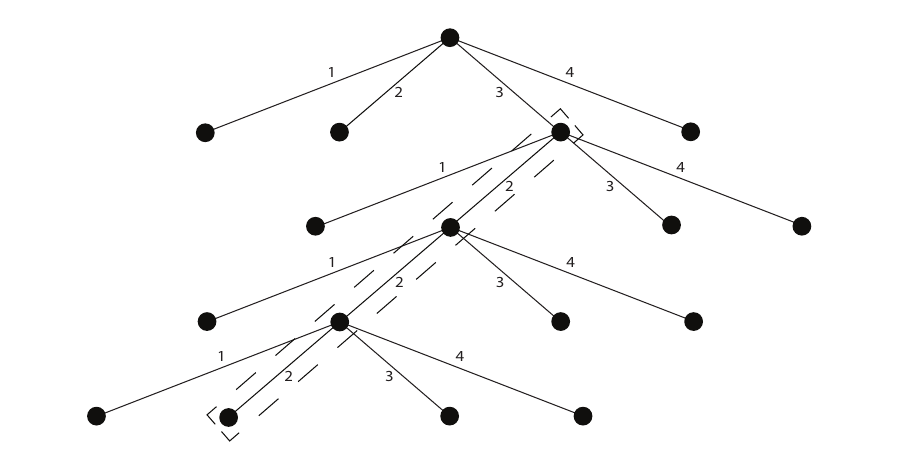}
\caption{An example of a 2-path (enclosed by the dashed line segments).}
\label{fig-2path}
\end{figure}

\begin{figure}[h]
\centering
\includegraphics[height=3.0in]{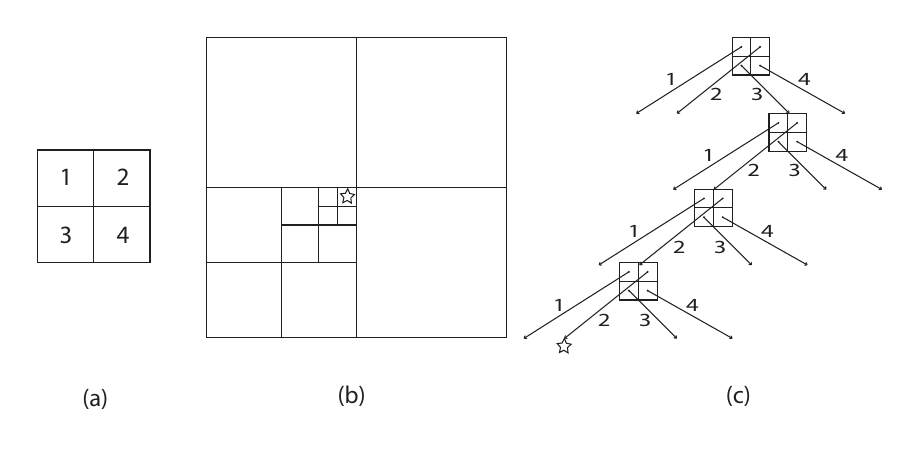}
\caption{A 2D example:
\textit{(a) Assign numbers to four sub-boxes}; \textit{(b) mark a quad-tree box with a star}; \textit{(c) the path in $T$ to the quad-tree box marked with a star.}}
\label{fig-direc}
\end{figure}

\begin{definition}
Let $e$ be an edge of the box $B$ and $v$ be a node in the aggregation-tree $T$. 
We say that $R(v)$ \emph{cuts} $e$ if $e$ intersects $R(v)$ and is not contained entirely in $R(v)$;
$e$ \emph{passes through} $R(v)$ if $e$ intersects $R(v)$ and none of its end vertices is inside $R(v)$.
\end{definition}

Now we discuss how to quickly find $v_l$ for a non-splittable node $v\in T$ in Step 5 of SlowFind. First, we consider the case that $R(v)$ cuts every edge of $B$ that intersects it (later, we will consider the case in which some edge of $B$ is fully contained in $R(v)$).  Note that in this case, it is impossible that an edge of $B$ passes through $R(v)$, since otherwise $v$ would be splittable. In this case, it means that exactly one vertex, say $u$, of $B$ is contained in $R(v)$. Let $i'$ be the label of the sub-box of $R(v)$ which contains $u$.
Let $P_I(v)$ be the $i'$-path in $T$ containing $v$.
We have the following three claims which can be easily verified.
\begin{claim}
$v_l$ must be in $P_I(v)$.
\end{claim}

\begin{claim}
For any node $v'$ lying strictly between $v$ and $v_l$ in $P_I(v)$, $v'$ will not satisfy the conditions ({\em i.e.,} in the ``if'' parts) in Steps 1 to 3 of SlowFind.
\end{claim}

\begin{claim}
If $v''$ is a proper descendant of $v_l$ in $P_I(v)$, then
$v''$ must satisfy the condition in at least one of the first three steps of SlowFind. Furthermore, if $v_{l}$ is splittable and $R(v'')$ intersects $B$, then $v''$ is also splittable.   
\end{claim}

Based on the above claims, we can perform a binary search on the $i'$-path $P_I(v)$ and find $v_{l}$ in $O(\log n)$ time. To do this, we need to prepare a data structure in the preprocessing.
The data structure stores every $i$-path of the aggregation-tree $T$, for $i=1, \ldots, 2^{d}$, in an array, and for every node $v$, stores a pointer pointing to the location of $v$ in each path containing $v$. Clearly, this data structure can be constructed in $O(n)$ time and space. To search for $v_{l}$, we just need to first find the  $i$-path $P_{I}(v)$, and use the three claims above to do binary search for $v_{l}$ on $P_{I}(v)$ ({\em i.e.,} use the conditions in Steps 1 to 3 to decide whether each searched node is an ancestor or descendant of $v_{l}$).

Next, we consider the case in which at least one edge of $B$ is fully contained in $R(v)$.
First, we give the following easy observations for any node $v_{0}$ in the aggregation-tree $T$. 
\begin{enumerate}
\item
If $R(v_0)$ does not fully contain an edge $e$ of $B$, then
for any descendant $v'$ of $v_0$ in $T$, $R(v')$ does not fully contain $e$.

\item
If $R(v_0)$ fully contains an edge $e$ of $B$, then
there is at most one child of $v_0$, say $v'$, whose $R(v')$ fully contains $e$.

\item
If $R(v_0)$ fully contains an edge $e$ of $B$, then
for any ancestor $v_a$ of $v_0$ in $T$, $R(v_{a})$ must fully contain $e$.
\end{enumerate}

By the above observations, we know that if $R(v)$ fully contains an edge $e$ of $B$, then all nodes $v'$ of $T$ whose $R(v')$ fully contains $e$ 
form a path in the aggregation-tree $T$ which starts at the root of $T$, reaches $v$, and may continue on some of $v$'s descendants. Clearly, it takes only $O(1)$ time to decide whether an edge $e$ of $B$ is fully contained in $R(v')$ for any node $v'$.
Let $Y(v)=\{e_1, e_2, \ldots, e_m\}$ be the set of edges of $B$ fully contained in $R(v)$, and $Z(v)$ be the path formed by the nodes in the aggregation-tree $T$ (starting at the root) whose corresponding boxes fully contain all edges in $Y(v)$. 

Since $|Y(v)|$ is a constant, for any descendant node $v'$ of $v$ in the aggregation-tree $T$, it is possible to decide in $O(1)$ time whether $v' \in Z(v)$.
Let $X(v)$ be the last node of $Z(v)$.
Then we have the following lemma.

\begin{lemma}
\label{lem-xvlog}
There is a data structure which can be pre-processed in $O(n\log n)$ time and $O(n)$ space, and can be used to find $X(v)$ in $O(\log n)$ time.
\end{lemma}

\begin{proof}
First, we describe the data structure for the search.
Consider the following procedure for partitioning the aggregation-tree $T$ into a set of chains.
Starting at the root of $T$, we walk down the tree by always choosing the child whose subtree has the largest number of nodes.
When a leaf node is reached, the path that we just walked is one of the chains to be produced.
Now if we take out the chain from the tree, the tree will be split into a set of subtrees.
Recursively perform the procedure on each of the subtrees.
We call the resulted chains the {\em majority paths} (see Fig.~\ref{fig-mp}). For each node in the aggregation-tree $T$, we assume that there is a pointer pointing to its location in the majority path containing it.

\begin{figure}
\center
\includegraphics[height=2.8in]{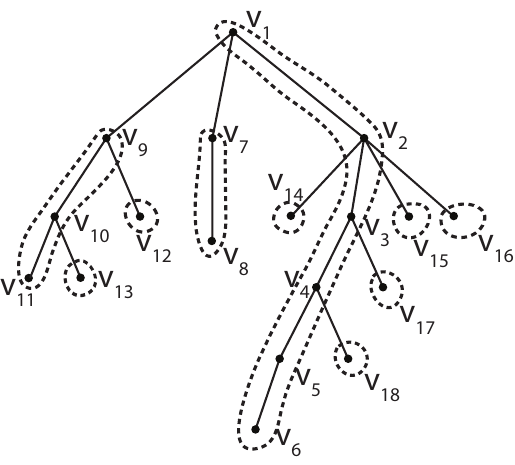}
\caption{An example of the majority path decomposition. Each majority path is enclosed by a dashed curve.}
\label{fig-mp}
\end{figure}

For a path $Z(v)$ where we want to find the tail $X(v)$,
the majority path decomposition decompose it into sub-paths.
Each sub-path ${v_1,v_2,\ldots,v_t}$ is the intersection of $Z(v)$ and some majority path $P_{m}$.
By performing a binary search on $P_m$,
the tail of the sub-path, $v_t$, can be identified.(Details will be shown below.)
Then either $v_t$ is $X(v)$,
or after $v_t$ the path $Z(v)$ enter another sub-path which is
the intersection of $Z(v)$ and another majority path.
We repeat the above process on the new sub-path until we find $X(v)$.
To make the strategy possible it suffices to build a binary search data structure for every majority path.

For any node $v'$ of the aggregation-tree $T$, let $T(v')$ denote the subtree of $T$ rooted at $v'$. For every majority path $P_{m}$, we build a binary tree $T_{P_{m}}$ for its nodes, say ${v_1,v_2,\ldots,v_m}$.
First we assign a weight to each node $v_i$. Let $\{v_1', v_2' \ldots, v_k'\}$ be all children of $v_i$ that are not in the majority path $P_{m}$.
The weight of $v_i$ is $1$ plus the total size of $T(v_1'), T(v_2'), \ldots, T(v_k')$, where the size of a subtree is the number of its nodes. 
To build the binary tree $T_{P_{m}}$ for ${v_1,v_2,\ldots,v_m}$, we first find the weighted median node $v_{j'}$ and
make $v_{j'}$ the root of $T_{P_{m}}$; then recursively build a subtree for ${v_1,v_2,\ldots,v_{j'-1}}$ and let its root be the left child of $v_{j'}$;
also recursively build a subtree for ${v_{j'+1},v_{j'+2},\ldots,v_m}$ and let its root be the right child of $v_{j'}$.

Let $Z'(v')$ be the sub-path of $Z(v)$ starting at node $v' \in Z(v)$ and ending at the last node $X(v)$ of $Z(v)$. Consider the following FindTail procedure.

\begin{algorithm}[h]
\caption{FindTail($T_s$, $v_s$)}
\textbf{Input:} A subtree $T_s$ of the aggregation-tree $T$ with its root being the head of
a majority path in $T$. A node $v_s$ of $T$ in $Z(v)$.\\
\textbf{Output:} $X(v)$.
\begin{algorithmic}[1]

\State{Let $T_{P_{m}(v_{s})}$ be the binary tree built for the majority path $P_{m}(v_{s})$ containing $v_{s}$. 
Conduct a binary search on $T_{P_{m}(v_{s})}$ 
to find the last node $w$ in the majority path $P_{m}(v_{s})$ which also appears in $Z'(v_s)$.}

\State{If all children of $w$
in $T$
are not in $Z(v)$, return $w$ as the last node $X(v)$ of $Z(v)$.}

\State{Otherwise, exactly one child of $w$, say $w_Z$, is in $Z(v)$. Call FindTail($T(w_Z)$, $w_Z$).}

\end{algorithmic}

\end{algorithm}

(Note: Input $T_s$ above is for purpose of analysis. When FindTail($T_s$, $v_s$) is called,
it means $X(v)$ is found to be in $T_s$ and FindTail will search in $T_s$ for $X(v)$.
However $T_s$ is not actually used in the procedure.)

Call FindTail($T$, $v$) to find the last node $X(v)$ of $Z(v)$.

As stated earlier, for any descendant node $v'$ of $v$ in the aggregation-tree $T$, it takes
$O(1)$ time to decide whether $v' \in Z(v)$.  Using this as a basic decision operation,
in the procedure FindTail above, the binary search in Step 1 is performed as follows. Start at the root $v_{r,s}$ of the binary tree $T_{P_{m}(v_{s})}$. If $v_{r,s}$ is a $Z(v)$ node and is also
the last $Z(v)$ node in the majority path $P_{m}(v_{s})$, then we are done with the binary search on $T_{P_{m}(v_{s})}$. 
If $v_{r,s}$ is a $Z(v)$ node but is not the last $Z(v)$ node in the majority path $P_{m}(v_{s})$
({\em i.e.,} this can be decided by checking whether the child $v'$ of $v_{r,s}$ in $T$ along the
path $P_{m}(v_{s})$ appears in $Z(v)$), then
search recursively on the right child of $v_{r,s}$ in the binary tree $T_{P_{m}(v_{s})}$.
If $v_{r,s}$ is not a $Z(v)$ node, then search recursively on the left child of $v_{r,s}$ in the binary tree.
From this, it is clear that calling FindTail($T$, $v$) will eventually find the last node $X(v)$ of $Z(v)$.

Clearly, in the procedure FindTail($T_s$, $v_s$),
$T(w_Z)$ is at most of half the size of $T_s$ due to the property of a majority path.
Thus, after each recursion of FindTail,
the search space is reduced by at least half of the size.
Since the aggregation-tree $T$ has $O(n)$ nodes, FindTail takes $O(\log n)$ recursions.

Let $W$ be the weight of the node $w$ ($w$ is found in Step 1 of FindTail($T_s$, $v_s$))
and $W'$ be the size of $T_s$. Clearly, $W'$ is equal to the total weight of
nodes of
$T_{P_{m}(v_{s})}$ (since the root of $T_{s}$ is the head of $P_{m}(v_{s})$). 
It is easy to see that the binary search in Step 1 takes $O(1) + O(\log \frac{W'}{W})$ time.
Also, $W$ is larger than the size of $T(w_Z)$.
FindTail($T$, $v$) produces a sequence of recursive calls.
Let $W_1, W_2, \ldots, W_a$ and $W_1',W_2', \ldots, W_a'$ be the values of $W$ and $W'$, respectively, in the sequence of FindTail calls, sorted by the time of the calls.
Note that $a = O(\log n)$ by the above discussion. Then the running time of FindTail($T$, $v$) is $$a \times O(1) + O(\log \frac{W'_1}{W_1} + \cdots + \log \frac{W'_a}{W_a})$$
 $$= O(\log n) + O(\log W'_1 - \log W_1 + \log W'_2 - \log W_2 + \cdots + \log W'_a - \log W_a)
\leq O(\log n) + O(\log W'_1).$$
The last inequality follows  from the fact that $W'_{i+1} \leq W_i$.
Since $W_1$ is the size of $T$, which is $O(n)$, it then follows that the running time of FindTail($T$, $v$) is $O(\log n)$.

The time and space of the preprocessed data structure are clearly $O(n \log n)$ and $O(n)$, respectively. Thus the lemma  is true.
\qed
\end{proof}

\begin{lemma}
\label{lem-xv}
If $Y(v)$ is not empty, then $v_l$ is either $X(v)$ or a descendant of $X(v)$ in $T$.
\end{lemma}

\begin{proof}
First, we show that $v_l$ is either a node in $Z(v)$ or a descendant of $X(v)$. Suppose this is not the case. Let $v_P$ be the last node of $Z(v)$ such that $v_P$ is an ancestor of $v_l$ in $T$.
Since $v \in Z(v)$ is an ancestor of $v_l$, such a node $v_P$ must exist. Since $v_P$ is not $X(v)$ and $v_P$ is not $v_l$, there are two distinct children of $v_P$, say $v_1$ and $v_2$, such that $v_1$ is $v_l$ or an ancestor of $v_l$, and $v_2$ is $X(v)$ or an ancestor of $X(v)$. This means that both $R'(v_1)$ and $R'(v_{2})$ intersect $B$ (by the definitions of $v_l$ and $X(v)$). 
But this contradicts with the fact that $v_{P}$ is non-splittable. 

Next, we show that if $R(v_l)$ is disjoint with $B$ or $S(v_l) \leq \norm{q_c - L(v_l)}\Delta^{-1}(\epsilon)/(3d)$, then $R(v_l)$ does not fully contain all edges in $Y(v)$. The case where $R(v_l)$ is disjoint with $B$ is trivial. Thus we focus only on the case of $S(v_l) \leq \norm{q_c - L(v_l)}\Delta^{-1}(\epsilon)/(3d)$. Suppose by contradiction $R(v_l)$ fully contains all edge in $Y(v)$. Then 
the edge length of $R(v_l)$ is larger than that of $B$.
Recall that $\Delta^{-1}(\epsilon)$ is set to be no bigger than $1/2$.  It is impossible that the edge length of $R(v_l)$ is larger than that of $B$ and also satisfies the inequality $S(v_l) \leq \norm{q_c - L(v_l)}\Delta^{-1}(\epsilon)/(3d)$. The reason is the following.  Since $R(v_l)$ intersects $B$, $\norm{q_c - L(v_l)}$ is no larger than the edge length $D(B)$ of $B$ plus the diameter of $R(v_l)$.
From this, we know that $S(v_l) \leq \norm{q_c - L(v_l)}\Delta^{-1}(\epsilon)/(3d)$ implies $$S(v_l) \leq (D(B) + dS(v_l))\Delta^{-1}(\epsilon)/(3d) \leq (S(v_l) + dS(v_l))\Delta^{-1}(\epsilon)/(3d)
\leq (1+d)S(v_l)/(6d).$$
This is impossible for any $d \ge 1$.
Therefore, if $R(v_l)$ is disjoint with $B$ or $S(v_l) \leq \norm{q_c - l(v_l)}\Delta^{-1}(\epsilon)/(3d)$, then $v_l$ is not in $Z(v)$ and must be a descendant of $X(v)$.

Finally, if $R(v_l)$ intersects $B$ and $S(v_l) > \norm{q_c - l(v_l)}\Delta^{-1}(\epsilon)/(3d)$, then $v_l$ must be splittable
(otherwise, $v_{l}$ will not be the last node in Step 5 of SlowFind).
Suppose $v_l$ is not $X(v)$ or a descendant of $X(v)$. Then $v_l$ and one of its children must be in $Z(v)$.
This means that $S(v_l)$ is at least 2 times the edge length of $B$.
Since $v_l$ is splittable, when $R(v_l)$ is divided into $2^d$ sub-boxes (in a quad-tree decomposition), at least one of these sub-boxes has one facet, say $f$, which intersects $B$ and is inside $R(v_{l})$
(i.e. is not part of a face of $R(v_{l})$)
.
The facet $f$ must intersect one of $B$'s edges, because its edge length is no smaller than that of $B$.
Therefore, some edge $e$ of $B$ must be cut after the decomposition. Note that $e$ cannot be any edge in $Y(v)$, since otherwise no child of $v_l$ will fully contain $e$, and this 
 contradicts with the fact that one child of $v_l$ is in $Z(v)$ and fully contains all edges in $Y(v)$.
This also means that $e$ is not entirely in $R(v_l)$.  Therefore, $e$ must pass through one of the $2^d$ sub-boxes of $R(v_{l})$ so that it can be possibly cut. This implies that
 the length of $e$ is larger than half of $S(v_l)$, which is a contradiction. 
Hence, $v_l$ is $X(v)$ or a descendant of $X(v)$.
\qed
\end{proof}

The above lemmas suggest that if $Y(v)$ is not empty, then we can use FindTail to first find $X(v)$ of $Z(v)$. $v_{l}$ is either $X(v)$ or its descendant. If it is the first case, then we have already found $v_{l}$. Otherwise, 
it means that at least one of the edges in $Y(v)$ has been cut while decomposing the box of $X(v)$. Thus, we can first determine the child $v'$ of $X(v)$ which is $v_{l}$ or its ancestor ({\em i.e.,} using the fact that $R'(v')$ intersects $B$). Then we generate a new set of edges, $Y(v')$, of $B$ which are fully contained in $R(v')$. Clearly, the size of $Y(v')$ is reduced by at least $1$ from that of $Y(v)$. If $Y(v')$ is not empty, then we repeat the above procedure to find a new $X(v')$. Since the size of $Y(v)$ is a constant, after a constant number of iterations, it will become zero. At that time, we can use the binary search method on $P_{I}(v'')$ to eventually find $v_{l}$, where $v''$ is the last node from the above process. The total time of the entire process is only $O(\log n)$.  This leads us to the following improved procedure Find($v,q_{c}$) for finding an effective cover.

\begin{algorithm}
\caption{Find($v$, $q_{c}$)}
\textbf{Input:} A node $v$ in the aggregation-tree $T$ and a query point $q_{c}$.\\
\textbf{Output:} Part of an effective cover for $q_{c}$ in the subtree $T(v)$ of $T$.
\begin{algorithmic}[1]

\State{If $R(v)$ does not intersect $B$, return.}

\State{If $R(v)$ is small enough, {\em i.e.}, $ S(v) \leq \norm{q_c - L(v)}\Delta^{-1}(\epsilon)/(3d)$, report $v$ as one of the output nodes, return.}

\State{If $v$ is splittable, call Find($v_{i}$, $q_{c}$) on each of $v$'s children, $v_{i}$, in $T$ that touches $B$, return.}

\State{Let $v_0=v$ and $i=0$. While $Y(v_{i})$ is not empty, do
\begin{enumerate}
\item[a.] Use FindTail to find $X(v_{i})$. If $X(v_{i})$ is $v_{l}$, let $v=X(v_{i})$ and  go to Step 1.
\item[b.] Otherwise, let $v_{i+1}$ be the child of $X(v_{i})$ which is $v_{l}$ or its ancestor. Let $i=i+1$ and continue the while loop.   
\end{enumerate}
Use binary search on $P_{I}(v_{i})$ to find $v_{l}$. Let $v=v_{l}$ and go to Step 1.
}

\end{algorithmic}

\end{algorithm}

\subsection{Algorithm Analysis}

Now we analyze the running time of our assignment algorithm.

\begin{lemma}
\label{lem-step4}
Step 4 of Find$(v,q_{c})$ takes $O(\log n)$ time.
\end{lemma}

\begin{proof}
Based on the above discussions, we know that the while loop in Step 4 can execute at most $O(1)$ iterations. In each iteration, the time is dominated by that of FindTail, which takes $O(\log n)$ time. Thus, the total time of the while loop is $O(\log n)$. The binary search on $P_{I}(v_{i})$ takes $O(\log n)$ time. Hence, the lemma follows. 
\qed
\end{proof}

The next lemma bounds the total time of the procedure Find($v,q_{c}$).

\begin{lemma}
\label{lem-find}
An effective cover of size $O(\log n)$ for $q_{c}$ can be obtained by the procedure Find($v_{r}, q_{c}$) in $O(\log^2 n)$ time.
\end{lemma}
\begin{proof}
Since {\bf Algorithm} Find improves only Step 5 of SlowFind, its recursion tree is the same as SlowFind and thus is of size $O(\log n)$ (by Lemma \ref{lem-slow}). By Lemma \ref{lem-step4}, we know that each recursion of Find takes $O(\log n)$ time. Hence, the total time for finding an
effective cover is $O(\log^{2}n)$. 
\qed
\end{proof}

\begin{lemma}
\label{lem-assign2}
{\bf Algorithm} Assign$(c)$ takes $O(\log^{\max\{2,d\}} n)$ time to assign a $(1-\epsilon)$-approximate maximum influence site to each type-2 cell of the AI decomposition.
\end{lemma}
\begin{proof}
The running time of Step 2 is $O(\log^{2}n)$ (by Lemma \ref{lem-find}). For the running time of Step 3, since an effective cover is of size $O(\log n)$ and each partition induced by a hyperplane
can be determined by $q_{c}$ and $d-1$ nodes (or more precisely, their representative points)
in the cover, the total number of these partitions is hence $O(\log ^{d-1}n)$. Each partition takes $O(\log n)$ time to compute the influence. 
Thus the total time of Step 3 is $O(\log^{\max\{2,d\}}n)$. Other steps take $O(1)$ time. Therefore, the total time of {\bf Algorithm} Assign is $O(\log^{\max\{2,d\}} n)$. The quality guarantee follows from Lemma \ref{lem-assign}.
\qed
\end{proof}

\begin{theorem}
A $(1-\epsilon)$-approximate vector CIVD can be constructed in $O(n \log^{\max\{3,d+1\}} n)$ time, where $n$ is the number of input points in $P$ and $d$ is the dimensionality of the space.
\end{theorem}
\begin{proof}
The correctness and approximation ratio follow from Lemma \ref{lem-assign2}. For the running time, we know that the AI decomposition takes $O(n\log n)$ time to generate totally $O(n \log n)$ cells.
Each cell takes $O(\log^{\max\{2,d\}}n)$ time to determine its approximate maximum influence site. Other preprocessing takes $O(n \log n)$ time. Thus, the total time is $O(n \log^{\max\{3,d+1\}}n)$.  
\qed
\end{proof}

 \section{Density-based CIVD}
\label{sec-ex2}

In this section, we show how to augment the AI-Decomposition algorithm to generate a $(1-\epsilon)$-approximate CIVD for the density-based CIVD problem.

\subsection{Problem Description and Properties of the Influence Function}

The density-based CIVD problem for a set $P$ of $n$ points in $\mathbb{R}^{d}$ is to partition the space into cells so that all points in each cell share the same subset $C$ of $P$ as their {\em densest cluster} (see Fig.~\ref{fig-dc}).  For a given query point $q \in \mathbb{R}^d$, the densest cluster $C_{m}(P,q)$ of $q$ is the subset $C$ of $P$ which maximizes the influence $F(C,q) = \abs {C}/V(C,q)$ over all subsets of $P$,
where $V(C,q)=\frac{\pi^{\frac{d}{2}} l^d}{\Gamma(\frac{d}{2} + 1)}$ is the volume of the smallest ball centered at $q$ and containing all points in $C$, $l$ is the maximum distance from $q$ to any point in $C$, and $\Gamma$ is the gamma function in the volume computation of a $d$-dimensional ball. 
In other words, $C_{m}(P,q)$ is the cluster with the highest density around $q$.
Fig.~\ref{fig-dencivd} shows an example of an approximate density-based CIVD generated by AI Decomposition.

Clearly, density-based CIVD is closely related to the widely used density-based clustering problem \cite{Cao06,Che05,Che07,Kri05,San98}. In the density-based clustering problem, two parameters (the radius $r$ of the neighborhood ball $B$ and the density $d$ of the input points inside $B$) are used to partition the input points into non-overlapping clusters.
 Different from the density-based clustering problem, our problem does not use such parameters, and can automatically determine the radius of each dense cluster.
It generates not only the dense clusters but also their associated Voronoi cells. Since density-based clustering is used in many data mining, pattern recognition, biomedical imaging,
and social network applications, we expect that the density-based CIVD is also applicable in these areas. 
Also, since our problem allows the generated clusters to overlap with one another, it has the potential to be applicable to overlapping clustering problems \cite{And12,Ban05,Bon11,Cle04}.

\begin{figure}
\center
\includegraphics[height=2.8in]{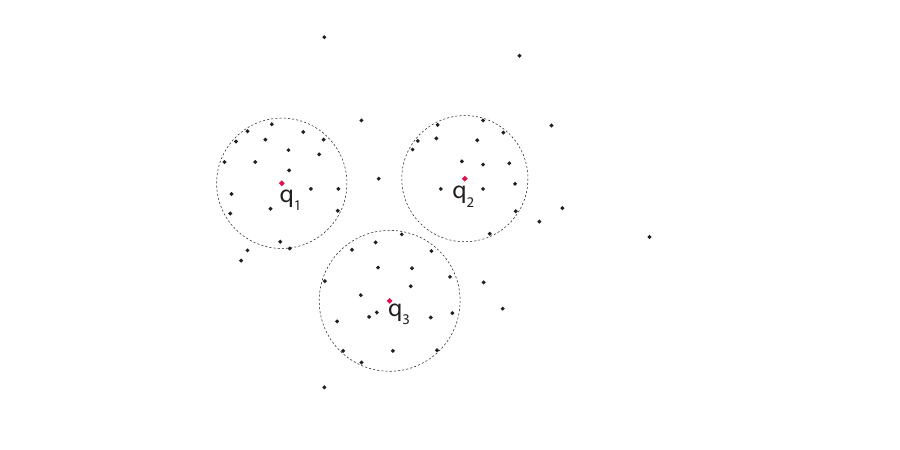}
\caption{Examples of the densest clusters for three query points $q_{1}, q_{2}$, and $q_{3}$.}
\label{fig-dc}
\end{figure}

\begin{figure}
\center
\includegraphics[height=2.4in]{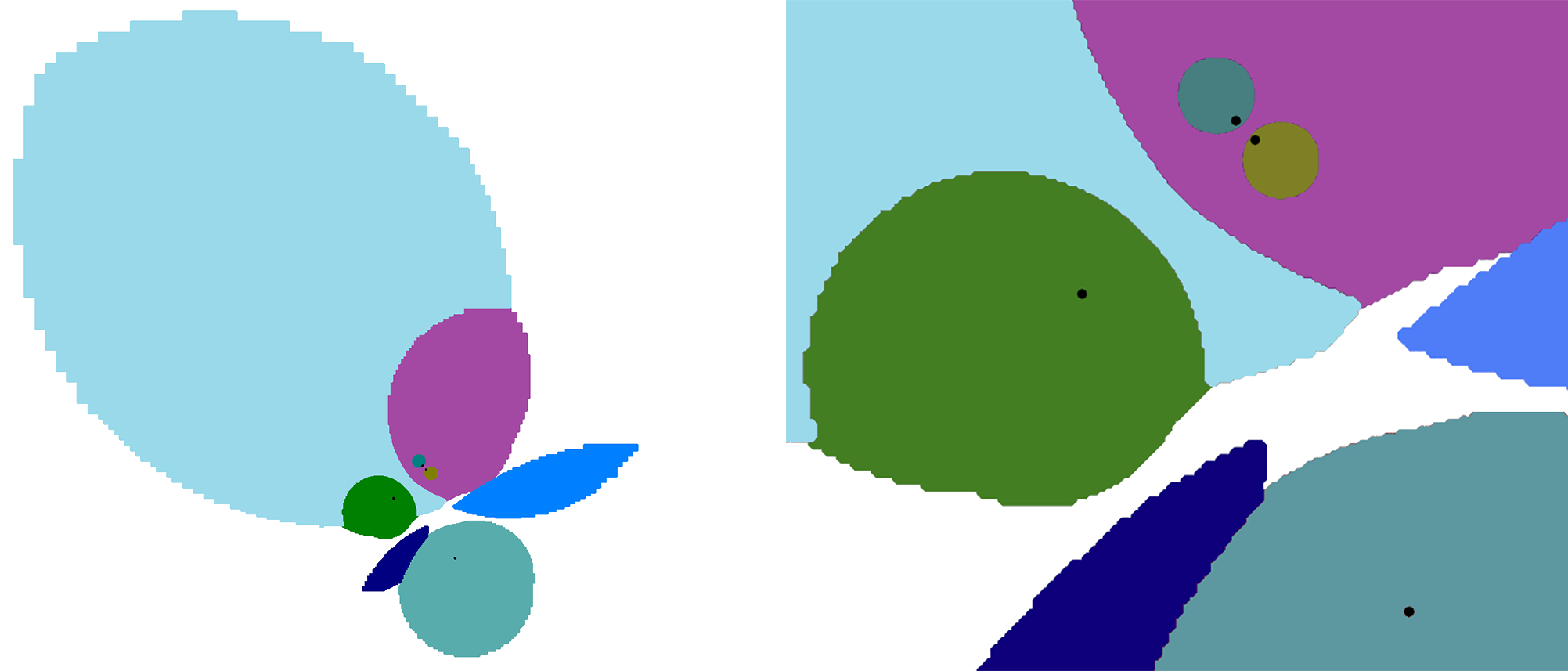}
\caption{An example of an approximate density-based CIVD for $4$ input points on the plane (generated by our algorithm with $\beta=0.1$). The figure on the right is a zoomed view of the figure on the left.}
\label{fig-dencivd}
\end{figure}

Similar to the vector CIVD problem, our goal for the density-based CIVD is also a $(1-\epsilon)$-approximate CIVD. To use the AI decomposition for this problem, we first show that it satisfies the three properties in Section \ref{sec-inf}.

\begin{theorem}
\label{the-densitypro}
The density-based CIVD problem satisfies the three properties in Section \ref{sec-inf}.
\end{theorem}

\begin{proof}
First, we show Property \ref{pro-1}.
Consider a set $C$ of input points and a query point $q$ in  $\mathbb{R}^d$. Let $\psi$ be an $\epsilon$-perturbation  on $P$ (with the witness point $q$) for some constant $0 < \epsilon < 1$. Let $p_{max}$ be the point in $C$ that is farthest from $q$, and $\psi(p_{max}')$ be the point in $\psi(C)$ that is farthest from $q$. Since 
\[
\norm{\psi(p_{max}') - q} \leq (1+\epsilon)\norm{p_{max}' - q}
\]
and
\[
\norm{p_{max}' - q} \leq \norm{p_{max} - q},
\]
we have
\begin{equation}\label{bound1}
\norm{\psi(p_{max}') - q} \leq (1+\epsilon)\norm{p_{max} - q}.
\end{equation}
Furthermore, from 
\[
\norm{\psi(p_{max}') - q} \geq \norm{\psi(p_{max}) - q}
\]
and
\[
\norm{\psi(p_{max}) - q} \geq (1-\epsilon)\norm{p_{max} - q},
\]
we get
\begin{equation}\label{bound2}
\norm{\psi(p_{max}') - q} \geq (1-\epsilon)\norm{p_{max} - q}.
\end{equation}
 By the influence function, and inequalities (\ref{bound1}) and (\ref{bound2}),
 we know that $$(1 + \epsilon)^{-d} F(C,q) \leq F(\psi(C),q) \leq (1 - \epsilon)^{-d} F(C,q).$$
This implies Property \ref{pro-1} (by setting $\delta(\epsilon)=\max\{1 - (1 + \epsilon)^{-d},(1 - \epsilon)^{-d}-1\}$).

For Property \ref{pro-2}, 
we assume that, for any query point $q \in \mathbb{R}^d$, there is a point $p \in P$ and a subset
$A$ of points in $P$ such that for every $a\in A$, $\norm{a-q} > n^{1/d} \norm{q-p}$. For any subset $B \subseteq P$ intersecting with $A$, let $b$ be a point in $A \cap B$. Then
$\norm{b-q} > n^{1/d} \norm{q-p}$. Now we compare $F(B,q)$ with $F(\{p\},q)$. It is clear that the smallest ball centered at $q$ and containing $B$ is at least $n$ times larger (in volume) than the smallest ball centered at $q$ and containing $\{p\}$. Since $\abs{B} \leq n$, we have $F(B,q) < F(\{p\},q)$. 

If there is a subset $P' \subseteq P$ and $ p \in P'$  such that  $n^{1/d}\norm{q-p} < \epsilon'\cdot\norm{q - p'} < \norm{q - p'}$  for all  $p' \in P\setminus P'$ and some constant $0 < \epsilon' < 1$, then by the above discussion, we have $C_{m}(P,q) = C_{m}(P',q)$. 
This means that for every cluster $C$ and query point $q$, the pair $(C,q)$ is stable.
Thus Property \ref{pro-2} holds.

For Property \ref{pro-3}, it is clear that after a scaling or a rotation about any query point $q \in \mathbb{R}^d $, the distance from every point in $P$ to $q$ is changed by the same factor which is uniquely determined by the transformation.  From the influence function, we know that Property \ref{pro-3} holds. 
\qed
\end{proof}

\subsection{Assignment Algorithm by Modifying the AI Decomposition}

To make use of the AI decomposition to construct an approximate density-based CIVD, our idea is to modify the AI-Decomposition algorithm ({\bf Algorithm \ref{alg-3}}) so that 
some additional information is maintained for assigning a cluster to each resulted type-2 cell. (Note that by Theorem \ref{the-aid}, for each type-1 cell $c$, we can simply use the distance-node $v$ which dominates $c$ as its densest cluster.) In this way, we can obtain the approximate CIVD at the same time when completing the AI decomposition.

Recall that  
an input point $p$ is recorded in the AI-Decomposition algorithm only when its distance to the current  to-be-decomposed box is large enough.
Therefore, for a cell $c$, it is most likely that an input point recorded earlier is farther away from $c$ than an input point recorded later. 
Intuitively, the recorded distances (of the input points) should be roughly in a decreasing order with respect to the order in which they are recorded.  
Below we discuss how to utilize this observation to modify the AI-Decomposition algorithm. A proof of this observation will be given later.

To show how to modify the AI-Decomposition algorithm, we first consider an example. Let $q$ be a query point, and $P = \{p_1,p_2,\ldots, p_n\}$ be a set of input points  in the decreasing order of their distances to $q$ ({\em i.e.},  $\norm{p_i - q} > \norm{p_j - q}$ for all $1 \leq i < j \leq n$, and no two different points in $P$ have the same distance to $q$). 
To find $C_{m}(P,q)$, we can use the following approach which scans $P$ only once in its sorted order and uses $O(1)$ additional space.
For each $1 \leq i \leq n$, we compute $D_i = \frac{c_d(n - i + 1)}{\norm{p_i - q}^d}$, and store the largest $D_{i}$ during the scanning process, where $c_{d}= \frac{\Gamma(\frac{d}{2} + 1)}{\pi^{\frac{d}{2}}}$.
Since the ball centered at $q$ and with radius $\norm{p_i - q}$ contains exactly $n-i+1$ input points, $\{p_{i}, p_{i+1}, \ldots, p_n\}$, the largest $D_i$ value, along with the corresponding $i$, gives us the desired densest cluster  
 $C_m(P,q) = \{p_{i}, p_{i+1}, \ldots, p_n\}$.
 
\begin{figure}
\center
\includegraphics[height=3.2in]{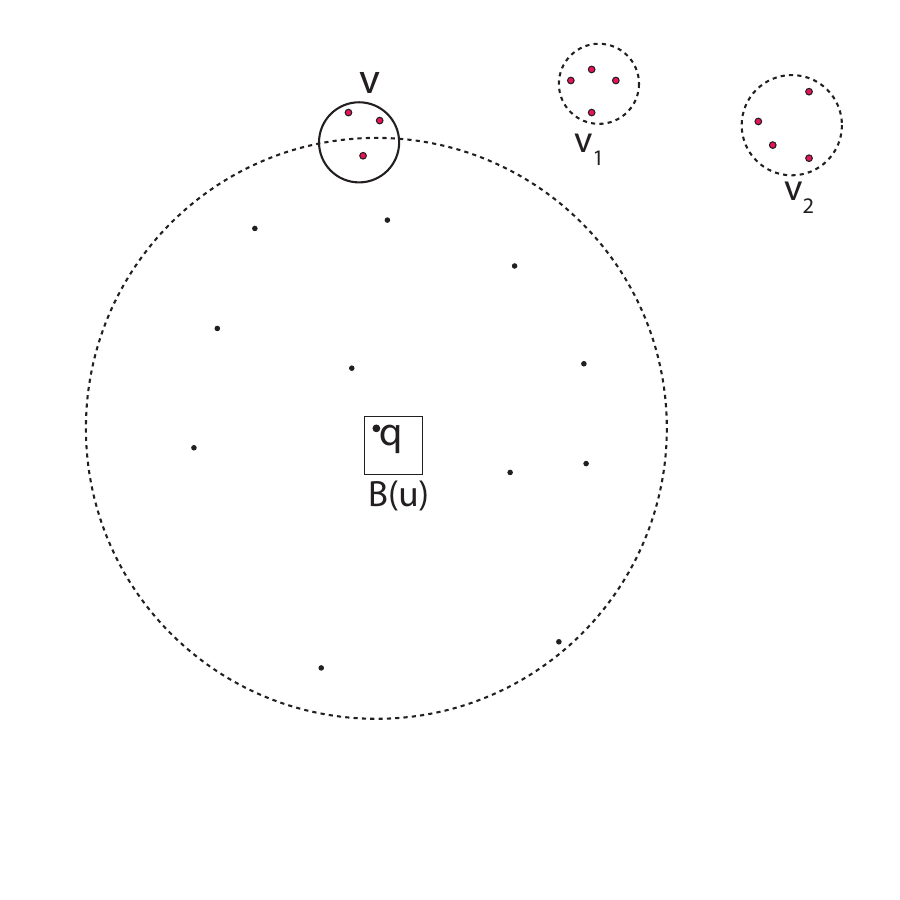}
\caption{A configuration with 22 input points, $B(u)$ being processed, $v_1$ and $v_2$ been removed from $L$ before,
and $v$ being removed from $L$.
For any $q$ in $B$ if we draw a ball centered at $q$ with radius $r$
being roughly the distance between $q$ and $v$($\norm{q-l(v)}$),
the ball should ``approximately" include all the input points not yet been recorded.(i.e. points not in $v_1$ and $v_2$.)
Denote set of these points by $P_u$.
Even without knowing much information about $P_u$ at this point,
it is possible to obtain an approximate value of $F(P_u,q)$.
($\frac{c_2(22-3-4)}{r^2}$ in this case.)
Since this works for all $q$ in $B(u)$,
This information can be passed down during the recursion. 
}
\label{fig-denalg}
\end{figure}

With the above illustration, we can now modify the Decomposition algorithm ({\bf Algorithm \ref{alg-2}}), as follows. 
In particular, we change Step 2 of the Decomposition algorithm, since this is the step in which distance-nodes are removed from the list $L$. 
Before the execution of Step 2, we sort the distance-nodes in $L$ by the decreasing order of their distances to the current box-node $u$ (if multiple nodes have the same distance, then we order them arbitrarily). Then, 
we execute Step 2 and try to remove distance-nodes from $L$ according to this sorted order. We assume that in the Decomposition algorithm, a number $M$ is maintained for storing the total number of input points which are recorded for the current box-node $u$. During the execution of Step 2,  after removing each distance-node $v$,  we compute a value
$D=\frac{c_d(n - M)}{r^d}$ and then update $M$
({\em i.e.,} increase $M$ by the cardinality $|P_{v}|$ of $v$). We save the largest value $D$ along with the corresponding distance-node $v$ and the box-node $u$.  In each recursive call to the Decomposition algorithm, we pass the stored $D$, $u$, and $v$ to the next level of the recursion.

See \textbf{Figure} \ref{fig-denalg} for better understanding of the strategy.

Clearly, such a modification on Step 2 of the Decomposition algorithm resembles the computation in the above example. The only difference is that in the above example, the input points are considered strictly in the decreasing order of their distances to the query point $q$, but in the modified Decomposition algorithm, distance-nodes are not always removed by the decreasing order of their distances to some  query point $q$.
This is because at different recursion levels, distance-nodes  may not be removed in a strictly decreasing order. Below we show that an approximate densest cluster for $q$ can still be obtained,
despite the above difference. 

Let $c$ be a type-2 cell generated by the AI decomposition and $q$ be a query point in $c$. Consider the root-to-$c$ path in the recursion tree of the Decomposition algorithm for $c$. Let $v_{1}, v_{2}, \ldots, v_{m}$ be the sequence of distance-nodes removed from $L$ along this recursion path (sorted by the increasing order of the time when they are removed),
and $x_{1}, x_{2}, \ldots, x_{m}$ be the closest distances to their corresponding box-nodes $u$ at the time when they are removed.  Let $D_{max}$ be the maximum value of $D$ passing through this recursion path and $v_{max}$, and $u_{max}$ be the corresponding box-node and distance-node when $D$ achieves its maximum value.

Below we prove a claim that if $v_{max}=v_{i}$ for some $i$, then the union of $v_{i},v_{i+1},\ldots, v_m$ is almost the densest cluster for $q$ for a properly chosen $\beta$. 
Since $v_1,v_2,\ldots, v_m$ are all distance-nodes recorded for the type-2 cell $c$, by Lemma \ref{lm:partition}, we know that they form a partition of $P$. 
For any $p \in v_j$, by Lemma \ref{lm:distance}, we have
\begin{equation}\label{eq:psi}
(1 - \beta)x_j \leq \norm{q - p} \leq (1 + \beta)x_j.
\end{equation}
Let $\psi$ be a mapping defined as follows: 
For any $p \in v_j$, $\psi(p)$ is on the ray that emits from $q$ and passes through $p$, and with $\norm{p - q} = x_j$. 
Let $C$ denote the union of $v_{i},v_{i+1},\ldots, v_m$. By Lemma \ref{lem-app}, we know that to prove the above claim, it is sufficient to show that $F(\psi(C),q)$ is almost as large as $F(C_m(P',q),q)$, where $P'=\psi(P)$.
Clearly, $P'$ can be partitioned into subsets 
$\psi(v_1),\psi(v_2),\ldots, \psi(v_m)$, with  all
points in each $\psi(v_i)$, for $i=1, \ldots, m$, having the same distance $x_{i}$ to $q$.
Based on the above discussion, we know that if $x_1,x_2,\ldots, x_m$ are in decreasing order, then $\psi(C)$ is exactly $C_m(P',q)$. 
The following lemma shows that $x_{1}, x_{2}, \ldots, x_{m}$ are actually in a roughly sorted order, which is sufficient for us to obtain an approximate densest cluster.  

\begin{lemma}\label{lm:sorted}
In the modified AI-Decomposition algorithm with an error tolerance $\beta$, 
$ x_j \leq (1 + \beta)x_i $ for any $1 \leq i < j \leq m$.
\end{lemma}

\begin{proof}
From the above discussion, we know that if $v_i$ and $v_j$ are removed in the same recursion of the Decomposition algorithm, then $x_j \leq x_i$ (since in the same recursion, all distance-nodes are removed in a decreasing order).  
Thus we can assume that $v_i$ and $v_{j}$ are removed in different recursions with recursive calls Decomposition$(u_1,\beta,L_1, T_{p}, r_1)$ and Decomposition$(u_2,\beta,L_2, T_{p}, r_2)$, respectively, where
$u_1$ is a proper ancestor of $u_2$ in the box-tree $T_{q}$. By definition, we know that $x_i$ is  the closest distance between $l(v_i)$ and $B(u_1)$, and $x_j$ is the closest distance between $l(v_j)$ and $B(u_2)$. Also, by the Decomposition algorithm, we know that $B(u_2)$ is contained in $B(u_1)$.

Consider the execution of Decomposition$(u_1,\beta,L_1, T_{p}, r_1)$.
Let $v_j'$ be the node in $L_1$ containing $l(v_j)$. Clearly, $v_j'$ is not removed in Step 2, since otherwise $v_j$ will not appear in Decomposition$(u_2,\beta,L_2$, $T_{p}, r_2)$.
This means that $x_i$ is larger than the closest distance $x_j'$ between $l(v_j')$ and $B(u_1)$, since $v_i$ is removed from $L_1$ in Step 2 but $v_j'$ is not.
Let $D(u_{1})$ denote the diameter of $B(u_1)$, and $x''_j$ denote the closest distance between $B(u_2)$ and $l(v_j')$.
By the triangle inequality, we have $x_j'' \leq D(u_{1}) + x_j'$. Since $x_j' \leq x_i$ and $D(u_{1}) \leq x_i\beta/2$ (by the fact that $v_i$ is removed from $L$ in Step 2), we have
$x_j'' \leq (1 + \beta/2) x_i$.

Let $E'(v_{j}')$ be the box co-centered with $E(v_j')$ and with the edge length half of that of $E(v_j')$ (see Algorithm \ref{alg-1}). Consider the following two possible cases.
\begin{enumerate}
\item
\emph{$B(u_1)$ does not intersect $E'(v_{j}')$.}  In this case, we have 
\[
\norm{l(v_j) - l(v_j')} \leq s(v_j') \leq  x_j'\beta/2 \leq x_i\beta/2.
\]
Note that
\[
x_j \leq \norm{l(v_j) - l(v_j')} + x_j''.
\]
Thus, we have
\[
x_j \leq (1 + \beta)x_i.
\]

\item
\emph{$B(u_1)$ intersects $E'(v_{j}')$.} In this case, note that some part of $B(u_1)$ must be outside of $E(v_j')$, since otherwise $v_j'$ would have been deleted in Step 1.
From this, we know that $D(u_{1}) \geq \sqrt{d}\alpha/2$, where $\alpha$ is the edge length of $E'(v_{j}')$.
Since $\beta < 1/2$ (by the assumption on $\beta$), we have 
\[
\norm{l(v_j) - l(v_j')} \leq s(v_j') =  (\alpha/2) \cdot \beta/2 \leq \alpha/8.
\]
Note that since $B(u_1)$ intersects $E'(v_{j}')$, we have $x_j' \leq \sqrt{d}\alpha/2$.
Thus, $x_j' + \norm{l(v_j) - l(v_j')} = \sqrt{d}\alpha/2 + \alpha/8 \leq \sqrt{d}\alpha \leq 2D(u_{1})$.
By the triangle inequality, we have $x_j' + \norm{l(v_j) - l(v_j')} + D(u_{1}) \geq x_j$, which means $x_j \leq 3 D(u_{1})$.
Recall that $D(u_{1}) \leq x_i\beta/2 \leq x_i/4$; we have $x_i \geq 4 D(u_{1})$. Therefore $x_i > x_j$.
\end{enumerate}
\qed
\end{proof}

With the above lemma, we immediately have the following lemma (in which the notation was defined
before Lemma \ref{lm:sorted}).

\begin{lemma}\label{lm:dense}
$F(\psi(C),q) \geq (1+\beta)^{-d} F(C_m(P',q),q)$.
\end{lemma}

\begin{proof}
From the discussion before Lemma \ref{lm:sorted}, we know that
$P'$ can be partitioned into $\psi(v_1),\psi(v_2),\ldots, \psi(v_m)$, and for any $p \in \psi(v_r)$, $\norm{p - q} = x_r$  for all $ 1 \leq r \leq m$.

It is easy to see that $C_m(P',q)$ can be written as the union of all distance-nodes in $\{\psi(v_j) \mid x_j \leq x_{i_{max}} \}$ for some $ 1 \leq i_{max} \leq m$.
$\psi(C)$ is the union of $\psi(v_{i}), \psi(v_{i+1}), \ldots, \psi(v_m)$.
Let
\[
F' = \frac{c_d(n - \abs{v_1} - \abs{v_2} -\cdots - \abs{v_{i - 1}})}{x_{i}^d}.
\]
Let  $x_{h}$ be the maximum in $\{x_{i},x_{i+1},\ldots, x_m \}$. 
Then, 
\[
F(\psi(C),q) = \frac{c_d(n - \abs{v_1} - \abs{v_2} -\cdots - \abs{v_{i - 1}})}{x_{h}^d}.
\]
By Lemma \ref{lm:sorted}, we know that $ x_h \leq (1 + \beta) x_{i} $. Thus,   
\[
F(\psi(C),q) \geq (1+\beta)^{-d} F'.
\]

Rewrite $\{\psi(v_j) \mid x_j \leq x_{i_{max}} \}$ as $\{v_{j_1}, v_{j_2}, \ldots, v_{j_w} \}$ with $ j_1 < j_2 < \cdots < j_w $. We have
\[
F(C_m(P',q),q) = \frac{c_d (\abs{v_{j_1}} + \abs{v_{j_2}} + \cdots + \abs{v_{j_w}})}{x_{i_{max}}^d}.
\]
Let 
\[
G' = \frac{c_d (\abs{v_{j_1}} + \abs{v_{j_1 + 1}} + \cdots + \abs{v_{m}})}{x_{j_1}^d}
= \frac{c_d (n - \abs{v_{j_1 - 1}} - \abs{v_{j_1 - 2}} - \cdots - \abs{v_1})}{x_{j_1}^d}.
\]
Note that $ \abs{v_{j_1}} + \abs{v_{j_2}} + \cdots + \abs{v_{j_w}} \leq \abs{v_{j_1}} + \abs{v_{j_1 + 1}} + \cdots + \abs{v_{m}} $.
Since $x_{i_{max}} \geq x_{j_1}$,  
we have 
\[
G' \geq F(C_m(P',q),q).
\]

Since $G' \leq F'$,
it follows that
\[
F(\psi(C),q) \geq (1+\beta)^{-d} F' \geq (1+\beta)^{-d} G' \geq (1+\beta)^{-d} F(C_m(P',q),q).
\]
\qed
\end{proof}

Based on the above two lemmas, we immediately have the following theorem.

\begin{theorem}
For any $\beta$ satisfying the conditions 
$1 - (1+\beta)^{-d} \leq \Delta^{-1}(\epsilon)$  and 
$\beta \leq \Delta^{-1}(\epsilon)/3$, the modified AI decomposition algorithm finds a $(1-\epsilon)$-approximate density-based CIVD in $O(n \log^{2}n)$ time.
\end{theorem}

\begin{proof}

By Equation (\ref{eq:psi}) and the definition of $\psi$, we know that for each  $p \in v_k$, 
\[
\norm{\psi(p) - p} = \abs{\norm{p - q} - \norm{\psi(p) - q}} = \abs{\norm{p - q} - x_k} \leq \beta x_k \leq \Delta^{-1}(\epsilon) x_k/3 \leq \Delta^{-1}(\epsilon) \norm{\psi(p) - q}.
\]
By Lemma \ref{lm:dense}, we have 
\[
F(\psi(C),q) \geq (1+\beta)^{-d} F(C_m(P',q),q) \geq (1 - \Delta^{-1}(\epsilon)) F(C_m(P',q),q).
\]
Consider the perturbation $\psi^{-1}$. If $\beta \leq \Delta^{-1}(\epsilon)/3$, then $\psi^{-1}$ is also a $(\Delta^{-1}(\epsilon)/3)$-perturbation.
From the proof of Theorem \ref{the-densitypro}, we know that $(C,q)$ is stable.
By Lemma \ref{lem-mgc} (note that Lemma \ref{lem-mgc} still holds for the density-based CIVD problem), we have $F(C,q') \geq (1 - \epsilon)F(C_m(P,q'),q')$ for any point $q'$ in $c$.

For the running time, we note that the additional computation in the modified AI decomposition algorithm includes sorting the distance-nodes in $L$ and maintaining the values of $M$, $D$, and the corresponding $u$ and $v$. Clearly, the additional time is dominated by sorting, which takes $O(|L| \log |L|)$ time for each recursive call to the Decomposition algorithm. Since the running time of the original Decomposition algorithm is $O(|L|)$, and the total running time of the entire AI decomposition is the sum over all recursions of the Decomposition algorithm, the total time of the modified AI decomposition thus increases only by an $O(\log n)$ factor.
Therefore, the theorem follows.  
\qed
\end{proof}


\begin{thebibliography}{7}





\bibitem{And12}
R. Andersen, D.F. Gleich, and V. Mirrokni, ``Overlapping Clusters for Distributed Computation,'' {\em Proc. 5th ACM International Conference in Web Search and Data Mining}, 2012, pp.~273-282.

\bibitem{Ary02}
S. Arya and T. Malamatos, ``Linear-Size Approximate Voronoi Diagrams,'' {\em Proceedings of the 13th annual ACM-SIAM symposium on Discrete algorithms (SODA'02)}, pp.~147--155, 2002.

\bibitem{Ary02a}
S. Arya, T. Malamatos, and D. M. Mount, ``Space-Efficient Approximate Voronoi Diagrams,'' {\em Proc. 34th ACM Symp. on Theory of Computing (STOC 2002)}, pp.~721--730, 2002.
 
 \bibitem{Ary98}
S. Arya, D. M. Mount, N. S. Netanyahu, R. Silverman, and A. Wu, ``An Optimal Algorithm for Approximate Nearest Neighbor Searching,'' {\em Journal of the ACM,} 45 (1998), pp.~891--923. 
 


\bibitem{Aur87}
 F. Aurenhammer, ``Power Diagrams: Properties, Algorithms and Applications,'' {\em SIAM J. on
Computing}, 16(1)(1987), 78-96.

\bibitem{Aur91}
 F. Aurenhammer, ``Voronoi Diagrams -- A Survey of a Fundamental Geometric Data Structure,'' {\em ACM
Computing Surveys}, 23(1991), 345-405.


\bibitem{Ban05}
A. Banerjee, C. Krumpelman, S. Basu, R.J. Mooney, and J. Ghosh, ``Model-based Overlapping Clustering,'' {\em Proc.  11th ACM SIGKDD International Conference on Knowledge Discovery and Data Mining}, 2005, pp.~532--537.

\bibitem{Bar02}
G. Barequet, M.T. Dickerson, and R.L.S. Drysdale III, ``2-Point Site Voronoi Diagrams,''
{\em Discrete Applied Mathematics}, 122(1-3)(2002), 37-54.
\bibitem{Bar11}
G. Barequet, M.T. Dickerson, D. Eppstein, D. Hodorkovsky, and K. Vyatkina, ``On 2-Site
Voronoi Diagrams under Geometric Distance Functions,'' {\em Proc. 8th International Symp. on Voronoi
Diagrams in Science and Engineering}, 2011, pp.~31-38.

\bibitem{Bon11}
F. Bonchi, A. Gionis, and A. Ukkonen, ``Overlapping Correlation Clustering,'' {\em Proc.
IEEE 11th International Conference on Data Mining}, 2011, pp.~51-60.

\bibitem {CK95}
P. Callahan and R. Kosaraju, ``A Decomposition of Multidimensional Point Sets with Applications to $k$-nearest-neighbors and $n$-body Potential Fields,'' {\em JACM}, 42(1)(1995), 67-90.


\bibitem{Cao06}
F. Cao, M. Ester, W. Qian, and A. Zhou,
``Density-based Clustering over an Evolving Data Stream with Noise,"
{\em Proceedings of the 6th SIAM International Conference on Data Mining}, 2006,
pp. 328-339.


\bibitem{Che05}
D.Z. Chen, M.H.M. Smid, and Bin Xu, ``Geometric Algorithms for
Density-based Data Clustering,'' {\em Int. J. Comput. Geometry Appl.}, 15(3)(2005),
239-260.

\bibitem{Che12}
N. Chen, J. Zhu, F. Sun, and E.P. Xing, ``Large-margin Predictive Latent Subspace Learning for Multi-view Data Analysis,'' {\em  IEEE Transaction on Pattern Analysis and Machine Intelligence}, 
34(12)(2012), 2365-2378. 

\bibitem{Che07}
Y. Chen and L. Tu, ``Density-based Clustering for Real-time Stream Data,"
{\em Proceedings of the 13th ACM SIGKDD International Conference on Knowledge Discovery and Data Mining}, 2007,
pp. 133-142.


\bibitem{Chr08} 
C.M. Christoudias, R. Urtasun, and T. Darrell, ``Multi-view Learning in the Presence of View Disagreement,'' arXiv:1206.3242, June 2012.

\bibitem{Cle04}
G. Cleuziou, L. Martin, and C. Vrain, ``Poboc: An Overlapping Clustering Algorithm,'' {\em Proc. 16th European Conference on Artificial Intelligence}, 2004, pp.~440-444.



\bibitem{Dic09}
 M.T. Dickerson and D. Eppstein, ``Animating a Continuous Family of Two-site Voronoi
Diagrams (and a Proof of a Bound on the Number of Regions),'' {\em Proc. 25th ACM Symp. Computational
Geometry}, 2009, pp. 92-93.

\bibitem{Dic08}
M.T. Dickerson and M.T. Goodrich, ``Two-site Voronoi Diagrams in Geographic Networks'',
{\em Proc. 16th ACM SIGSPATIAL International Conf. Advances in Geographic Information Systems}, 2008,
doi:10.1145/1463434.1463504.

\bibitem {DX11}
H. Ding and J. Xu, ``Solving Chromatic Cone Clustering via Minimum Spanning Sphere,''
{\em Proc. 38th International Colloquium on Automata, Languages and Programming (ICALP)},
2011, pp.~773-784.

\bibitem{Gre09} 
     D. Greene and P. Cunningham, ``Multi-view Clustering for Mining Heterogeneous Social Network Data,'' {\em Proc. 31st European Conference on Information Retrieval, Workshop on Information Retrieval over Social Networks}, LNCS, Vol. 5478, 2009.     

\bibitem{Gre88}
L. Greengard, {\em The Rapid Evaluation of Potential Fields in Particle Systems.} MIT Press, Cambridge (1988).

\bibitem{Gre90}
 L. Greengard. ``The Numerical Solution of the N-body Problem,'' {\em Computers in Physics,} 4, pp.~142--152 (1990).

\bibitem{Gre94}
L. Greengard. ``Fast Algorithms for Classical Physics.'' {\em Science} 265, 909--914 (1994).

\bibitem{Han10}
 I. Hanniel and G. Barequet, ``On the Triangle-Perimeter Two-site Voronoi Diagram,''
{\em Trans. on Computational Science}, 9(2010), 54-75.

\bibitem{Har01} 
 S. Har-Peled, ``A Replacement for Voronoi Diagrams of Near Linear Size,'' {\em Proc. 42nd Annu. IEEE Sympos. Found. Comput. Sci. (FOCS 2001)}, pp.~94--103, 2001.
 
 \bibitem{Har12}
 Sariel Har-Peled and Nirman Kumar, ``Down the Rabbit Hole: Robust Proximity Search and Density Estimation in Sublinear Space.'' {\em FOCS 2012}: 430-439.
\bibitem{Hod05}
D. Hodorkovsky, ``2-Point Site Voronoi Diagrams,'' M.Sc. Thesis, Technion, Haifa, Israel, 2005.

\bibitem{Ind98}
P. Indyk and R. Motwani, ``Approximate Nearest Neighbors: Towards Removing the Curse of Dimensionality,'' {\em Proc. 30th ACM Symp. on Theory of Computing (STOC 1998)}, pp.~604--613, 1998.



\bibitem{Kri05}
H.-P. Kriegel and M. Pfeifle, ``Density-based Clustering of Uncertain Data,"
{\em Proceedings of the 11th ACM SIGKDD International Conference on Knowledge Discovery in Data Mining}, 2005,
pp. 672-677.

\bibitem{LD81}
D.T. Lee and R.L.S. Drysdale, III, ``Generalization of Voronoi Diagrams in the Plane,''
{\em SIAM J. Comput.}, 10(1)(1981), 73-87.

 \bibitem{Liu12}
   A.Y. Liu and D.N. Lam, ``Using Consensus Clustering for Multi-view Anomaly Detection,'' {\em Proc. IEEE CS Security and Privacy Workshops}, 2012, pp. 117-125.

\bibitem{Oka00}
A. Okabe, B. Boots, K. Sugihara, and S.N. Chiu, {\em Spatial Tessellations: Concepts and 
Applications of Voronoi Diagrams}, 2nd Eds., John Wiley \& Sons, 2000.

\bibitem{Pap04}
E. Papadopoulou, ``The Hausdorff Voronoi Diagram of Point Clusters in the Plane,"
{\em Algorithmica}, 40(2004), 63-82.


\bibitem{San98}
J. Sander, M. Ester, H.-P. Kriegel, and X. Xu, ``Density-Based Clustering in Spatial Databases: The Algorithm GDBSCAN and Its Applications,"
{\em Data Mining and Knowledge Discovery}, 2(2)(1998), 169-194.


\bibitem{Vya10}
K. Vyatkina and G. Barequet, ``On 2-Site Voronoi Diagrams under Arithmetic Combinations of
Point-to-Point Distances,'' {\em Proc. 7th International Symp. Voronoi Diagrams in Science and
Engineering}, 2010, pp. 33-41.

\end{thebibliography}
\end{document}